%% file: ms.tex
\documentclass[final]{fundam}

\input{tex/preamble.tex}

\begin{document}

\title{Computing Parameterized Invariants of Parameterized Petri Nets}

\address{welzel@in.tum.de}

\author{Javier Esparza\\
esparza{@}in.tum.de\\
\and Mikhail Raskin\\
raskin{@}in.tum.de\\
\and Christoph Welzel\\
welzel{@}in.tum.de\\
Department of Informatics\\
Technical University of Munich \\
Munich, Germany}

\runninghead{J. Esparza, M. Raskin, C. Welzel}{Parameterized Invariants of
Parameterized Nets}

\maketitle

\begin{abstract}
  A fundamental advantage of Petri net models is the possibility to
  automatically compute useful system invariants from the syntax of the net.
  Classical techniques used for this are place invariants, P-components,
  siphons or traps. Recently, Bozga et al. have presented a novel technique for
  the \emph{parameterized} verification of safety properties of systems with a
  ring or array architecture. They show that the statement \enquote{for every
  instance of the parameterized Petri net, all markings satisfying the linear
  invariants associated to all the P-components, siphons and traps of the
  instance are safe} can be encoded in \acs{WS1S} and checked using tools like
  MONA. However, while the technique certifies that this infinite set of linear
  invariants extracted from P-components, siphons or traps are strong enough to
  prove safety, it does not return an explanation of this fact understandable
  by humans. We present a CEGAR loop that constructs a \emph{finite} set of
  \emph{parameterized} P-components, siphons or traps, whose infinitely many
  instances are strong enough to prove safety. For this we design
  parameterization procedures for different architectures.
\end{abstract}

\begin{keywords}
  parameterized systems, logic, theorem proving, first-order, WS1S
\end{keywords}

\input{tex/introduction.tex}
\input{tex/preliminaries.tex}
\input{tex/parameterized-petri-nets.tex}

\input{tex/checking-1-boundedness.tex}
\input{tex/checking-safety-properties.tex}
\input{tex/rings.tex}
\input{tex/crowds.tex}

\input{tex/generalized-ppn.tex}

\input{tex/experiments.tex}
\input{tex/conclusion.tex}

\bibliographystyle{fundam}
\bibliography{refs}

\appendix
\input{tex/embedding.tex}
\end{document}

%% file: tex/preamble.tex
\usepackage{amssymb}
\usepackage{amsmath}
\usepackage{mathtools}
\usepackage{stmaryrd}
\usepackage{acro}
\usepackage{csquotes}
\usepackage{caption}
\usepackage{tikz}
\usepackage{bm}
\usepackage{multirow}
\usepackage{thm-restate}
\usepackage{todonotes}
\usepackage[final]{listings}

\usetikzlibrary{tikzmark,arrows,shapes,positioning,automata,petri,matrix}

\newcommand{\etal}{\textit{et al.}\ }

\newcommand\set[1]{\left\{ #1 \right\}}
\newcommand\tuple[1]{\left< #1 \right>}
\newcommand\size[1]{\left| #1 \right|}
\newcommand\supportOf[1]{\Lbag #1 \Rbag}

\newcommand{\allfrom}[1]{\mathit{Allfrom}\_{#1}}
\newcommand{\Qequals}[1]{\mathit{Q}_{= #1}}
\newcommand{\Requals}[1]{\mathit{R}_{= #1}}

\newcommand{\shift}{\mathit{shift}}
\newcommand{\unshift}{\mathit{rev}}

\newcommand{\Uv}{{\cal U}}
\newcommand{\Sv}{{\cal S}}
\newcommand{\Xv}{{\cal X}}
\newcommand{\Yv}{{\cal Y}}
\newcommand{\Zv}{{\cal Z}}

\newcommand{\Mv}{{\cal M}}
\newcommand{\Lv}{{\cal L}}
\newcommand{\PN}{{\cal N}}
\newcommand{\places}{{\cal P}}
\newcommand{\relations}{{\cal R}}
\newcommand{\transitions}{\mathit{Tr}}
\newcommand{\markings}{\mathit{Initial}}

\newcommand{\postset}[1]{{#1}^{\bullet}}
\newcommand{\preset}[1]{{}^{\bullet}{#1}}
\newcommand{\trans}[1]{\xrightarrow{#1}}       

\DeclareMathOperator{\logicnext}{succ}

\newcommand{\var}{\mathit{var}}
\newcommand{\val}{\mathit{val}}

\newcommand{\thinking}{\text{think}}
\newcommand{\waiting}{\text{wait}}
\newcommand{\eating}{\text{eat}}
\newcommand{\free}{\text{free}}
\newcommand{\taken}{\text{taken}}
\newcommand{\grabFirst}{\text{GrabFirst}}
\newcommand{\grabSecond}{\text{GrabSecond}}
\newcommand{\release}{\text{Release}}

\newcommand{\idle}{\text{idle}}
\newcommand{\reading}{\text{rd}}
\newcommand{\writing}{\text{wr}}
\newcommand{\notwriting}{\text{not\_wr}}
\newcommand{\startreading}{\text{StartR}}
\newcommand{\startwriting}{\text{StartW}}
\newcommand{\stopreading}{\text{StopR}}
\newcommand{\stopwriting}{\text{StopW}}

\newcommand{\Trap}{\mathit{Trap}}
\newcommand{\TrapTest}{\mathit{TrapTest}}
\newcommand{\Siphon}{\mathit{Siphon}}
\newcommand{\SiphonTest}{\mathit{SiphTest}}
\newcommand{\Mrk}[1]{\mathit{Marked}}

\newcommand{\Safe}{\mathit{Safe}}
\newcommand{\PotReach}{\mathit{PReach}}
\newcommand{\MarkedAt}{\mathit{Marked}}
\newcommand{\EmptyAt}{\mathit{Empty}}
\newcommand{\Safety}{\mathit{SafetyCheck}}
\newcommand{\interp}[1]{\mathbf{#1}}
\newcommand{\Param}{\mathit{ParTrap}}

\newcommand{\oplusn}{\oplus_{\bm{n}}}
\newcommand{\oplusm}{\oplus_{\bm{m}}}

\newcommand{\Balanced}{\mathit{1Bal}}
\newcommand{\BalanceTest}{\mathit{1BBTest}}
\newcommand{\Onebounded}{\textit{1Bnd}}
\newcommand{\WitTrap}{\textit{WTrap}}

\newcommand{\Covered}{\textit{Cover}}

\newcommand{\trapset}{{\cal T}}

\DeclareMathOperator{\drop}{drop}

\newcommand{\origin}{\mathit{origin}}
\newcommand{\target}{\mathit{target}}

\newcommand{\success}{\mathit{succ}}
\newcommand{\fail}{\mathit{fail}}

\newcommand{\stateinitial}{\mathit{initial}}
\newcommand{\stateloop}{\mathit{loop}}
\newcommand{\statelooping}{\mathit{looping}}
\newcommand{\statebreak}{\mathit{break}}
\newcommand{\statecrit}{\mathit{crit}}
\newcommand{\statedone}{\mathit{done}}
\newcommand{\flag}{\mathit{flag}}
\newcommand{\trueval}{\mathit{true}}
\newcommand{\falseval}{\mathit{false}}
\newcommand{\idled}{\mathit{idle}}
\newcommand{\trying}{\mathit{trying}}

\newcommand{\ptr}{\mathit{ptr}}

\newcommand{\states}{Q}
\newcommand{\loopingState}{q_{\circ}}
\newcommand{\loopingRelation}{r_{\circ}}

\newcommand{\ostrich}{\texttt{ostrich}}
\newcommand{\heron}{\texttt{heron}}

\DeclareAcronym{WS1S}{
  short=WS1S,
  long=Weak Second-order Logic With One Successor
}

\DeclareAcronym{FO}{
  short=FO,
  long=First-order Logic
}

%% file: tex/introduction.tex
\section{Introduction}
A fundamental advantage of Petri net system models is the possibility to
automatically extract useful system invariants from the syntax of the net at
low computational cost. Classical techniques used for this purpose are place
invariants, P-components, siphons or traps
\cite{murata1989petri,Reisig13,DE05}. All of them are syntactic objects that
can be computed using linear algebra or boolean logic, and from which
semantic linear invariants can be extracted. For example, from the fact that a
set of places $Q$ is an initially marked trap of the net one extracts the
linear invariant $\sum_{p \in Q} M(Q) \geq 1$, which  is satisfied for every
reachable marking $M$. This information can be used to prove safety properties:
Given a set ${\cal S}$ of safe markings, if every marking satisfying the
invariants extracted from a set of objects is safe, then all reachable markings
are safe.

Classical net invariants have been very successfully used in the verification
of single systems \cite{BensalemBNS09,EsparzaLMMN14,BFHH16}, or as complement
to state-space exploration \cite{WimmelWolf12}. Recently, an extension of this
idea to the \emph{parameterized} verification of safety properties of systems
with a ring or array architecture has been presented in
\cite{BozgaIosifSifakis19a,BozgaEISW20}. The parameterized verification problem
asks whether a system composed of $n$ processes is safe for every $n \geq 2$
\cite{2015Bloem,Esparza16,AbdullaST18}. Bozga \etal show in
\cite{BozgaIosifSifakis19a,BozgaEISW20} that the statement
\begin{quote}
  \enquote{For every instance of the parameterized system, all markings
  satisfying the linear invariants associated to all the P-components, siphons
  and traps of the corresponding Petri net are safe}
\end{quote}
\noindent can be encoded in \acf{WS1S}, or its analogous WS2S for two
successors. This means that the statement holds if{}f its formula encoding is
valid. This problem is decidable, and highly optimized tools exist for it, like
MONA \cite{monaproject,HenriksenJJKPRS95}. The method of \cite{BozgaEISW20} is
not complete (i.e., there are safe systems for which the invariants derived
from P-components, siphons and traps are not strong enough to prove safety),
but it succeeds for a remarkable set of examples. Further, incompleteness is
inherent to every algorithmic method, since safety of parameterized nets is
undecidable {even if processes only manipulate data from a bounded  domain
\cite{AptKozen86,2015Bloem}.

While the technique of \cite{BozgaIosifSifakis19a,BozgaEISW20} is able to prove
interesting properties of numerous systems, it does not yet provide an
explanation of why the property holds. Indeed, when the technique succeeds for
a given parameterized Petri net,  the user only knows that the set of all
invariants deduced from siphons, traps, and P-components together are strong
enough to  prove safety. However, the technique does not return a minimal set
of these invariants. Moreover, since the parameterized Petri net has infinitely
many instances, such a set contains infinitely many invariants. In this paper
we show how to overcome this obstacle. We present a technique that
automatically computes  a \emph{finite} set of \emph{parameterized invariants},
readable by humans.
This is achieved by lifting a CEGAR (counterexample-guided abstraction
refinement) loop, introduced in \cite{EsparzaM00} and further developed in
\cite{EsparzaLMMN14,EsparzaM15,BlondinEH0M20}, to the parameterized case. Each
iteration of the loop of \cite{EsparzaM00,EsparzaLMMN14} first computes a
counterexample, i.e., a marking that violates the desired safety property but
satisfies all invariants computed so far, and then computes a P-component,
siphon, or trap showing that the marking is not reachable. If no counterexample
exists the property is established, and if no P-component, siphon or trap can
be found the method fails. The technique is implemented on top of an
SMT-solver, which receives as input a linear constraint describing the set of
safe markings, and iteratively computes the set of linear invariants derived
from P-components, siphons, and traps.

If we naively lift the CEGAR loop to the parameterized case, the loop never
terminates. Indeed, since the loop computes one new invariant per iteration,
and  infinitely many invariants are needed to prove correctness of all
instances, termination is not possible. So we need a procedure to extract from
one single invariant for one instance a \emph{parameterized} invariant, i.e.,
an infinite set of invariants for all instances, finitely represented as a
WS1S-formula. We present a \emph{semi-automatic} and an \emph{automatic
approach}. In the semi-automatic approach the user guesses the parameterized
invariant, and automatically checks it, using the WS1S-checker. The automatic
approach does not need user interaction, but only works for systems with
symmetric structure. We provide automatic procedures for systems with a ring
topology, and for barrier crowds, a class of systems closely related to
broadcast protocols. We also show how to extend our results to inspection
programs, a class of distributed programs in which an agent can loop through
all other agents, inspecting their local states. In this extension infinite
sets of invariants can no longer be represented by a WS1S-formula and we must
move to a more general logical framework. While the satisfiability problem is
undecidable for this extended framework, we can still prove correctness of some
systems with the help of an automatic theorem prover for first-order logic. Finally, we present
experimental results on a number of systems.

\paragraph{Related work.}
The parameterized verification problem has been extensively studied for systems
whose associated transition systems are well-structured
\cite{GermanSistla92,AbdullaCJT96,FinkelS01} (see e.g. \cite{AbdullaST18} for a
survey). In this case the verification problem reduces to a coverability
problem, for which different algorithms exist
\cite{BlondinFHH16,GeffroyLS18,ReynierS19,FinkelHK20}; the marking equation
(which is roughly equivalent to place invariants) have also been applied
\cite{AthanasiouLW16}. However, the transition systems of parametric rings and
arrays are typically not well-structured.

Parameterized verification of ring and array systems has also been studied in a
number of papers. Three popular techniques are regular model checking (see e.g.
\cite{KestenMalerMarcusPnueliShahar01,AbdullaJNS04,AbdullaHendaDelzannoRezine07}),
abstraction \cite{BaukusBLS00,BaukusLS02}, and automata learning
\cite{ChenHongLinRummer17}. All of them apply symbolic state-space exploration
to try to compute a finite automaton recognizing the set of reachable markings
of all instances, or an abstraction thereof. Our technique avoids any
state-space exploration. Also, symbolic state-space exploration techniques are
not geared towards providing explanations. Indeed, while the set of reachable
markings of all instances is the strongest invariant of the system, it is also
one single monolithic invariant, typically difficult to interpret by human
users. Our CEGAR loop aims at finding a \emph{collection} of invariants, each
of them simple and interpretable.

Many works in the parameterized setting follow the cut-off approach, where one manually
proves a \emph{cut-off} bound $c\geq2$ such that correctness for at most $c$
processes implies correctness for any number of processes (see e.g.
\cite{ClarkeGrumbergBrowne86,EmersonNamjoshi95,EmersonK00,AusserlechnerJK16,JacobsS18},
and \cite{2015Bloem} for a survey).  It then suffices to prove the property for
systems of up to $c$ processes, which can be done using finite-state model
checking techniques. Compared to this technique, ours is fully automatic.

%% file: tex/preliminaries.tex
\section{Preliminaries}
\paragraph{\acs{WS1S}.} Formulas of \ac{WS1S} over first-order variables
$\bm{x}, \bm{y}, \ldots$ and second-order variables $\bm{X}, \bm{Y}, \ldots$
have the following syntax:
\begin{equation*}
  \begin{array}{rcllr}
    t & := & \bm{x} \mid 0 \mid \logicnext(t) & \quad & \text{ (terms)} \\[0.1cm]
    \phi & := & t_1 \leq t_2 \mid \bm{x} \in \bm{X} \mid \phi_1 \wedge \phi_2
    \mid \neg\phi_1 \mid \exists \bm{x} \colon \phi \mid \exists \bm{X}
    \colon \phi &  & \text{(formulas)}
  \end{array}
\end{equation*}
An interpretation assigns elements of $\mathbb{N}_{0} = \set{0, 1, 2, 3,
\ldots}$ to first order variables and \emph{finite} subsets of $\mathbb{N}_{0}$
to second-order variables. Given an interpretation, the semantics that assigns
numbers to terms and truth values to formulas is defined in the usual way.

We extend the syntax with constants $0, 1, 2, 3, \ldots$, and terms of the form
$\bm{x} + c$ with $c \in \mathbb{N}_{0}$. Further, a term $\bm{x}\oplusn 1$ in
a formula $\varphi$ stands for
\begin{equation*}
  (\bm{x} + 1 < \bm{n} \land \varphi[\bm{x} \oplusn 1 \leftarrow \bm{x} + 1])
  \lor (\bm{n} = \bm{x} + 1 \land \varphi[\bm{x} \oplusn 1 \leftarrow 0])
\end{equation*}
where $\varphi[t \leftarrow t']$ denotes the result of substituting $t'$ for
$t$ in $\varphi$. The terms $\bm{x}\oplusn c$ for every $1 \leq c$ are defined
similarly. We let $\varphi(\bm{x}_{1}, \ldots, \bm{x}_{\ell}, \bm{X}_{1},
\ldots, \bm{X}_{k})$ denote that $\varphi$ uses at most $\bm{x}_{1}, \ldots,
\bm{x}_{\ell}$ and $\bm{X}_{1}, \ldots, \bm{X}_{k}$ as free first-order resp.
second-order variables. Finally, we also make liberal use of the following
macros:
\begin{center}
  \begin{tabular}{lcl}
    $\bm{X} = \emptyset$ & \multirow{8}{*}{ \ \ stands for \ \ }
      & $\forall \bm{x} \colon \neg (\bm{x}\in\bm{X})$ \\
    $\bm{X} = \set{\bm{x}}$ &
      &  $\bm{x} \in \bm{X} \land \forall \bm{y} \colon \bm{y} \in \bm{X} \rightarrow \bm{y} = \bm{x}$ \\
    $\bm{X} = [\bm{n}]$ &
      & $\forall \bm{x} \colon \bm{x} \in \bm{X} \leftrightarrow \bm{x} < \bm{n}$ \\
    $\bm{X} \cap \bm{Y} = \emptyset$ &
      & $\forall \bm{x} \colon \neg (\bm{x} \in \bm{X} \wedge \bm{x} \in \bm{Y})$ \\
    $\size{\bm{X}} = 1$ &
      & $\exists \bm{x} \colon \bm{X} = \set{\bm{x}}$ \\
    $\size{\bm{X}} \leq 1$ &
      & $\bm{X} = \emptyset \vee \size{\bm{X}} = 1$ \\
    $\bm{X} = \overline{\bm{Y}}$ &
      & $\forall \bm{x} \colon \bm{x} \in \bm{X} \leftrightarrow (\bm{x}<\bm{n}\land\neg (\bm{x} \in \bm{Y}))$ \\
    $\bm{Y} = \bm{X} \oplusn 1$ &
      & $\forall \bm{x} \colon \bm{x} \oplusn 1 \in \bm{Y} \leftrightarrow \bm{x} \in \bm{X}$
  \end{tabular}
\end{center}

\paragraph{Petri Nets.}
We use a presentation of Petri nets equivalent to but slightly different from
the standard one. A \emph{net} is a pair $\tuple{P, T}$ where $P$ is a
nonempty, finite set of \emph{places} and $T \subseteq 2^P \times 2^P$ is a set
of \emph{transitions}. Given a transition $t = \tuple{P_{1}, P_{2}}$, we call $P_{1}$
the \emph{preset} and \emph{postset} of $t$, respectively. We also denote
$P_{1}$ by $\preset{t}$ and $P_{2}$ by $\postset{t}$. Given a place $p$, we
denote by $\preset{p}$ and $\postset{p}$ the sets of transitions $\tuple{P_1,
P_2}$ such that $p \in P_2$ and $p \in P_1$, respectively. Given a set $X$ of
places or transitions, we let $\preset{X} \coloneqq \bigcup_{x \in X}
\preset{x}$ and $\postset{X} \coloneqq \bigcup_{x \in X} \postset{x}$.

A \emph{marking} of $N=\tuple{P, T}$ is a function $M \colon P \rightarrow
\mathbb{N}$. A Petri net is a pair $\tuple{N, M}$, where $N$ is a net and $M$
is the \emph{initial marking} of $N$. A transition $t=\tuple{P_{1}, P_{2}}$ is
enabled at a marking $M$ if $M(p) \geq 1$ for every $p \in P_{1}$. If $t$ is
enabled at $M$ then it can \emph{fire}, leading to the marking $M'$ given by
$M'(p) = M(p) + 1$ for every $p \in P_{2} \setminus P_{1}$, $M'(p) = M(p) - 1$
for every $p \in P_{1} \setminus P_{2}$, and $M'(p) = M(p)$ otherwise. We write
$M \trans{t} M'$, and $M \trans{\sigma} M'$ for a finite sequence $\sigma =
t_{1} t_{2} \ldots t_n$ if there are markings $M_{1}, \ldots, M_{n}$ such that
$M \trans{t_{1}} M_{1} \trans{t_{2}} \cdots M_{n-1} \trans{t_n} M'$. $M'$ is
reachable from $M$ if $M \trans{\sigma} M'$ for some sequence $\sigma$.

A marking $M$ is \emph{1-bounded} if $M(p) \leq 1$ for every place $p$. A Petri
net is \emph{1-bounded} if every marking reachable from the initial marking is
1-bounded. A 1-bounded marking $M$ of a Petri net is also defined by the set of
marked places; i.e., $\supportOf{M} = \set{p\in P\colon M(p) = 1}$.

%% file: tex/parameterized-petri-nets.tex
\section{Parameterized Petri Nets}
\label{sec:ppnets}

Intuitively, a parameterized net is a collection $\set{ N_n}_{n \geq 1}$ of
nets. The places of $N_n$ are the result of replicating a set $\places$ of
place names $n$ times. For example, if $\places = \{p, q\}$, then the set of
places of $N_n$ is $\{p(0), \ldots, p(n-1), q(0), \ldots, q(n-1)\}$. Crucially,
the transitions of all the nets in the collection are described by a single
logical formula of \ac{WS1S}. Intuitively, the models of the formula are
triples $\tuple{n, P_1, P_2}$, where $P_1$ and $P_2$ are sets of places of
$N_n$, indicating that $N_n$ has a transition with $P_1$ and $P_2$ as input and
output places, respectively.

\begin{definition}[Parameterized nets]
  A parameterized net is a pair $\PN = \tuple{\places, \transitions}$, where
  $\places$ is a finite set of place names and $\transitions(\bm{n}, \Xv, \Yv)$
  is a \ac{WS1S}-formula over one first-order variable $\bm{n}$ which
  represents the size of the instance, and two tuples $\Xv$ and $\Yv$
  of second-order variables containing one variable for each place name of $\places$; i.e., for a fixed
  enumeration $p_{1}, \ldots, p_{k}$ of the elements of $\places$ we have $\Xv =
  \tuple{\Xv_{p_{i}}}_{i = 1}^{k}$ and $\Yv = \tuple{\Yv_{p_{i}}}_{i = 1}^{k}$.
  We call such tuples of variables \emph{placeset variables}.
\end{definition}
Let $[n] = \{0, \ldots, n-1\}$. A parameterized net $\PN$ induces a net $\PN(n)
= \tuple{P_n, T_n}$ for every $n \geq 1$, where $P_n = \places \times [n]$
(i.e., $P_n$ consists of $n$ copies of $\places$), and $T_n$ contains a
transition $\tuple{P_{1}, P_{2}}$ for every pair $P_{1}, P_{2} \subseteq P_{n}$
of sets of places such that \enquote{$\transitions(n, P_{1}, P_{2})$ holds}.
More formally, this means that $\mu\models \transitions$ for the interpretation
$\mu$ given by $\mu(\bm{n}) = n$, $\mu(\Xv_{p}) = \set{i\in[n]: \tuple{p, i}\in
P_{1}}$, and $\mu(\Yv_{p}) = \set{i\in[n]: \tuple{p, i}\in P_{2}}$ for all $p
\in \places$. Therefore, the intended meaning of $\transitions(n, \Xv, \Yv)$ is
\enquote{the pair $\tuple{\Xv,\Yv}$ of placesets is (the preset and postset of)
a transition of the net $\PN(n)$}. We say that $\PN(n)$ is an \emph{instance} of
$\PN$.

In the following we use $\tuple{p, i}$ and $p(i)$ as equivalent notations
for the elements of $P_n = \places \times [n]$.

\begin{example}
\label{ex:diningphil}
  We consider a version of the dining philosophers. Philosophers and forks are
  numbered $0$, $1$, \ldots, $n-1$. For every $i > 0$ the $i$-th philosopher
  first grabs the $i$-th fork, and then the $(i \oplus_{n} 1)$-th fork, where
  $\oplus_{n}$ denotes addition modulo $n$. Philosopher $0$ proceeds the other
  way round: she first grabs fork $1$, and then fork $0$. After eating, a
  philosopher returns both forks in one single atomic step. We formalize this in the
  following parameterized net $\PN=\tuple{\places, \transitions}$:
  \begin{itemize}
    \item $\places = \set{\thinking, \waiting, \eating, \free, \taken}$.
      Intuitively, $\set{\thinking(i), \waiting(i), \eating(i)}$ are the states
      of the $i$-th philosopher, and $\set{\free(i), \taken(i)}$ the states of
      the $i$-th fork.
    \item $\transitions(\bm{n}, \Xv, \Yv) = \grabFirst \vee \grabSecond \vee
      \release$. The formulas for $\grabFirst$, $\grabSecond$, and $\release$
      are shown in Equation~\ref{table:phil}.
  \end{itemize}

\begin{table}[t]
\caption{Transitions of the dining philosophers.}
\label{table:phil}
      $$\begin{array}{lcl}
        \grabFirst & \coloneqq &
        \begin{aligned}
          &
         \left(
         \begin{aligned}
           \exists \bm {x} ~.~ 1 \leq \bm{x} < \bm{n}
              &\land (\Xv_\thinking = \Xv_\free = \Yv_{\waiting} = \Yv_{\taken}
                      = \set{\bm{x}} )\\
              &\land (\Xv_\waiting = \Xv_\eating = \Xv_\taken = \emptyset)\\
              &\land (\Yv_\thinking = \Yv_\eating = \Yv_\free = \emptyset )
            \end{aligned}
            \right)
            \\
          & \qquad \lor & \\
          & \left(
            \begin{aligned}
              & (\Xv_{\thinking} = \Yv_{\waiting} = \set{0})
              \land (\Xv_{\free} = \Yv_{\taken} = \set{1})\\
              \land \; & (\Xv_\waiting = \Xv_\eating = \Xv_\taken = \emptyset) \\
              \land \; & (\Yv_\thinking = \Yv_\eating = \Yv_\free = \emptyset)
            \end{aligned}
         \right) 
        \end{aligned} \\[2.2cm]
        \grabSecond & \coloneqq &
        \begin{aligned}
          &  \left(
            \begin{aligned}
              \exists \bm{x} ~.~ 1 \leq \bm{x} < \bm{n} \;\;
              \land \;\; &(\Xv_\waiting = \Yv_\eating = \set{\bm{x}})\\
              \land \;\; &(\Xv_\free = \Yv_\taken = \set{\bm{x} \oplusn 1})\\
              \land \;\; &(\Xv_\thinking = \Xv_\eating = \Xv_\taken = \emptyset) \\
              \land \;\; &(\Yv_\thinking =  \Yv_\waiting = \Yv_\free = \emptyset)\\
            \end{aligned}
            \right) \\
          & \qquad \lor & \\
          & \left(
            \begin{aligned}
              & (\Xv_{\thinking} = \Xv_{\free} = \Yv_{\taken} = \Yv_{\waiting}
                 = \set{0})\\
              \land \;\; &(\Xv_\waiting = \Xv_\eating = \Xv_\taken = \emptyset) \\
              \land \;\; &(\Yv_\thinking = \Yv_\eating = \Yv_\free = \emptyset)
            \end{aligned}
         \right) 
        \end{aligned} \\[2.3cm]
        \release & \coloneqq  &
        \exists \bm{x} ~.~ \bm{x} < \bm{n} \;\; 
        \begin{aligned}[t]
          \land \;\; &(\Xv_\eating = \Yv_\thinking = \set{x} \wedge \Xv_\taken = \Yv_\free
            = \set{x, x\oplus 1})\\
          \land \;\; &(\Xv_\thinking = \Xv_\waiting = \Xv_\free = \emptyset) \\
          \land \;\; &(\Yv_\waiting =  \Yv_\eating =  \Yv_\taken = \emptyset)\\
        \end{aligned}
      \end{array}$$
\end{table}     

Intuitively, the preset of $\grabFirst$ is a philosopher in state $\thinking$
  and her left (resp. right fork for philosopher $0$) in state free; the postset
  puts the philosopher in state $\waiting$ and the fork in state $\taken$.
  The instance $\PN(3)$ is  shown in Figure~\ref{fig:dining-philosophers}.
  \label{ex:dining-philosophers}
\end{example}

\begin{figure}[ht]
  \caption{$\PN(3)$ for Example \ref{ex:dining-philosophers}. Places which are
  colored green are initially marked w.r.t. $\markings(\Xv)$ from Example
  \ref{ex:dining-philosophers-initial-marking}. Note the repeating structure
  for philosophers $1$ and $2$ while philosopher $0$ grabs her forks in the
  opposite order. We abbreviate $\thinking(i)$ to $\text{th}(i)$, and similarly
  with the other states.}
  \label{fig:dining-philosophers}
\renewcommand{\thinking}{\text{th}}
\renewcommand{\waiting}{\text{wa}}
\renewcommand{\eating}{\text{ea}}
\renewcommand{\free}{\text{fr}}
\renewcommand{\taken}{\text{ta}}
  \begin{center}
    \resizebox{0.65\textwidth}{!}{%
      \begin{tikzpicture}[every node/.style={scale=0.7}]
        \node [place, minimum width=0.8cm, fill=green] (free1) at (000:2.2cm)
        {$\free(0)$};
        \node [place, minimum width=0.8cm] (taken1) at (340:2.2cm)
        {$\taken(0)$};

        \node [place, minimum width=0.8cm, fill=green] (thinking1) at (310:2.2cm)
        {$\thinking(0)$};
        \node [place, minimum width=0.8cm] (waiting1) at (290:2.2cm)
        {$\waiting(0)$};
        \node [place, minimum width=0.8cm] (eating1) at (270:2.2cm)
        {$\eating(0)$};

        \node [place, minimum width=0.8cm, fill=green] (free2) at (240:2.2cm)
        {$\free(1)$};
        \node [place, minimum width=0.8cm] (taken2) at (220:2.2cm)
        {$\taken(1)$};

        \node [place, minimum width=0.8cm, fill=green] (thinking2) at (190:2.2cm)
        {$\thinking(1)$};
        \node [place, minimum width=0.8cm] (waiting2) at (170:2.2cm)
        {$\waiting(1)$};
        \node [place, minimum width=0.8cm] (eating2) at (150:2.2cm)
        {$\eating(1)$};

        \node [place, minimum width=0.8cm, fill=green] (free3) at (120:2.2cm)
        {$\free(2)$};
        \node [place, minimum width=0.8cm] (taken3) at (100:2.2cm)
        {$\taken(2)$};

        \node [place, minimum width=0.8cm, fill=green] (thinking3) at (070:2.2cm)
        {$\thinking(2)$};
        \node [place, minimum width=0.8cm] (waiting3) at (050:2.2cm)
        {$\waiting(2)$};
        \node [place, minimum width=0.8cm] (eating3) at (030:2.2cm)
        {$\eating(2)$};

        \node [transition] (grab1l) at (255:1cm) {$g^{1}_{0}$}
          edge [pre] (thinking1)
          edge [pre] (free2)
          edge [post] (waiting1)
          edge [post] (taken2);

        \node [transition] (grab1r) at (325:1cm) {$g^{2}_{0}$}
          edge [pre] (waiting1)
          edge [pre] (free1)
          edge [post] (eating1)
          edge [post] (taken1);

        \node [transition] (grab2r) at (205:1cm) {$g^{1}_{1}$}
          edge [pre] (thinking2)
          edge [pre] (free2)
          edge [post] (waiting2)
          edge [post] (taken2);

        \node [transition] (grab2l) at (135:1cm) {$g^{2}_{1}$}
          edge [pre] (waiting2)
          edge [pre] (free3)
          edge [post] (eating2)
          edge [post] (taken3);

        \node [transition] (grab3l) at (085:1cm) {$g^{1}_{2}$}
          edge [pre] (thinking3)
          edge [pre] (free3)
          edge [post] (waiting3)
          edge [post] (taken3);

        \node [transition] (grab3r) at (015:1cm) {$g^{1}_{2}$}
          edge [pre] (waiting3)
          edge [pre] (free1)
          edge [post] (eating3)
          edge [post] (taken1);

        \node [transition] (release1) at (290:3.2cm) {$r_{0}$}
          edge [pre] (eating1)
          edge [pre, bend right] (taken1)
          edge [pre, bend left=80] (taken2)
          edge [post] (thinking1)
          edge [post, bend right=80] (free1)
          edge [post, bend left] (free2);

        \node [transition] (release2) at (170:3.2cm) {$r_{1}$}
          edge [pre] (eating2)
          edge [pre, bend right] (taken2)
          edge [pre, bend left=80] (taken3)
          edge [post] (thinking2)
          edge [post, bend right=80] (free2)
          edge [post, bend left] (free3);

        \node [transition] (release3) at (050:3.2cm) {$r_{2}$}
          edge [pre] (eating3)
          edge [pre, bend right] (taken3)
          edge [pre, bend left=80] (taken1)
          edge [post] (thinking3)
          edge [post, bend right=80] (free3)
          edge [post, bend left] (free1);
      \end{tikzpicture}
   }
  \end{center}
\renewcommand{\thinking}{\text{think}}
\renewcommand{\waiting}{\text{wait}}
\renewcommand{\eating}{\text{eat}}
\renewcommand{\free}{\text{free}}
\renewcommand{\taken}{\text{taken}}
\end{figure}
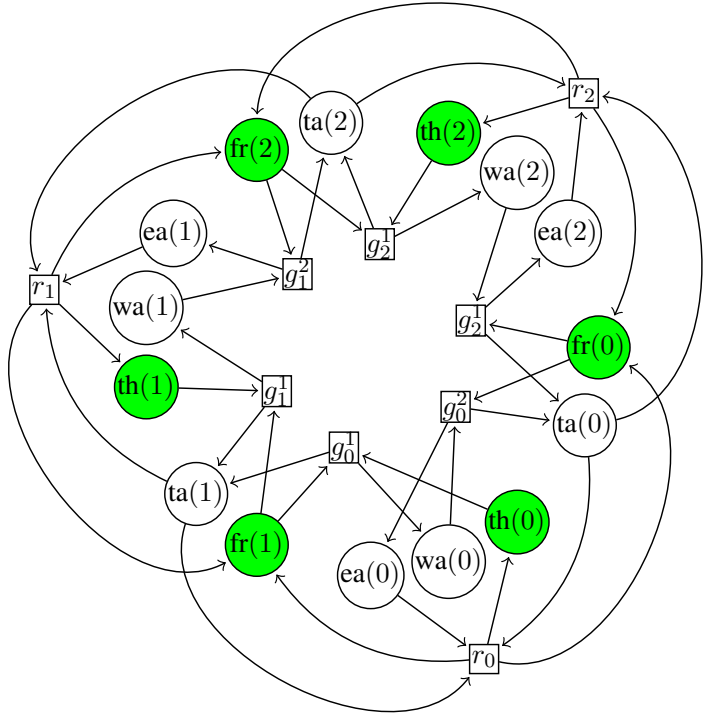

\noindent Parameterized Petri nets are parameterized nets with a
\ac{WS1S}-formula defining its initial markings:

\begin{definition}[Parameterized Petri nets]
  A parameterized Petri net is a pair $\tuple{\PN, \markings}$, where $\PN$ is
  a parameterized net, and $\markings(\bm{n}, \Mv)$ is a \ac{WS1S}-formula over
  a first-order variable $\bm{n}$ and a placeset variable $\Mv$.
\end{definition}

A parameterized Petri net defines an infinite family of Petri nets. Loosely
speaking, a Petri net $\tuple{N,M}$ belongs to the family if $N$ is an instance
of $\PN$, i.e., $N = \PN(n)$ for some $n \geq 1$, and $M$ is a 1-bounded
marking of $N$ satisfying $\markings(n, \Mv)$. For example, if $\places =
\set{p_{1}, p_{2}}$, $n=3$ and $\markings(\set{0,1}, \set{0,2})$ holds, then
the family contains a Petri net $\tuple{\PN(3), M_{3}}$ such that $M_{3}$ is a
1-bounded marking with $\supportOf{M_{3}} = \set{p_{1}(0), p_{1}(1), p_{2}(0),
p_{2}(2)}$.

\begin{example}
  \label{ex:dining-philosophers-initial-marking}
  The family of initial markings in which all philosophers think and all forks
  are free is modeled by:
  \begin{equation*}
    \markings(\bm{n}, \Mv) \coloneqq (\Mv_\thinking = \Mv_\free = [\bm{n}])
      \land (\Mv_\waiting =  \Mv_\eating =  \Mv_\taken = \emptyset).
  \end{equation*}
\end{example}
\begin{example}
  Let us now model a simple version of the readers/writers system. A process
  can be idle, reading, or writing. An idle process can start to read if no
  other process is writing, and it can start to write if every other process is
  idle. We obtain the parameterized net $\PN=\tuple{\places, \transitions}$,
  where
  \begin{itemize}
    \item $\places = \set{\idle, \reading, \writing, \notwriting}$.
    \item $\transitions(\bm{n}, \Xv, \Yv) = \startreading \lor \stopreading
      \lor \startwriting \lor \stopwriting$. We give the formulas
      $\startreading$ and $\startwriting$, the other two being simpler.
  \end{itemize}
  \begin{equation*}
    \startreading \coloneqq
    \exists \bm{x} ~.~
    \left(
    \begin{aligned}
      & \phantom{\wedge} (\Xv_\idle = \set{\bm{x}} \wedge \Xv_\notwriting =
        \overline{\Xv_\idle} \wedge (\Xv_\reading = \Xv_\writing = \emptyset)\\
      & \wedge \Yv_\reading = \set{\bm{x}} \wedge \Yv_\notwriting =
        \overline{\Xv_\idle} \wedge (\Yv_\idle = \Yv_\writing = \emptyset)
    \end{aligned}
    \right)
  \end{equation*}
  \begin{equation*}
    \startwriting \coloneqq
    \exists \bm{x} ~.~
    \left(
    \begin{aligned}
      &\phantom{\wedge} \Xv_\idle = [\bm{n}] \wedge \Xv_\notwriting =
      \set{\bm{x}} \wedge (\Xv_\reading = \Xv_\writing = \emptyset)
      \\
      & \wedge \Yv_\idle = [\bm{n}] \setminus \set{\bm{x}} \wedge \Yv_\writing
      = \set{\bm{x}} \wedge (\Yv_\reading  = \Yv_\notwriting = \emptyset)
    \end{aligned}
    \right)
  \end{equation*}
  So the preset of a $\startreading$ transition is $\set{\idle(i),
  \notwriting(0), \ldots, \notwriting(n-1)}$ for some $i$, and the postset is
  $\set{\reading(i), \notwriting(0),  \ldots, \notwriting(n-1)}$.
 The initial markings in which every process is initially idle are modeled by:
  \begin{equation*}
    \markings(\bm{n}, \Xv) \coloneqq \Xv_\idle = [\bm{n}] \wedge
    \Xv_\notwriting = [\bm{n}] \wedge (\Xv_\reading = \Xv_\writing = \emptyset)
  \end{equation*}
  Observe that in the dining philosophers transitions have presets and postsets
  of size $3$, independently of the number of philosophers. On the contrary, in
  the readers and writers problems the transitions of $\PN(n)$ have presets and
  postsets of size $n$. Intuitively, our formalism allows to model transitions
  involving all processes or, for example, all even processes. Observe also
  that in both cases the formula $\markings$ has exactly one model for every
  $n \geq 1$, but this is not required.
\end{example}

\paragraph{Proving deadlock-freedom for the dining philosophers.}
Let us now give a taste of what our paper achieves for
Example~\ref{ex:dining-philosophers}. It is well known that this version of the
dining philosophers is deadlock-free. However, finding a proof based on
parameterized invariants of the systems is not so easy. Using the
semi-automatic version of the approach we present, we can find the five
invariants shown below, and automatically prove that they imply deadlock-freedom.
The fully automatic analysis of this example gives ten
properties of the system which collectively induce deadlock-freedom.
 
The first two invariants express that at every reachable marking $M$,
and for every $0 \leq i \leq n-1$ , the $i$-th philosopher is either thinking,
waiting, or eating,  and the $i$-th fork is either free or taken:
\begin{align}
    \label{eq:philosopher-1-bounded}
    M(\thinking(i)) + M(\waiting(i)) + M(\eating(i)) & = 1 \\
    \label{eq:fork-1-bounded}
    M(\free(i)) + M(\taken(i)) & = 1.
 \end{align}
The last three invariants provide the key insights; the last one holds for
every $1 \leq  i \leq n-2$:
\begin{align}
  \label{eq:0-1-fork}
  M(\waiting(0))+ M(\eating(0))+M(\free(1))+M(\waiting(1))+ M(\eating(1)) &= 1 \\
  \label{eq:0-n-1-fork}
  M(\eating(0))+M(\free(0))+M(\eating(n-1))&= 1\\
  \label{eq:p-p+1-fork}
  M(\eating(i))+M(\eating(i+1))+M(\free(i+1))+M(\waiting(i+1)) &= 1
\end{align}
Let us sketch why (1)-(5) imply deadlock freedom. Let $P_i$ denote the $i$-th
philosopher and $F_i$ the $i$-th fork. If $P_0$ is eating, then $F_0$ and $F_1$
are taken by  (\ref{eq:philosopher-1-bounded})-(\ref{eq:0-n-1-fork}), and there
is no deadlock because $P_0$ can return them. The same holds if $P_1$ is eating
by (\ref{eq:philosopher-1-bounded})-(\ref{eq:0-1-fork}) and
(\ref{eq:p-p+1-fork}), or if any of $P_2, \ldots, P_{n-1}$ is eating by
(\ref{eq:philosopher-1-bounded})-(\ref{eq:fork-1-bounded}) and
(\ref{eq:p-p+1-fork}). If no philosopher eats, then by
(\ref{eq:philosopher-1-bounded})-(\ref{eq:0-1-fork}) and (\ref{eq:p-p+1-fork})
either $P_{i+1}$ is thinking and $F_{i+1}$ is free for some $i \in \{1, \ldots,
n-2\}$, or  $P_{i+1}$ is waiting for every $i \in \{1, \ldots, n-2\}$. In the
first case  $P_{i+1}$ can grab $F_{i+1}$. In the second case $P_{n-1}$ is
waiting, and since $F_0$ is free by
(\ref{eq:philosopher-1-bounded})-(\ref{eq:fork-1-bounded}) and
(\ref{eq:0-n-1-fork}), it can grab $F_0$.



%% file: tex/checking-1-boundedness.tex
\section{Checking 1-boundedness}
Our techniques work for parameterized Petri nets whose instances are 1-bounded.
We present a technique that automatically checks 1-boundedness of all our
examples. We say that a set of places $Q$ of a Petri net $\tuple{N, M}$, where
$N=\tuple{P,T}$, is
\begin{itemize}
\item \emph{1-balanced} if for every transition $\tuple{P_{1}, P_{2}} \in T$
  either $\size{P_{1} \cap Q} = 1 = \size{P_2 \cap Q}$, or $\size{P_{1} \cap Q}
    = 0 = \size{P_2 \cap Q}$, or $\size{P_{1} \cap Q} \geq{} 2$.
  \item \emph{1-bounded} at $M$  if $M(Q) \leq 1$.
\end{itemize}

The following proposition is an immediate consequence of the definition:
\begin{proposition}
  If $Q$ is a 1-balanced and 1-bounded set of places of $\tuple{N, M}$, then
  $M'(Q)=M(Q)$ holds for every reachable marking $M'$.
\end{proposition}

We abbreviate \enquote{1-bounded and 1-balanced set} to 1BB-set, and say that
$N$ is \emph{covered} by 1BB-sets if every place belongs to some 1BB-set at
initial marking $M$. By the proposition above, if $N$ is covered by 1BB-sets at
$M$, then $M'(p)\leq{1}$ holds for every reachable marking $M'$ and every place
$p$, and so $N$ is 1-bounded.

%
%

Given a parameterized Petri net $(\PN, \markings)$, we can check if all
instances are covered by 1BB-sets with the following formula:
\begin{center}
  \begin{tabular}{rcll}
    $\Balanced(\bm{n}, \Xv)$
    & $\coloneqq$
    & $\forall \Yv, \Zv  \colon  \transitions(\bm{n}, \Yv,\Zv) \rightarrow$
    & $(\size{\Xv \cap \Yv}= 0 =\size{\Xv \cap \Zv} ) \vee $ \\
    & & &   $(\size{\Xv \cap \Yv}= 1 =\size{\Xv \cap \Zv}) \vee$ \\
    & & &   $(\size{\Xv \cap \Yv} > 1)$ \\
    $\Onebounded(\bm{n}, \Xv,\Mv)$
    & $\coloneqq$
    & $\size{\Xv \cap \Mv} \leq 1$ \\
    $\Covered$
    & $\coloneqq$
    & $\forall \bm{n},\forall \Mv \colon \markings(\bm{n}, \Mv) \rightarrow$
    & $(\bigwedge_{p \in \places}\forall \bm{x} \colon \exists \Xv \colon
    \bm{x} \in \Xv_{p} \wedge$\\
    & & & $\Balanced(\bm{n}, \Xv) \wedge$\\
    & & & $\Onebounded(\bm{n}, \Xv, \Mv))$
  \end{tabular}
\end{center}

Observe that if $Q$ is a 1BB-set then at every reachable marking exactly one of
the places of $Q$ is marked, with exactly one token. The sets of places
corresponding to a philosopher, a fork, a reader, or a writer are 1BB-sets.
Unsurprisingly, all our parameterized Petri net models are covered by 1BB-sets.
Checking the formula $\Covered$ above gives us an automatic proof that all the
Petri nets we consider are 1-bounded.

%% file: tex/checking-safety-properties.tex
\section{Checking safety properties}
\label{sec:safety}
Let $\tuple{\PN, \markings}$ be a parameterized Petri net, and let
$\Safe(\bm{n}, \Mv)$ be a WS1S-formula describing a set of \enquote{safe}
markings of the instances of $\PN$ (for example, \enquote{safe} could mean
deadlock-free). It is easy to prove (using simulations of Turing machines by
Petri nets like those of \cite{Esparza96}) that the existence of some unsafe
reachable marking in some instance of a given parameterized Petri net
$\tuple{\PN, \markings}$ is undecidable. In
\cite{BozgaEISW20,BozgaIosifSifakis19a} there is a semi-algorithm for the
problem that derives from $\tuple{\PN, \markings}$ a formula $\PotReach(\bm{n},
\Mv)$ describing a superset of the set of reachable markings of all instances,
and checks that the formula
\begin{equation*}
  \Safety \coloneqq \forall \bm{n} \forall \Mv \colon \PotReach(\bm{n}, \Mv)
  \rightarrow \Safe(\bm{n}, \Mv)
\end{equation*}
holds. We recall the main construction of
\cite{BozgaEISW20,BozgaIosifSifakis19a}, adapted and expanded.

\paragraph{1BB-sets again.}
Recall that if a marking $M'$ of some instance $\tuple{N, M}$ of a net
$\tuple{\PN, \markings}$ is reachable from $M$, then $M'(Q)\leq 1$ holds for
every 1BB-set of places $Q$ of $\tuple{N, M}$. So this latter property can be
interpreted as a test for potential reachability: Only markings that pass the
test can be reachable. We introduce a formula $\BalanceTest(\bm{n}, \Mv', \Mv)$
expressing that $\Mv'$ passes the test with respect to $\Mv$ (i.e., $\Mv'$
might be reachable from $\Mv$).

\begin{equation*}
  \BalanceTest(\bm{n}, \Mv', \Mv) \coloneqq
  \forall \Xv \colon \left(
  \begin{aligned}
    &\Balanced(\bm{n}, \Xv)\\
    \wedge &\Onebounded(\bm{n}, \Xv, \Mv)
  \end{aligned}
  \right) \rightarrow \Onebounded(\bm{n}, \Xv, \Mv')
\end{equation*}

\paragraph{Siphons and traps.}
Let $\tuple{N,M}$ be a Petri net with $N=\tuple{P,T}$ and let $Q \subseteq P$
be a set of places. $Q$ is a \emph{trap} of $N$ if $\preset{Q} \subseteq
\postset{Q}$, and a \emph{siphon} of $N$ if $\postset{Q} \subseteq \preset{Q}$.
\begin{itemize}
  \item If $Q$ is a siphon and $M(Q)=0$, then $M'(Q)=0$ for all
    markings $M'$ reachable from $M$.
  \item If $Q$ is a trap and $M(Q) \geq 1$, then $M'(Q) \geq 1$ for all
    markings $M'$ reachable from $M$.
\end{itemize}

If $M'$ is reachable from $M$ then it satisfies the following property: $M'(Q)
\geq 1$ for every trap $Q$ such that $M(Q) \geq 1$. A marking satisfying this
property \emph{passes the trap test} for $\tuple{N, M}$. We construct a formula
$\TrapTest(\bm{n}, \Mv)$ expressing that $\Mv$ passes the trap test for some
instance of a parameterized Petri net. We first introduce a formula expressing
that a set $\Xv$ of places is a trap.
\begin{equation*}
  \Trap(\bm{n}, \Xv)
  \coloneqq
  \forall \Yv, \Zv \colon
  (\transitions(\bm{n}, \Yv, \Zv ) \wedge \Xv \cap \Yv \neq \emptyset)
  \rightarrow \Xv \cap \Zv \neq \emptyset
\end{equation*}
Now we have:
\begin{center}
  \begin{tabular}{rcl}
    $\MarkedAt(\bm{n}, \Xv, \Mv)$
    & $\coloneqq$
    & $\Xv \cap \Mv \neq \emptyset$ \\
    $\TrapTest(\bm{n}, \Mv', \Mv)$
    & $\coloneqq$
    & $\forall \Xv \colon \left(
    \begin{aligned}
      &\Trap(\bm{n}, \Xv)\\
      \wedge &\MarkedAt(\bm{n}, \Xv,\Mv)
    \end{aligned}\right)
     \rightarrow \MarkedAt(\bm{n}, \Xv,\Mv')$
  \end{tabular}
\end{center}
Similarly we obtain a formula for a siphon test:
\begin{center}
  \begin{tabular}{rcl}
    $\EmptyAt(\bm{n}, \Xv, \Mv)$
    & $\coloneqq$
    & $\Xv \cap \Mv  = \emptyset$ \\
    $\SiphonTest(\bm{n}, \Mv', \Mv)$
    & $\coloneqq$
    & $\forall \Xv \colon \left(
    \begin{aligned}
      &\Siphon(\bm{n}, \Xv)\\
      \wedge &\EmptyAt(\bm{n}, \Xv,\Mv)
    \end{aligned}\right)
       \rightarrow \EmptyAt(\bm{n}, \Xv,\Mv')$
  \end{tabular}
\end{center}
We can now give the formula $\PotReach$:
\begin{center}
  \begin{tabular}{rcl}
    $\PotReach(\bm{n}, \Mv', \Mv)$
    & $\coloneqq$
      & $\left(\begin{aligned}
        &\BalanceTest(\bm{n}, \Mv', \Mv)\\
        \wedge &\TrapTest(\bm{n}, \Mv', \Mv)\\
        \wedge &\SiphonTest(\bm{n}, \Mv', \Mv)\end{aligned}\right)$
    \\
    $\PotReach(\bm{n}, \Mv')$
    & $\coloneqq$
    & $\exists \Mv \colon \markings(\bm{n}, \Mv) \wedge \PotReach(\bm{n}, \Mv',
    \Mv)$
  \end{tabular}
\end{center}

%
%
%

\subsection{Automatic computation of parameterized invariants}
In \cite{BozgaEISW20} it was shown that many safety properties of parameterized
Petri nets can be proved to hold \emph{for all instances}  by checking validity
of the corresponding $\PotReach$ formula. However, the technique does not
return a set of invariants strong enough to prove the property. In this section
we show how to overcome this problem. We design a CEGAR loop which, when
successful, yields a finite set of \emph{parameterized} invariants that imply
the safety property being considered.

We proceed as follows. In the first part of the section, we describe a CEGAR
loop for the non-parameterized case. The input to the procedure is a
parameterized Petri net $\tuple{\PN, \markings}$ and a number $n$ such that all
reachable markings of all instances $\PN(1), \ldots, \PN(n)$ are safe. The
output is a set of invariants of $\PN(1), \ldots, \PN(n)$, derived from
balanced sets, siphons, and traps, which are strong enough to prove safety.
Since the set of all 1BB-sets, siphons, and traps of these instances
is finite, the procedure is guaranteed to terminate even if it computes one
invariant at a time. Then we modify the loop by inserting an additional
\emph{parameterization procedure} that exploits the regularity of $\tuple{\PN,
\markings}$. The procedure transforms a 1BB-set (siphon, trap) of a
particular instance, say $\PN(4)$, into a possibly infinite set of 1BB-sets
(siphons, traps) of all instances, encoded as the set of models of a
WS1S-formula. This formula is a finite representation of the infinite set.

For the sake of brevity, in the rest of the section we describe a CEGAR loop
that only constructs traps. This allows us to avoid numerous repetitions of the
phrase ``1BB-sets, siphons, and traps''.  Since the structure of the
loop is completely generic, this is purely a presentation issue without  loss
of generality\footnote{The CEGAR loop for the non-parametric case could be
formulated in SAT and solved using a SAT-solver. However, we formulate it in
WS1S, since this allows us to give a uniform description of the non-parametric
and the parametric cases.}.

\subsubsection{A CEGAR loop for the non-parameterized case.}
\label{sec:non-param-cegar}
We need some preliminaries. Let $\PN=\tuple{\places, \transitions}$ be a
parameterized net, and let $\Xv$ be a placeset variable. An
\emph{interpretation} of $\Xv$ is a pair $\interp{X} =\tuple{n, Q}$, where
$n \geq 1$ and $Q$ is a set of places of $\PN(n)$. We identify
$\interp{X}$ and the tuple $\tuple{\interp{X}_{p}}_{p \in \places}$,  where
$\interp{X}_{p} \subseteq [n]$, defined by $j \in \interp{X}_{p}$ if{}f
$p(j) \in Q$. For example, if $\places = \set{p, q, r}$, $n = 2$, and
$Q=\set{p(0), p(1), q(1)}$, then $\tuple{\interp{X}_{p}, \interp{X}_{q},
\interp{X}_{r}} = \tuple{ \{0,1\}, \{1\}, \emptyset  }$. Given a formula
$\phi(\ldots, \Xv, \ldots)$ and an interpretation $\interp{X} = \tuple{n, Q}$
of $\Xv$, we define the formula $\phi( \ldots, \interp{X}, \ldots)$ as follows:

\begin{center}
  \begin{tabular}{rcl}
    $x \in \interp{X}_p$
      & $\coloneqq$
      & $\displaystyle \bigvee_{j \in \interp{X}_{p}} x = j$
    \\[0.2cm]
    $\Xv = \interp{X}$
      & $\coloneqq$
      & $ \bm{n} = n \wedge \displaystyle \bigwedge_{p \in \places}
        \forall \bm{x} \colon \bm{x} < \bm{n} \rightarrow \left( \bm{x} \in
        \Xv_{p} \leftrightarrow \bm{x} \in \interp{X}_p \right)$
    \\[0.2cm]
    $\phi(\ldots, \interp{X}, \ldots)$
      & $\coloneqq$
      & $\forall \Xv \colon \Xv = \interp{X} \rightarrow
        \phi(\ldots, \Xv, \ldots)$
  \end{tabular}
\end{center}

The CEGAR procedure maintains an (initially empty) set $\trapset$ of
\emph{indexed traps} of $\PN(1)$, $\PN(2)$, \ldots, $\PN(n)$, where an indexed
trap is a pair $\mathbf{T} = \tuple{i, Q}$ such that $1 \leq i \leq n$ and $Q$
is a trap of $\PN(i)$. After every update of $\trapset$ the procedure
constructs the formula $\Safety_{\trapset}$, defined as follows:

\newcommand{\trapsetf}{\mathit{TrapSet}}
\begin{equation}
\label{formula:safetychecknonpar}
\begin{aligned}
\trapsetf_\trapset(\bm{n}, \Xv) & \coloneqq && \displaystyle \bigvee_{\mathbf{X} \in \trapset} \Xv = \mathbf{X} \\[0.2cm]
\PotReach_\trapset (\bm{n}, \Mv', \Mv) & \coloneqq && \forall \Xv \colon \left(
    \begin{aligned}
      &\trapsetf_\trapset(\bm{n}, \Xv)\\
      \land {\ } &\MarkedAt(\bm{n}, \Xv ,\Mv)
    \end{aligned}\right) \rightarrow \MarkedAt(\bm{n}, \Xv,\Mv') \\[0.2cm]
\PotReach_\trapset (\bm{n}, \Mv') & \coloneqq && \exists \Mv \colon \markings(\bm{n}, \Mv) \wedge \PotReach_\trapset(\bm{n}, \Mv', \Mv) \\[0.2cm]
 \Safety_{\trapset} & \coloneqq  && \forall \bm{n} \forall \Mv \colon \left( \bm{n} < n \land \PotReach_\trapset(\bm{n}, \Mv) \right) \rightarrow \Safe(\bm{n}, \Mv)
 \end{aligned}
\end{equation}
\noindent Intuitively, $\PotReach_\trapset(\bm{n}, \Mv', \Mv)$ states that
according to the set $\trapset$ of (indexed) traps computed so far, $\Mv'$
could still be reachable from $\Mv$, because every trap of $\trapset$ marked at
$\Mv$ is also marked at $\Mv'$. Therefore,  if $\Safety_{\trapset}$ holds then
$\trapset$ is already strong enough to show that every reachable marking is
safe.

If $\trapset$ is not strong enough, then the negation of $\Safety_{\trapset}$
is satisfiable. The \acs{WS1S}-checker returns a counter-example, i.e., a model
$\interp{M} = \tuple{n, M}$ of the formula $\PotReach_\trapset(\bm{n}, \Mv)
\wedge \neg \Safe(\bm{n}, \Mv)$. Observe that $n$ is a number, and $M$ is a
marking of the instance $\PN(n)$, which is potentially reachable from an
initial marking but not safe. In this case we search for a trap of $\PN(n)$
that is marked at every initial marking of $\PN(n)$, but empty at $\interp{M}$.
Such traps are the models of the formula
\begin{equation}
\label{formula:wtrap}
  \WitTrap_\interp{M}(n, \Xv) \coloneqq \left(
  \begin{aligned}
    &\Trap(n, \Xv)\\
    \wedge &(\forall \Mv \colon \markings(n, \Mv) \rightarrow
    \MarkedAt(n, \Xv,\Mv))\\
    \wedge &\EmptyAt(n, \Xv,M)
  \end{aligned}\right)
\end{equation}
\noindent and so they can also be found with the help of the
\acs{WS1S}-checker; notice, however, that after fixing $\bm{n} \mapsto n$ the
universal quantifier of $\WitTrap_\interp{M}(n, \Xv)$ can be replaced by a
conjunction, and so $\WitTrap_\interp{M}(n, \Xv)$ is equivalent to a Boolean
formula.

If the formula has a model
$\interp{T} = \tuple{n, Q}$, then $Q$ is a trap of $\PN(n)$. We can now
take $\trapset \coloneqq \trapset \cup \{ \interp{T} \}$, and iterate. Observe
that after updating  $\trapset$ the interpretation $\interp{M}=\tuple{n, M}$ is
no longer a model of $\PotReach_\trapset(\bm{n}, \Mv) \wedge \neg \Safe(\bm{n},
\Mv)$. Since $\PN(1), \ldots, \PN(n)$ only have finitely many traps, the
procedure eventually terminates.

\subsubsection{A CEGAR loop for the parameterized case.}
\label{sec:param-cegar}
In all nontrivial examples, proving safety of the infinitely many instances
requires to compute infinitely many traps. Since the previous procedure only
computes one trap per iteration, it does not terminate. The way to solve this
problem is to insert a \emph{parametrization step} that transforms the witness
trap $\interp{T} = \tuple{n, Q}$ into a formula
$\Param_\interp{T}(\bm{n}, \Xv)$ satisfying two properties: (1) all models of
the formula are traps, and (2) $\interp{T}$ is a model. Since
$\Param_\interp{T}(\bm{n}, \Xv)$ can have infinitely many models, it
constitutes a finite representation of an infinite set of traps. These models
are also similar to each other and can be understood as capturing a single
property of the system.

\begin{example}
  Consider a parameterized net $\PN=\tuple{\places, \transitions}$ exhibiting
  rotational symmetry: For every instance $\PN(n)$, a pair $(P_1, P_2)$ of sets
  is a transition of $\PN(n)$ if{}f the pair $(P_1 \oplus_n 1, P_2 \oplus_n 1)$
  is also a transition, where $P \oplus_n 1$ denotes the result of increasing
  all indices by 1 modulo $n$. Assume that $\places = \set{p,q,r}$ and
  $\interp{T}=\tuple{3, \set{p(1), q(2)}}$, i.e., $\set{p(1), q(2)}$ is a trap
  of $\PN(3)$. It is intuitively plausible (and we will later prove) that, due
  to the rotational symmetry, $\set{p(i), q(i \oplus_m 1)}$ is a trap of
  $\PN(j)$ for every $m \geq 3$ and every $0 \leq i \leq m-1$. We can then
  define the formula $\Param_\interp{X}(\bm{n}, \Xv)$ as:
  \begin{equation*}
    \begin{aligned}
      \Param_\interp{T}(\bm{n}, \Xv) \coloneqq \bm{n} \geq 3 \; \wedge \;
        &\exists \bm{i} \colon \bm{i} < \bm{n}\\
        \land \; &\forall \bm{x} \colon \bm{x} < \bm{n} \rightarrow
      \left(\begin{aligned}
        &(\bm{x} \in \Xv_p \leftrightarrow \bm{x}=\bm{i})\\
        \wedge \; &(\bm{x} \in \Xv_q \leftrightarrow \bm{x}=\bm{i} \oplusn 1)\\
      \wedge \; &\bm{x} \notin \Xv_r\end{aligned}\right).
    \end{aligned}
  \end{equation*}
\end{example}

Now, in order to describe the CEGAR procedure for the parameterized case we
only need to \emph{redefine} the formula $\trapsetf_\trapset(\bm{n}, \Xv)$.
Instead of the formula $\trapsetf_\trapset(\bm{n}, \Xv) \coloneqq
\bigvee_{\interp{T} \in \trapset} \Xv=\interp{T}$, which holds only
when $\Xv$ is one of the finitely many traps in $\trapset$, we insert the
parametrization procedure and define
\begin{equation}
\label{formula:safetycheck}
\begin{aligned}
\trapsetf_\trapset(\bm{n}, \Xv) & \coloneqq && \bigvee_{\interp{T} \in \trapset}
  \Param_\interp{T}(\bm{n}, \Xv) \\[0.2cm]
\PotReach_\trapset (\bm{n}, \Mv', \Mv) & \coloneqq && \forall \Xv \colon \left(
    \begin{aligned}
      &\trapsetf_\trapset(\bm{n}, \Xv)\\
      \land {\ } &\MarkedAt(\bm{n}, \Xv ,\Mv)
    \end{aligned}\right) \rightarrow \MarkedAt(\bm{n}, \Xv,\Mv') \\[0.2cm]
\PotReach_\trapset (\bm{n}, \Mv') & \coloneqq && \exists \Mv \colon \markings(\bm{n}, \Mv) \wedge \PotReach_\trapset(\bm{n}, \Mv', \Mv) \\[0.2cm]
 \Safety_{\trapset} & \coloneqq  && \forall \bm{n} \forall \Mv \colon \PotReach_\trapset(\bm{n}, \Mv) \rightarrow \Safe(\bm{n}, \Mv)
 \end{aligned}
\end{equation}
Notice the two differences with (\ref{formula:safetychecknonpar}): the definition of $\trapsetf_\trapset(\bm{n}, \Xv)$, and the absence
of the condition $\bm{n} < n$ in the definition of  $\Safety_{\trapset}$.
The question is how to obtain the
formula $\Param_\interp{T}(\bm{n}, \Xv)$ from $\interp{T}$. We discuss this
point in the rest of the section.

\paragraph{A semi-automatic approach}
If we \emph{guess} the formula $\Param_\interp{T}(\bm{n}, \Xv)$ we can use the
WS1S-checker to automatically prove that the guess is correct. Indeed, it
suffices to check that all models of $\Param_\interp{T}(\bm{n}, \Xv)$ are
traps, which reduces to proving that the formula
\begin{equation*}
  \forall \bm{n}\forall \Xv \colon \Param_\interp{T}(\bm{n}, \Xv) \rightarrow
  \Trap(\bm{n}, \Xv)
\end{equation*}
holds. Let us see how this works in Example \ref{ex:diningphil}. Assume that
the CEGAR procedure produces a trap $\interp{T}=\tuple{3, \set{p(1), q(2)}}$.
The user finds it plausible that, due to the identical behavior of philosophers
$1, 2, \ldots, n-1$, the set $\set{p(i), q(i \oplus 1)}$ will be a trap of
$\PN(n)$ for every $n \geq 3$ and for every $1 \leq i \leq n-2$ (i.e., the user
excludes the case in which $i$ or $i \oplus_n 1$ are equal to $0$). So the user
guesses a new formula

\begin{align*}
  \Param_\interp{T}(\bm{n}, \Xv) \coloneqq
  & \bm{n} \geq 3 \wedge \exists \bm{i} \colon (1 \leq \bm{i} \leq \bm{n}-2)
  \wedge \forall \bm{x} \colon \\
  & (\bm{x} \in \Xv_{p} \leftrightarrow \bm{x}=\bm{i})
  \wedge (\bm{x} \in \Xv_{q} \leftrightarrow \bm{x}=\bm{i} \oplusn 1) \wedge
  \bm{x} \notin \Xv_{r}.
\end{align*}

\noindent The user now automatically checks that all models of
$\Param_\interp{T}(\bm{n}, \Xv)$ are traps. The formula can then be safely
added to $\trapsetf_\trapset(\bm{n}, \Xv)$ as a new disjunct.

\paragraph{An automatic approach for specific architectures.}
Parameterized Petri nets usually have a regular structure. For example, in the
readers-writers problem all processes are indistinguishable, and in the
philosophers problem, all right-handed processes behave in the same way. In the
next sections we show how the structural properties of ring topologies and
crowds (two common structures for parameterized systems) can be exploited to
automatically compute the formula $\Param_\interp{T}(\bm{n}, \Xv)$ for each
witness trap $\interp{T}$.

%% file: tex/rings.tex
\section{Trap parametrization in rings}
\label{sec:rings}
Intuitively, a parameterized net $\PN$ is a ring if for every transition of
every instance $\PN(n)$ there is an index $i \in [n]$ and sets $\mathcal{P}_L,
\mathcal{P}_R, \mathcal{Q}_L, \mathcal{Q}_R \subseteq \places$ such that the
preset of the transition is $(\mathcal{P}_L \times \set{i}) \cup (\mathcal{P}_R
\times \set{i \oplus_n 1})$ and the postset is $(\mathcal{Q}_L \times {i}) \cup
(\mathcal{Q}_R \times {i \oplus_n 1})$. In other words, every transition
involves only two neighbor-processes of the ring. In a fully symmetric ring all
processes behave identically, while in a headed ring there is one
distinguished process, as in Example~\ref{ex:dining-philosophers}. To ease
presentation in this section we only consider fully symmetric rings. The
extension to headed rings can be found in \cite{FullPaperVersion}.

The informal statement \enquote{all processes behave identically} is captured
by requiring the existence of a finite set of \emph{transition patterns}
$\tuple{\mathcal{P}_L, \mathcal{P}_R, \mathcal{Q}_L, \mathcal{Q}_R}$ such that
the transitions of $\PN(n)$  are the result of \enquote{instantiating} each
pattern with all pairs $i$ and $i \oplus_n 1$ of consecutive indices.

\begin{definition}
  A parameterized net $\PN=\tuple{\places, \transitions}$ is a \emph{fully
  symmetric ring} if there is a finite set of transition patterns of the form
  $\tuple{\mathcal{P}_L, \mathcal{P}_R, \mathcal{Q}_L, \mathcal{Q}_R}$, where
  $\mathcal{P}_L, \mathcal{P}_R, \mathcal{Q}_L, \mathcal{Q}_r \subseteq
  \places$, such that for every instance $\PN(n)$ the following condition
  holds: $\tuple{P, Q}$ is a transition of $\PN(n)$ if{}f there is $i \in [n]$
  and a pattern  such that $P = \mathcal{P}_{L}\times\set{i}\cup
  \mathcal{P}_{R}\times\set{i\oplus_n 1}$ and $Q =
  \mathcal{Q}_{L}\times\set{i}\cup \mathcal{Q}_{R}\times\set{i\oplus_n 1}$.
\end{definition}

It is possible to decide if a given parameterized Petri net is a fully
symmetric ring:
\begin{proposition}
  There is a formula of WS1S such that a parameterized net is a fully symmetric
  ring if{}f  the formula holds.
\end{proposition}

\begin{proof}
  We introduce a WS1S formula describing symmetric rings in several steps. To
  avoid dealing with edge cases we assume that any transition formula
  $\transitions(\bm{n}, \Xv, \Yv)$ enforces a minimal size of its models; i.e.,
  $\transitions(\bm{n}, \Xv, \Yv) \models \bm{n} > 3$. This streamlines the
  argument and formulas. However, it is straightforward to adapt the formulas
  to the full generality.

  The following formula expresses that for every transition of every instance
  there is an index $i$ such that all places in the preset and postset of the
  transition have index $i$ or $i \oplusn 1$. We call $i$ the index of the
  transition.
  \begin{equation}
    \label{eq:ring-phi}
    \begin{aligned}
    \varphi \coloneqq \forall \bm{n}, \Xv, \Yv \colon
      & \transitions(\bm{n}, \Xv, \Yv)
      \longrightarrow \\
      &\exists \bm{i}: \bm{i} < \bm{n} \land \forall \bm{x}: \bm{x} < \bm{n}
      \rightarrow \left(
      \begin{aligned}
        \bigvee\limits_{p\in \places}&\bm{x} \in \Xv_{p} \lor \bm{x} \in \Yv_{p}
        \\
        &\leftrightarrow [\bm{x} = \bm{i} \lor \bm{x} = \bm{i}\oplusn 1]
      \end{aligned}
      \right)
    \end{aligned}
  \end{equation}
  Now we express that if some instance, say $\PN(n)$, contains a transition
  with index $i$, then for every other instance, say $\PN(m)$, and for every
  index $0 \leq j \leq m$, substituting $j$ for $i$ yields a transition of
  $\PN(m)$:
  \begin{equation}
    \label{eq:ring-psi}
    \begin{aligned}
      \psi \coloneqq \forall \bm{n}, \bm{i}, \Xv, \Yv, \bm{m}, \bm{j} \colon
      &\left(\bm{i} < \bm{n} \land \transitions(\bm{n}, \Xv, \Yv) \land \bm{j} < \bm{m} \right) \longrightarrow &\\
      & \exists \Xv', \Yv' \colon  \transitions(\bm{m}, \Xv', \Yv') \land {\ } \\
      &\begin{aligned}
        \bigwedge\limits_{p\in\places}\left(
        \begin{aligned}
          &\bm{i} \in \Xv_{p} \leftrightarrow \bm{j} \in \Xv_{p}'\\
          \land \;&\bm{i} \in \Yv_{p} \leftrightarrow \bm{j} \in \Yv_{p}'\\
          \land \;&\bm{i} \oplusn 1 \in \Xv_{p} \leftrightarrow \bm{j} \oplusm 1
            \in \Xv_{p}'\\
          \land \; &\bm{i} \oplusn 1 \in \Yv_{p} \leftrightarrow \bm{j} \oplusm 1
            \in \Yv_{p}'\\
        \end{aligned}
        \right).\\
      \end{aligned}
    \end{aligned}
  \end{equation}

  We prove that $\PN = \tuple{\places, \transitions}$ is a fully symmetric ring
  if{}f its associated formula $\varphi \wedge \psi$ is valid.

  First, we show that $\varphi$ is valid if and only if for all $n$ and every
  transition $\tuple{P_1, P_2}$ of $T_{n}$ there is an index $0 \leq i \leq
  n-1$ such that $P_1 \cup P_2 \subseteq \places \times \set{i, i \oplus_n 1}$.

  Assume for all $n$ and every transition $\tuple{P_1, P_2}$ of $T_{n}$ there
  is an index $0 \leq i \leq n-1$ such that $P_1 \cup P_2 \subseteq \places
  \times \set{i, i \oplus_n 1}$.
  Then for any interpretation $\mu$ of $\bm{n}$, $\Xv$, $\Yv$ with
  $\mu\models\transitions(\bm{n}, \Xv, \Yv)$ we have $\tuple{P, Q}$ and an
  index $j < \mu(\bm{n})$ such that
  $\mu(\Xv_{p}) = P\cap\set{p}\times\set{j, j\oplus_{\mu(\bm{n})} 1}$ and
  $\mu(\Yv_{p}) = Q\cap\set{p}\times\set{j, j\oplus_{\mu(\bm{n})} 1}$ for all
  $p\in\places$.

  Consequently,
  \begin{equation*}
    \mu[\bm{i} \mapsto i] \models \forall \bm{x}: \bm{x} < \bm{n} \rightarrow
    \left(
    \begin{aligned}
      \bigvee\limits_{p\in \places}
      &
      x \in \Xv_{p} \lor x \in \Yv_{p}\\
     &
      \leftrightarrow [x = i \lor x = i\oplusn 1]\\
   \end{aligned}
    \right).
  \end{equation*}
  Which renders $\varphi$ valid in general.

  On the other hand, if $\varphi$ is valid careful examining $\varphi$ gives
  the desired result: let $\mu \models \transitions(\bm{n}, \Xv, \Yv)$. Then
  fix any $i \in [\mu(\bm{n})]$ such that
  \begin{equation*}
    \mu[\bm{i} \mapsto i] \models \forall \bm{x}: \bm{x} < \bm{n} \rightarrow
    \left(
    \begin{aligned}
      \bigvee\limits_{p\in \places}&x \in \Xv_{p} \lor x \in \Yv_{p}\\
      &\leftrightarrow [x = i \lor x = i\oplusn 1]\\
    \end{aligned}
    \right).
  \end{equation*}
  For the transition$\tuple{P, Q}$ of $\mu$; i.e.,
  \begin{equation*}
    \begin{aligned}
      &P = \set{\tuple{p, i} \in \places \times [\mu(\bm{n})]\mid i \in
      \mu(\Xv_{p})},\\
      &Q = \set{\tuple{p, i} \in \places \times [\mu(\bm{n})]
      \mid i \in \mu(\Yv_{p})}
    \end{aligned}
  \end{equation*}
  we see that $P \subseteq \places\times\set{i, i\oplus_{\mu(\bm{n})} 1}$ and
  $Q \subseteq \places\times\set{i, i\oplus_{\mu(\bm{n})} 1}$.

  Using this observation we restrict the remaining argument to the case that
  every transition of $\PN(n)$ has an index $i \in [n]$.
  It remains to show that -- under this condition -- $\psi$ is valid if and
  only if $\PN$ is a fully symmetric ring: assume $\PN$ to be a fully symmetric
  ring. Let $\mu$ be an arbitrary interpretation of $\bm{n}$, $\bm{m}$, $\Xv$,
  $\Yv$, $\bm{i}$, $\bm{j}$. If $\mu \not \models
  \bm{i} < \bm{n} \land \transitions(\bm{n}, \Xv, \Yv) \land
  \bm{j} < \bm{m}$ then there is nothing to show. Let now $\mu \models
  \bm{i} < \bm{n} \land \transitions(\bm{n}, \Xv, \Yv) \land \bm{j} < \bm{m}$.
  Let $n = \mu(\bm{n})$, $m = \mu(\bm{m})$, $i = \mu(\bm{i})$,
  $j = \mu(\bm{j})$ and $\tuple{P, Q}$ such that $P = \set{\tuple{p, i}\in
  \places\times [\mu(\bm{n})]\mid i \in \mu(\Xv_{p})}$ and $Q = \set{
    \tuple{p, i}\in \places\times [\mu(\bm{n})]\mid i \in \mu(\Yv_{p})}$. Since
  $\PN$ is assumed to be a fully symmetric ring we know that $\tuple{P, Q}$ is
  an instance of the pattern $\tuple{\mathcal{P}_L, \mathcal{P}_R,
  \mathcal{Q}_L, \mathcal{Q}_R}$ at an index $i$. More formally,
  $P = P_{L} \times \set{i} \cup P_{R} \times \set{i \oplus_{n} 1}$ and
  $Q = Q_{L} \times \set{i} \cup Q_{R} \times \set{i \oplus_{n} 1}$. If
  $\mu(\bm{i}) \notin \set{i, i \oplus_{n} 1}$, then expanding the
  interpretation $\mu$ to an interpretation $\mu'$ which chooses values
  $\mu'(\Xv')$ and $\mu'(\Yv')$ which yield a transition $\tuple{P', Q'}$
  as an instance of $\tuple{\mathcal{P}_L, \mathcal{P}_R, \mathcal{Q}_L,
  \mathcal{Q}_R}$ for an index $j'$ such that $\set{j', j' \oplus_{m} 1} \cap
  \set{j, j \oplus_m 1} = \emptyset$. (Note that we use implicitly here that
  the formula $\transitions$ enforces models of sufficient size. Adapting
  $\psi$ such that $\bm{i}$ has to be the index of $\tuple{P, Q}$ is
  straightforward.)

  On the other hand, if $\mu(\bm{i}) = i$ then expanding $\mu$ to $\mu'$ with
  values for $\mu'(\Xv')$ and $\mu'(\Yv')$ such that the associated
  $\tuple{P', Q'}$ is an instance of $\tuple{\mathcal{P}_L, \mathcal{P}_R,
  \mathcal{Q}_L, \mathcal{Q}_R}$ at index $\mu(\bm{j})$ yields the desired
  result. Analogously, for $\mu(\bm{i}) = i \oplus_n 1$. It follows that
  $\psi$ is valid.

  Now, assume $\psi$ to be valid. The result follows from carefully examining
  $\psi$. For any transition $\tuple{P, Q}$ in an instance $\PN(n)$ we can
  extract its structure; i.e., a pattern $\tuple{\mathcal{P}_L, \mathcal{P}_R,
  \mathcal{Q}_L, \mathcal{Q}_R}$ such that $P = P_{L} \times \set{i} \cup P_{R}
  \times \set{i \oplus_{n} 1}$ and $Q = Q_{L} \times \set{i} \cup Q_{R} \times
  \set{i \oplus_{n} 1}$ for an appropriate $i$ (remember that we assume
  $\varphi$ to be valid). By the validity of $\psi$ we see that the same
  pattern can be instantiated (represented by the choice of $\Xv'$ and $\Yv'$)
  at all other indices (corresponding to the choice for $\mu(\bm{j})$) for all
  other instances (corresponding to the choice for $\mu(\bm{m})$).
\end{proof}

We need to distinguish between \emph{global} and \emph{local} traps of an
instance. Loosely speaking, a global trap contains places of all processes,
while a local trap does not. To understand why this is relevant, consider a
fully symmetric ring $\PN = \tuple{\places, \transitions}$ where $\places=
\set{p, q}$ and the transitions of each instance $\PN(n)$ are the pairs
$\tuple{\set{p(i), q(i \oplus_n 1)}, \set{p(i \oplus_n 1), q(i)}}$ for every
$i\in [n]$. The sets $\set{p(0), q(0)}$ and $\set{p(0), p(1), p(2), p(3)}$ are
both traps of $\PN(4)$ (they are even 1-balanced sets). However, they are of
different nature. Intuitively, in order to decide that $\set{p(0), q(0)}$ is a
trap it is not necessary to inspect all of $\PN(4)$, but only process 0 and its
neighborhood. On the contrary, $\set{p(0), p(1), p(2), p(3)}$ involves all
processes. This has consequences when parametrizing. Due to the symmetry of
the ring, $\set{p(i),q(i)}$ is a trap of every instance $\PN(n)$ for every $i
\in [n]$. However,  $\set{p(i), p(i \oplus_n 1), \ldots, p(i \oplus_n 3)}$ is
\emph{not} a trap of every instance for every $i \in [n]$, for example
$\set{p(0), p(1), p(2), p(3)}$ is not a trap of $\PN(5)$. The correct
parametrization is a different one, namely $\set{p(0), p(1), \ldots,
p(n-1)}$. The difference between the two traps is captured by the following
definition.

\begin{definition}
  Let $\PN=\tuple{\places, \transitions}$ be a parameterized net. An indexed
  trap $\interp{T} = \tuple{n, Q}$ of $\PN$ is \emph{global} if $Q \cap
  (\places \times \{i\}) \neq \emptyset$ for every $i \in [n]$, otherwise
  $\interp{T}$ is \emph{local}.
\end{definition}

\subsection{Parametrizing local traps}

We first observe that local indexed traps can be \enquote{shifted} locally
while maintaining their trap property.
\begin{lemma}
\label{lem:trap-by-one}
  Let $\PN = \tuple{\places, \transitions}$ be a fully symmetric ring and let
  $\tuple{n, Q}$ be a local indexed trap of $\PN$. Then $\tuple{n,Q'}$ with $Q'
  = \set{ \tuple{p, i \oplus_n 1}: \tuple{p, i} \in Q}$ is a local indexed trap
  of $\PN$.
\end{lemma}

\begin{proof}
  Assume $Q'$ is not an indexed local trap. Then there is
  $t \in T_{n}$ such that $\preset{t} \cap Q' \neq \emptyset = \postset{t}
  \cap Q'$. Since $\PN$ is a fully symmetric ring, there is a pattern
  $\tuple{\tuple{P_{L}, P_{R}}, \tuple{Q_{L}, Q_{R}}}$ and an index $i \in [n]$
  such that $t$ is the instance of $\tuple{\tuple{P_{L}, P_{R}}, \tuple{Q_{L},
  Q_{R}}}$ with index $i$. Let $t'$ be the transition obtained from the same
  pattern with index $i \oplus_{n} (n-1)$; i.e., moved one index to the left.
  It follows $\preset{t'} = \set{\tuple{p, j \oplus_{n} (n-1)} \colon
  \tuple{p, j} \in \preset t}$ and $\postset{t'} = \set{\tuple{p, j \oplus_{n}
  (n-1)} \colon \tuple{p, j} \in \postset t}$. By definition of $Q'$ we have
  $Q = \set{\tuple{p, j \oplus_{n} (n-1)}\colon \tuple{p, j}\in Q'}$. That,
  however, gives $\preset{t'} \cap Q \neq \emptyset = \postset{t'} \cap Q$ in
  contradiction to $Q$ being a local indexed trap.
\end{proof}


Our second lemma states that for any indexed local traps $\tuple{n, Q}$ with
$Q \cap (\places \times \set{n-1})$, the set $Q$ remains a trap in any instance
$\PN(n')$ with $n \leq n'$.
\begin{lemma}
\label{lem:trap-in-larger-instances}
  Let $\PN$ be a fully symmetric ring and $\tuple{n, Q}$ a local indexed trap
  such that $Q \cap (\places \times \set{n-1}) = \emptyset$. Then
  $\tuple{n', Q}$ is a local indexed trap for all $n' \geq n$.
\end{lemma}

\begin{proof}
  Assume the statement is false. That is, $\tuple{n', Q}$ is not a local
  indexed trap. If $n' = n$ there is an immediate contradiction with the
  assumption that $\tuple{n, Q}$ is a local indexed trap.
  Hence, let $n' > n$ \emph{minimal} such that $\tuple{n', Q}$ is \emph{not} a
  local indexed trap. So there is a transition $t$ in $\PN(n')$ such that
  $\preset{t} \cap Q \neq \emptyset = \postset{t} \cap Q$.
  Since $Q \cap (\places \times \set{n-1}) = \emptyset$ by assumption of the
  lemma and $n' > n$ by case distinction we have
  $Q \cap (\places \times \set{n' - 2, n' - 1}) = \emptyset$. With this and the
  facts that fully symmetric rings only allow for transitions using places of
  two adjacent indices and $\preset{t} \cap Q \neq \emptyset$ we get
  $\preset{t} \cap \places \times \set{n' - 1} = \emptyset$ and
  $\postset{t} \cap \places \times \set{n' - 1} = \emptyset$. That means,
  however, that $t$ is also a transition in $\PN(n' - 1)$ because $\PN$ is a
  fully symmetric ring and, consequently, $\tuple{n' - 1, Q}$ already is not a
  local indexed trap. This contradicts that $n'$ was chosen minimal and
  concludes the proof.
\end{proof}

We can now show how to obtain a sound parameterization of a given indexed trap.
The formula $\Param_\interp{T}(\Xv)$ states that $\Xv$ is the result of
\enquote{shifting} $\interp{T}=\tuple{n, Q}$ in $\PN(n')$ for some $n' \geq n$.

\begin{theorem}
  \label{thm:generalize-in-rings}
  Let $\PN=\tuple{\places, \transitions}$ be a fully symmetric ring and let
  $\tuple{n, Q}$ be a local indexed trap of $\PN(n)$ such that $Q \subseteq
  (\places \times I)$ for a minimal set $I \subset [n]$.  Assume $I =
  \set{i_{0}, \ldots, i_{k-1}}$ with $0 \leq i_0 < i_1 < \ldots <  i_{k-1} <
  n-1$. Then every model of the formula
  \begin{equation*}
    \begin{aligned}
      \Param_\interp{T}& (\bm{n}, \Xv)\coloneqq
       n \leq \bm{n} \; \land \; \exists \bm{y}\colon \bm{y} < \bm{n} \; \land \; \bigwedge_{p \in \places}
      \forall \bm{x} \colon \bm{x} < \bm{n} \rightarrow \\
      &\left(
      \bm{x} \in \Xv_{p} \leftrightarrow
      \left(
      \begin{aligned}
        &\bigvee_{\tuple{i_{0}, p} \in Q} \bm{x} = \bm{y}\\
        &\lor \bigvee_{j > 0, \tuple{i_{j}, p} \in Q}
          \bm{x} = \bm{y} \oplusn (i_{j} - i_{j-1})
      \end{aligned}
      \right)\right)
    \end{aligned}
  \end{equation*}
  is an indexed trap of $\PN$.
\end{theorem}

\begin{proof}
  Assume $\mu \models \Param_\interp{T}(\bm{n}, \Xv)$. Then there exists a
  tuple $\tuple{k, P}$ such that $\mu(\bm{n}) = k$ and $\mu(\Xv_{p}) =
  \set{i\in [k]\mid \tuple{p, i}\in P}$ for every $p \in \places$. We have  $n
  \leq k$ by the first conjunct of $\Param_\interp{T}$. Let $j \in [k]$ be any
  value such that assigning $j$ to  $\bm{y}$ satisfies the existentially
  quantified subformula of $\Param_\interp{T}$.

  Since $i_{k-1} < n-1$ we have $Q\cap (\places\times\set{n-1}) = \emptyset$.
  So we can apply Lemma~\ref{lem:trap-in-larger-instances} to $\tuple{n, Q}$,
  and in fact we can apply it $(n-k)$ times, yielding a local trap $\tuple{k,
  Q}$. Now, fix indices $i_{0}', \ldots, i_{k-1}'$ such that $i_{0}' = j$ and
  $i_{\ell+1}' = i_{\ell}'\oplus_{k}(i_{\ell+1}-i_{\ell})$ for $0 \leq \ell <
  k-1$. Carefully examining $\Param_\interp{T}$ one can now observe that
  \begin{equation*}
    \label{eq:PandQ}
    \set{p\in\places \mid \tuple{p, i_{\ell}'} \in P}
    = \set{p\in\places \mid \tuple{p, i_{\ell}} \in P}\text{ for all }0 \leq
    \ell < k
  \end{equation*}
  and $\emptyset = \places\times([k]\setminus\set{i_{0}', \ldots, i_{k-1}'})$.
  We can now apply $(k - i_{0}) + i_{0}'$ times Lemma~\ref{lem:trap-by-one}  to
  the local trap $\tuple{k, Q}$, which shows that $\tuple{k, P}$ is a local
  trap.
\end{proof}

\begin{remark}
  Since Theorem  \ref{thm:generalize-in-rings} requires $i_{k-1} < n-1$, it can
  only be applied to local traps $\tuple{n, Q}$ such that $Q \cap
  (\places\times\set{n-1}) = \emptyset$. However, for every local trap
  $\tuple{n, Q}$ Lemma~\ref{lem:trap-by-one} allows us to find a local trap
  $\tuple{n, Q'}$ satisfying $Q' \cap (\places\times\set{n-1}) = \emptyset$,
  which we can then parameterize via Theorem~\ref{thm:generalize-in-rings}.
\end{remark}

\subsection{Parametrizing global traps}
In contrast to local traps, global traps involve all indices $[n]$ of the
instance $\PN(n)$. Let $\tuple{n, Q}$ be an indexed global trap. We denote with
$Q[i]$ the set $P \subseteq \places$ such that $P \times \set{i} = Q \cap
(\places \times \set{i})$; i.e., the set of places in $Q$ at index $i$. Moreover,
we say $Q$ has period $p$ if $p$ is the smallest divisor of $n$ such that
for all $0 \leq j < p$ we have $Q[j] = Q[k \cdot p + j]$ for all $0 \leq k <
\frac{n}{p}$. That is, $Q$ is a repetition of the same $p$ sets in a row. Since
$n$ is a period of $Q$ we know that every $Q$ has a period, which we denote
$p_{Q}$. Recall the global trap $Q = \set{p(0), p(1), p(2), p(3)}$ from before.
Then, $Q[0] = Q[1] = Q[2] = Q[3] = \set{p}$ and, consequently, $p_{Q} = 1$.
Intuitively, we can repeat a period over and over again and still obtain a
trap. So we can parameterize global traps by capturing the repetition of
periodic behavior:
\begin{theorem}
  \label{thm:generalizing-global-traps-rings}
  Let $\tuple{n, Q}$ be an indexed global trap
  with $n\geq{}2$. Then every model of the formula
  \begin{equation*}
    \begin{aligned}
      \Param_\interp{T}&(\bm{n}, \Xv)\coloneqq \exists \bm{P}: 0 \in \bm{P}
      \land \bm{n} \in \bm{P}\\
      \land \;\; &\forall \bm{x}: \bm{x} \leq \bm{n}
        \rightarrow \bm{x} \in \bm{P} \leftrightarrow
        \left(
        \begin{aligned}
          &\bigwedge_{0 \leq k < p_{Q}} \bm{x} + k \notin \bm{P}\\
          \land \;\; &\bm{x} + p_{Q} \in \bm{P}
        \end{aligned}
        \right)\\
      \land \;\; &\forall \bm{x}_{0}, \ldots, \bm{x}_{p_{Q}-1}:
      \left(
      \begin{aligned}
        &\bigwedge_{0 < k \leq p_{Q}-1} \bm{x}_{k-1} + 1 = \bm{x}_{k}\\
        \land \;\; &\bm{x}_{p_{Q}-1} < n \land \bm{x}_{0} \in \bm{P}
      \end{aligned}
      \right)\\
      &\rightarrow \bigwedge_{0 \leq k < p_{Q}}\bigwedge_{p \in Q[k]}\bm{x}_{k}
      \in \Xv_{p} \land \bigwedge_{p \in \places\setminus Q[k]} \bm{x}_{k}
      \notin \Xv_{p}
    \end{aligned}
  \end{equation*}
  is an indexed global trap.
\end{theorem}

\begin{proof}
  Let $\mu$ be a model of $\Param_\interp{T}(\bm{n}, \Xv)$. Observe that we
  have $\mu(\bm{P}) = \set{0, p_{Q}, 2\cdot p_{Q},  \ldots, \ell \cdot p_{Q}}$
  for some $\ell > 0$. Let $k := \ell\cdot p_{Q} = \mu(\bm{n})$ and let $P$ be
  the set of places of $\PN(k)$ such that $\mu(\Xv_{p}) = \set{i \in [k]:
  \tuple{p, i} \in P}$. Examining $\Param_\interp{T}(\bm{n}, \Xv)$ further we
  observe that $P[j\cdot p_{Q} + o] = Q[o]$ holds for all $0 \leq o < p_{Q}$
  and $0 \leq j < \ell$.

  It remains to show that $P$ is indeed a trap. Assume the contrary. Then there
  is a transition $t$ in $\PN(k)$ such that $P\cap\preset{t} \neq \emptyset =
  P\cap\postset{t}$. Since $\PN$ is a fully symmetric ring there is an index
  $i\in [k]$ such that $\preset{t} \cup \postset{t}  \subseteq
  \places\times\set{i, i\oplus_{k} 1}$. Pick $j$ such that $i = j\cdot p_{Q} +
  o$ for $0 \leq o < p_{Q}$. Observe that $P[i] = Q[o]$ and $P[i \oplus_{k} 1]
  = Q[o \oplus_{n} 1]$. Again, by $\PN$ being a symmetric ring, we can find a
  transition $t'$ such that $\preset{t'} = \set{\tuple{p, o}:\tuple{p, i} \in
  \preset{t}} \cup \set{\tuple{p, o \oplus_{n} 1}:\tuple{p, i \oplus_{k} 1} \in
  \preset{t}}$ and $\postset{t'} = \set{\tuple{p, o}:\tuple{p, i} \in
  \postset{t}} \cup \set{\tuple{p, o \oplus_{n} 1}:\tuple{p, i \oplus_{k} 1}
  \in \postset{t}}$. This, however, yields a contradiction since $t'$ is a
  witness for $Q$ not being a trap in contradiction to the assumptions.
\end{proof}

%% file: tex/crowds.tex
\section{Trap parametrization in barrier crowds}
\label{sec:crowds}

Barrier crowds are parameterized systems in which communication happens by
means of global steps in which each process makes a move.  An initiator process
decides to start a step, and all the other processes get a chance to veto
it; if the step is not blocked (if all the processes accept it), all the
processes, including the initiator, update their local state. Barrier crowds
are slightly more general than broadcast protocols
\cite{DBLP:conf/lics/EsparzaFM99}, which, loosely speaking, correspond to the
special case in which no process makes use of the veto capability. Like
broadcast protocols, barrier crowds can be used to model cache coherence
protocols \cite{Delzanno00}.

As for fully symmetric rings, transitions of the instances of a barrier crowd
are generated from a finite set of \enquote{transition patterns}. A transition
pattern of a barrier crowd $\PN$ is  a pair $\tuple{\mathcal{I}, \mathbb{A}}$,
where $\mathcal{I} \in 2^\places \times 2^\places$  and $\mathbb{A} \subseteq
2^\places \times 2^\places$.  Assume for example that each process can be in
states $p, q, r$, and maintains a boolean variable with values $\set{\falseval,
\trueval}$. The corresponding parameterized net has $\places = \set{p, q, r,
\falseval, \trueval}$ as set of places. Consider the transition pattern  with
$\mathcal{I} = \tuple{\set{p, \falseval}, \set{q, \trueval} }$, and $\mathbb{A}
= \set{ \tuple{\set{p}, \set{p}}, \tuple{\set{q, \falseval}, \set{r,
\falseval}}, \tuple{\set{q, \trueval}, \set{r, \falseval}} }$. This pattern
models that the initiator process, say process $i$, proposes a step that takes
it from $p$ to $q$, setting its variable to $\trueval$. Each other process
reacts as follows, depending on its current state: if in $p$, it stays in $p$,
leaving the variable unchanged; if in $q$, it moves to $r$, setting the
variable to $\falseval$; if in $r$, it vetoes the step (because $\mathbb{A} $
does not offer a way to accept from state $r$). We depict an instance with
three agents for this example in Figure~\ref{fig:crowd-example}.

\begin{figure}
  \caption{An example $\PN(3)$ of an instance of a crowd $\PN$. The places of
  $\PN$ are $\places = \set{p, q, r, \falseval, \trueval}$ and we consider only
  the transition pattern with $\mathcal{I} = \tuple{\set{p, \falseval},
  \set{q, \trueval}}$ and $\mathbb{A} = \set{\tuple{\set{p}, \set{p}},
  \tuple{\set{q, \falseval}, \set{r, \falseval}}, \tuple{\set{q, \trueval},
  \set{r, \falseval}}}$. We give the transitions only for the case that the
  agent with index $0$ executes the pattern of $\mathcal{I}$. This means we
  give 9 transitions where agents $1$ and $2$ execute one of the patterns
  $\tuple{\set{p}, \set{p}}$, $\tuple{\set{q, \falseval}, \set{r, \falseval}}$,
  $\tuple{\set{q, \trueval}, \set{r, \falseval}}$. Specifically,
  we only drew the sixth transition from the left continuously and red. This
  transition corresponds to agent 1 using the pattern $\tuple{\set{q,
  \falseval}, \set{r, \falseval}}$ and agent 2 using the pattern
  $\tuple{\set{q, \trueval}, \set{r, \falseval}}$.}
  \label{fig:crowd-example}
  \begin{center}
    \resizebox{1.0\textwidth}{!}{%
      \begin{tikzpicture}[every node/.style={scale=0.5}]
        \node [place, minimum width=1.5cm] (p0)     at (-2,  0) {$p(0)$};
        \node [place, minimum width=1.5cm] (q0)     at (-1,  0) {$q(0)$};
        \node [place, minimum width=1.5cm] (r0)     at ( 0,  0) {$r(0)$};
        \node [place, minimum width=1.5cm] (false0) at ( 1,  0) {$\falseval(0)$};
        \node [place, minimum width=1.5cm] (true0)  at ( 2,  0) {$\trueval(0)$};

        \node [place, minimum width=1.5cm] (p1)     at (-6, -8.0) {$p(1)$};
        \node [place, minimum width=1.5cm] (q1)     at (-5, -8.3) {$q(1)$};
        \node [place, minimum width=1.5cm] (r1)     at (-4, -8.6) {$r(1)$};
        \node [place, minimum width=1.5cm] (false1) at (-3, -8.9) {$\falseval(1)$};
        \node [place, minimum width=1.5cm] (true1)  at (-2, -9.2) {$\trueval(1)$};

        \node [place, minimum width=1.5cm] (p2)     at ( 2, -9.2) {$p(2)$};
        \node [place, minimum width=1.5cm] (q2)     at ( 3, -8.9) {$q(2)$};
        \node [place, minimum width=1.5cm] (r2)     at ( 4, -8.6) {$r(2)$};
        \node [place, minimum width=1.5cm] (false2) at ( 5, -8.3) {$\falseval(2)$};
        \node [place, minimum width=1.5cm] (true2)  at ( 6, -8.0) {$\trueval(2)$};

        \node [transition] (pp) at (-6, -4) {}
          edge [dashed, pre]  (p0)
          edge [dashed, pre]  (false0)
          edge [dashed, post] (q0)
          edge [dashed, post] (true0)
          edge [dashed, <->, shorten >=1pt, shorten <=1pt]  (p1)
          edge [dashed, <->, shorten >=1pt, shorten <=1pt]  (p2)
          ;

        \node [transition] (pqfrf) at (-4.5, -4) {}
          edge [dashed, pre]  (p0)
          edge [dashed, pre]  (false0)
          edge [dashed, post] (q0)
          edge [dashed, post] (true0)
          edge [dashed, <->, shorten >=1pt, shorten <=1pt]  (p1)
          edge [dashed, pre]  (q2)
          edge [dashed, <->, shorten >=1pt, shorten <=1pt]  (false2)
          edge [dashed, post] (r2)
          ;

        \node [transition] (pqtrf) at (-3, -4) {}
          edge [dashed, pre]  (p0)
          edge [dashed, pre]  (false0)
          edge [dashed, post] (q0)
          edge [dashed, post] (true0)
          edge [dashed, <->, shorten >=1pt, shorten <=1pt]  (p1)
          edge [dashed, pre]  (q2)
          edge [dashed, pre]  (true2)
          edge [dashed, post] (false2)
          edge [dashed, post] (r2)
          ;

        \node [transition] (qfrfp) at (-1.5, -4) {}
          edge [dashed, pre]  (p0)
          edge [dashed, pre]  (false0)
          edge [dashed, post] (q0)
          edge [dashed, post] (true0)
          edge [dashed, pre]  (q1)
          edge [dashed, <->, shorten >=1pt, shorten <=1pt]  (false1)
          edge [dashed, post] (r1)
          edge [dashed, <->, shorten >=1pt, shorten <=1pt]  (p2)
          ;

        \node [transition] (qfrfqfrf) at (0, -4) {}
          edge [dashed, pre]  (p0)
          edge [dashed, pre]  (false0)
          edge [dashed, post] (q0)
          edge [dashed, post] (true0)
          edge [dashed, pre]  (q1)
          edge [dashed, <->, shorten >=1pt, shorten <=1pt]  (false1)
          edge [dashed, post] (r1)
          edge [dashed, pre]  (q2)
          edge [dashed, <->, shorten >=1pt, shorten <=1pt]  (false2)
          edge [dashed, post] (r2)
          ;

        \node [transition, red, thick] (qfrfqtrf) at (1.5, -4) {}
          edge [pre,  red, thick]  (p0)
          edge [pre,  red, thick]  (false0)
          edge [post, red, thick] (q0)
          edge [post, red, thick] (true0)
          edge [pre, red, thick]  (q1)
          edge [<->, shorten >=1pt, shorten <=1pt, red, thick]  (false1)
          edge [post, red, thick] (r1)
          edge [pre,  red, thick]  (q2)
          edge [pre,  red, thick]  (true2)
          edge [post, red, thick]  (false2)
          edge [post, red, thick] (r2)
          ;

        \node [transition] (qtrfp) at (3, -4) {}
          edge [dashed, pre]  (p0)
          edge [dashed, pre]  (false0)
          edge [dashed, post] (q0)
          edge [dashed, post] (true0)
          edge [dashed, pre]  (q1)
          edge [dashed, pre]  (true1)
          edge [dashed, post] (false1)
          edge [dashed, post] (r1)
          edge [dashed, <->, shorten >=1pt, shorten <=1pt]  (p2)
          ;

        \node [transition] (qtrfqfrf) at (4.5, -4) {}
          edge [dashed, pre]  (p0)
          edge [dashed, pre]  (false0)
          edge [dashed, post] (q0)
          edge [dashed, post] (true0)
          edge [dashed, pre]  (q1)
          edge [dashed, pre]  (true1)
          edge [dashed, post] (false1)
          edge [dashed, post] (r1)
          edge [dashed, pre]  (q2)
          edge [dashed, <->, shorten >=1pt, shorten <=1pt]  (false2)
          edge [dashed, post] (r2)
          ;

        \node [transition] (qtrfqtrf) at (6, -4) {}
          edge [dashed, pre]  (p0)
          edge [dashed, pre]  (false0)
          edge [dashed, post] (q0)
          edge [dashed, post] (true0)
          edge [dashed, pre]  (q1)
          edge [dashed, pre]  (true1)
          edge [dashed, post] (false1)
          edge [dashed, post] (r1)
          edge [dashed, pre]  (q2)
          edge [dashed, pre]  (true2)
          edge [dashed, post] (false2)
          edge [dashed, post] (r2)
          ;

      \end{tikzpicture}
    }
  \end{center}
\end{figure}

\begin{definition}
  A parameterized Petri net $\PN = \tuple{\places, \transitions}$ is a
  \emph{barrier crowd} if there is a finite set of transition patterns of the
  form $\tuple{\mathcal{I}, \mathbb{A}}$ such that for every instance $\PN(n)$
  the following condition holds: a pair $\tuple{P, Q}$ is a transition of
  $\PN(n)$ if{}f there exists a pattern $\tuple{\mathcal{I}, \mathbb{A}}$ and
  $i \in [n]$ such that:
  \begin{itemize}
    \item $P\cap(\places\times\set{i})=P_I\times\set{i}$ and
      $Q\cap(\places\times\set{i})=Q_I\times\set{i}$, where
      $\mathcal{I}=\tuple{P_I,Q_I}$.
    \item for every $j\ne{}i$ there is $\tuple{P_A,Q_A}\in\mathbb{A}$ such that
      $P\cap(\places\times\set{j})=P_A\times\set{j}$ and
      $Q\cap(\places\times\set{j})=Q_A\times\set{j}$.
  \end{itemize}
\end{definition}

Note that the number of transitions of $\PN(n)$ grows quickly in $n$, even
though the structure of the system remains simple, making parameterized
verification particularly attractive.

In the rest of the section we present an automatic parametrization procedure
for traps of barrier crowds. First we show that barrier crowds satisfy two
important structural properties.

Given a set of places $P \subseteq \places \times [n]$ and a permutation $\pi
\colon [n] \to [n]$, let $\pi(P)$ denote the set of places $\set{p(\pi(i)):
p(i) \in P}$. Given an index $0 \leq{} k < n$, let $\drop_{k,n}(P)$ denote the
set of places defined as follows: $p(i) \in \drop_{k,n}(P)$ if{}f either $0
\leq i < k$ and $p(i) \in P$, or $k < i \leq n-1$ and $p(i+1) \in P$.

\begin{definition}
  Let $\PN$ be a parameterized Petri net. A transition $\tuple{P_1, P_2}$ of
  $\PN(n)$ is:
  \begin{itemize}
    \item \emph{order invariant} if $\tuple{\pi(P_{1}),  \pi(P_{2})}$ is also a
      transition of $\PN(n)$ for every permutation $\pi \colon [n] \to [n]$.
    \item \emph{homogeneous} if there is an index $0 \leq i < n$ such that for
      every $k \in [n]\setminus\set{i}$ the pair $\tuple{\drop_{k,n}(P_1),
      \drop_{k,n}(P_2)}$ is a transition of $\PN(n-1)$.
  \end{itemize}
  $\PN$ is homogeneous (order invariant) if all transitions of all instances
  $\PN(n)$ is homogeneous (order invariant).
\end{definition}

Intuitively, order invariance indicates that processes are indistinguishable.
Homogeneity indicates that transitions in the large instances are not
substantially different from the transitions in the smaller ones.

\begin{proposition}
  Barrier crowds are order invariant and homogeneous.
\end{proposition}

\begin{proof}
  Let $\PN$ be a barrier crowd. For order invariance, let $\tuple{P, Q}$ be a
  transition of an instance $\PN(n)$, and let $\pi \colon [n] \to [n]$ be a
  permutation. We show that $\tuple{\pi(P),  \pi(Q)}$ is also a transition of
  $\PN(n)$. By the definition of barrier crowds there is a pattern
  $\tuple{\mathcal{I}, \mathbb{A}}$, where $\mathcal{I}=\tuple{P_I,Q_I}$, and
  an index $i$ such that
  \begin{itemize}
    \item $P\cap(\places\times\set{i})=P_I\times\set{i}$ and
      $Q\cap(\places\times\set{i})=Q_I\times\set{i}$; and
    \item for every $j\ne{}i$ there is $\tuple{P_A^j,Q_A^j}\in\mathbb{A}$ such
      that $P\cap(\places\times\set{j})=P_A^j \times\set{j}$ and
      $Q\cap(\places\times\set{j})=Q_A^j \times\set{j}$.
  \end{itemize}
  Intuitively, by the definition of barrier crowds, the result of instantiating
  $\tuple{\mathcal{I}, \mathbb{A}}$ with the index  $\pi(i)$ instead of $i$ is
  also a transition of  $\PN(n)$. Formally, the pair $\tuple{P', Q'}$ given by
  \begin{itemize}
    \item $P'\cap(\places\times\set{\pi(i)})=P_I\times\set{\pi(i)}$ and
      $Q'\cap(\places\times\set{\pi(i)})=Q_I\times\set{\pi(i)}$, and
    \item $P'\cap(\places\times\set{\pi(j)})=P_A^j\times\set{\pi(j)}$ and
      $Q'\cap(\places\times\set{\pi(j)})=Q_A^j\times\set{\pi(j)}$ for every
      $j\ne{}i$
  \end{itemize}
  is a transition of $\PN(n)$. By construction we have $\pi(P) = P'$ and
  $\pi(Q) = Q'$. So $\tuple{\pi(P),  \pi(Q)}$ is a transition of $\PN(n)$, and
  we are done.

  For homogeneity, let $\tuple{P, Q}$ be a transition of $\PN(n)$. Let
  $\tuple{\mathcal{I}, \mathbb{A}}$ with  $\mathcal{I} = \tuple{P_I, Q_I}$ be
  the pattern of which $\tuple{P, Q}$ is an instance, i.e., $P \cap
  (\places\times\set{i}) = P_{I}\times\set{i}$ and $Q \cap
  (\places\times\set{i}) = Q_{I}\times\set{i}$ for every $i \in [n]$. By the
  definition of barrier crowds, for every $j \in [n]\setminus\set{i}$ there is
  $\tuple{P_{I}^{j}, Q^{j}_{I}} \in\mathbb{A}$ such that $P \cap (\places
  \times \set{j}) = P^{j}_{I}$ and $Q \cap (\places \times \set{j}) =
  Q^{j}_{I}$. For every $k \in [n]\setminus\set{i}$, we carefully instantiate
  $\tuple{\mathcal{I}, \mathbb{A}}$ to obtain a transition $\tuple{P', Q'}$
  satisfying $\tuple{P', Q'}= \tuple{\drop_{k,n}(P), \drop_{k,n}(Q)}$, which
  concludes the proof. We need to consider two cases:
  \begin{itemize}
    \item If $i < k$, then:
      \begin{itemize}
        \item $P' \cap (\places\times\set{i}) = P_{I}$ and $Q' \cap
          (\places\times\set{j}) = P^{j}_{I}$ for all $i \neq j < k$;
        \item $P' \cap (\places\times\set{j}) = P^{j}_I$ and $P'\cap
          (\places\times\set{j}) = P^{j-1}_{I}$ for all $k < j$;
        \item $Q' \cap (\places\times\set{i}) = Q_{I}$ and $Q' \cap
          (\places\times\set{j}) = Q^{j}_{I}$ for all $i \neq j < k$; and
        \item  $Q' \cap (\places\times\set{j}) = Q^{j}_{I}$ and $Q'\cap
          (\places\times\set{j}) = Q^{j-1}_{I}$ for all $k <  j$.
      \end{itemize}
    \item If $k < i$, then $\tuple{P', Q'}$ is defined as for the case $i < k$,
      with the exception that now $P' \cap (\places\times\set{i-1}) = P_{I}$
      and $Q' \cap \places\times\set{i-1} = Q_{I}$.
  \end{itemize}
  In both cases $\tuple{P', Q'}$ is a transition in $\PN(n-1)$ by definition of
  barrier crowds. $P' = \drop_{k, n}(P)$ and $Q' = \drop_{k, n}(Q)$ which
  concludes the argument.
\end{proof}

\subsection{Parametrizing traps for barrier crowds}
By order invariance, if $Q$ is a trap of an instance, say $\PN(n)$, then
$\pi(Q)$ is also a trap for every permutation $\pi$. The set of all traps that
can be obtained from $Q$ by permutations can be described as a multiset $\mathcal{Q} \colon
2^\places \rightarrow [n]$. For example, assume $\places=\set{p, q}$, $n=5$,
and $Q= \set{p(0), p(1), q(1), p(2), q(2), q(4)}$. Then $\mathcal{Q}(\set{p,q})
= 2$ (because of indices $1$ and $2$), $\mathcal{Q}(\set{p})
=\mathcal{Q}(\{q\})=1$ (index 0 and 4, respectively), and
$\mathcal{Q}(\emptyset) = 1$ (index 3). Any assignment of indices to the
elements of $\mathcal{Q}$ results in a trap. We call $\mathcal{Q}$ the
\emph{trap family of $Q$}.

\begin{proposition}
\label{prop:crowds}
Let $\PN$ be an order invariant and homogeneous parameterized Petri net, let
  $Q$ be a trap of an instance $\PN(n)$, and let $\mathcal{Q} \colon 2^\places
  \rightarrow [n]$ be the trap family of $Q$. We have:
\begin{itemize}
  \item If $\mathcal{Q}(\emptyset) \geq 1$ and $\mathcal{Q}'$ is obtained from
    $\mathcal{Q}$ by increasing the multiplicity of $\emptyset$, then
    $\mathcal{Q}'$ is also a trap family of another instance of $\PN$.
  \item For every $S \in  2^\places$, if $\mathcal{Q}(S) \geq 2$ and
    $\mathcal{Q}'$ is obtained from $\mathcal{Q}$ by increasing the
    multiplicity of $S$, then  $\mathcal{Q}'$ is also a trap family of another
    instance of $\PN$.
  \end{itemize}
\end{proposition}

\begin{proof}
  First consider increasing the multiplicity of $\emptyset$. It suffices to
  consider the case of the family  $\mathcal{Q}'$ obtained from $\mathcal{Q}$
  by setting $\mathcal{Q}'(\emptyset) = \mathcal{Q}(\emptyset) + 1$, since the
  general statement follows by induction in a straightforward manner. Assume
  that $\tuple{n + 1, Q'}$ is not a trap in $\PN(n+1)$, but has the
  multiplicities of $\mathcal{Q}'$. Let $t'$ be a transition of $\PN(n+1)$ such
  that $Q' \cap \preset{t'} \neq \emptyset = Q' \cap \postset{t'}$. Since
  $\mathcal{Q}'(\emptyset) \geq 2$, there are at least two distinct indices $i,
  k$ such that $Q' \cap (\places\times\set{i, k}) = \emptyset$. By homogeneity
  of $\PN$ we can choose $i, k$ so that  a transition $t$ in $\PN(n)$ satisfies
  $\preset{t} = \drop_{k, n}(\preset{t'})$ and $\postset{t} = \drop_{k,
  n}(\postset{t'})$. Further, let $Q = \drop_{k, n}(Q')$. Note that $Q$ is an
  instance of the trap family $\mathcal{Q}$. However, $\preset{t} \cap Q \neq
  \emptyset = \postset{t} \cap Q$ by the definition of $t$ and $\drop_{k, n}$
  in contradiction to $\mathcal{Q}$ being a trap family.

  In the case of increasing the multiplicity of a non-empty set $S$ we know
  that $\mathcal{Q}(S) \geq 2$ and $\mathcal{Q'}(S) = \mathcal{Q}(S) + 1 \geq
  3$. The argument is analogous to the previous case. First we assume
  $\tuple{n+1, Q'}$ is an instance of $\mathcal{Q'}$ that is not a trap. For
  $Q'$, let $k_{1}, k_{2}, k_{3}$ be three distinct indices in $[n+1]$ such
  that $Q'\cap \places\times\set{k_{1}, k_{2}, k_{3}} = S \times \set{k_{1},
  k_{2}, k_{3}}$. Then, we find a transition $t'$ in $\PN(n+1)$ witnessing that
  $Q'$ is not a trap. We consider two cases:
  \begin{itemize}
    \item $\preset{t'} \cap S \times \set{k_{1}, k_{2}, k_{3}} = \emptyset$. \\
      By homogeneity, there is $k \in \set{k_{1}, k_{2}, k_{3}}$ such that the
      result of applying the $\drop_{k, n}$ operation to $t'$ is a transition
      $t$ of $\PN(n)$. The set $Q = \drop_{k, n}(Q')$ is a set of places of
      $\PN(n)$ with trap family $\mathcal{Q}$. We have $\preset{t} \cap Q \neq
      \emptyset$ since $\preset{t'} \cap Q' \neq \emptyset$ and $Q' \cap
      (\places \times \set{k}) = \emptyset$; further, $\postset{t} \cap Q =
      \emptyset$ since $\postset{t'} \cap Q' = \emptyset$. So $Q$ is not a
      trap, contradicting the assumption.

    \item $\preset{t'} \cap S \times \set{k_{1}, k_{2}, k_{3}} \neq \emptyset$.
      \\
      By homogeneity there are $k_{1}', k_{2}' \in \set{k_{1}, k_{2}, k_{3}}$
      such that the result of applying $\drop_{k_{1}', n}$ and $\drop_{k_{2}',
      n}$ to $t'$ are two transitions $t_{1}$ and $t_{2}$ of $\PN(n)$. Since
      $\postset{t'}\cap Q' = \emptyset$ we also have $\postset{t_{1}}\cap Q =
      \emptyset$ and $\postset{t_{1}} \cap Q = \emptyset$. Let $k_{o}$ denote
      the only element in $\set{k_{1}, k_{2}, k_{3}} \setminus \set{k_{1}',
      k_{2}'}$. If $S \times \set{k_{o}} \cap \preset{t'} \neq \emptyset$ then
      $\preset{t_{1}} \cap Q \neq \emptyset$ and $\preset{t_{1}} \cap Q \neq
      \emptyset$ because $k_{o}$ is not dropped. If, on the other hand, $S
      \times \set{k_{o}} \cap \preset{t'} = \emptyset$, then either $S \times
      \set{k_{1}} \cap \preset{t'} \neq \emptyset$ or $S \times \set{k_{2}}
      \cap \preset{t'} \neq \emptyset$. By symmetry we can assume w.l.o.g. $S
      \times \set{k_{1}} \cap \preset{t'} \neq \emptyset$. Then $\preset{t_{2}}
      \cap Q \neq \emptyset$. So $Q$ is not a trap although its trap family is
      $\mathcal{Q}$, which contradicts the assumption.
  \end{itemize}
\end{proof}

Proposition \ref{prop:crowds} leads to a parameterization procedure for barrier
crowds. Given a trap $Q$ of some instance $\PN(n)$ and its trap family
$\mathcal{Q}$, consider all multisets obtained from $\mathcal{Q}$ by applying
the operations of Proposition \ref{prop:crowds}. We call this set of multisets
the \emph{extended trap family} of $Q$. Observe that $\mathcal{Q}$ represents a
set of traps of $\PN(n)$, while the extended family represents a set of traps
across all instances $\PN(n')$ with $n' \geq n$.

Give an indexed trap $\interp{T}=\tuple{n, Q}$,  we choose the formula
$\Param_\interp{T}(\Xv)$ so that its models correspond to the traps of the
extended family of $Q$. For this, we capture the minimal required
multiplicities of $\tuple{n, Q}$ by quantifying for every $S\subseteq \places$
with $\mathcal{Q}(S) > 0$ indices $\bm{i}_{S, 1}, \ldots,
\bm{i}_{S, \mathcal{Q}(S)}$ for which precisely the places in $S$ are marked.
Making all indices introduced this way pairwise distinct ensures that any model
of the formula at least covers the multiset $\mathcal{Q}$. Additionally, we can
capture that the subset $S$ of $\places$ which are marked in every other index
are chosen such that Proposition~\ref{prop:crowds} ensures that we still obtain
a trap.

\begin{equation*}
  \begin{aligned}
    \Param_\interp{T}(\bm{n}, \Xv) \coloneqq
    \exists_{S\subseteq\places}
      \bm{i}_{S,1}, \ldots, \bm{i}_{S,\mathcal{Q}(S)} \colon \Bigg\{
          \left(
          \bigwedge_{(S,k)\ne(S',k')}
          (
          \bm{i}_{S,k}\ne{}\bm{i}_{S',k'}
          )
          \right)
          \wedge \\
          \forall \bm{j} \colon \bm{j} < \bm{n} \rightarrow
          \Bigg[
            \left(
            \bigvee\limits_{S\subseteq\places, k = 1, \ldots, \mathcal{Q}(S)}
            \left(
            \bm{j}=\bm{i}_{S,k} \wedge
          \left(
            \begin{aligned}
              &\bigwedge_{p \in S} \bm{j} \in \Xv_{p}\\
              \land &\bigwedge_{p \in \places\setminus S} \bm{j} \notin \Xv_{p}
            \end{aligned}
          \right)
            \right)
            \right)
            \vee \\
            \left(
            \left(\bigwedge_{S\subseteq\places,k = 1, \ldots, \mathcal{Q}(S)}
              \bm{j}\ne{}\bm{i}_{S,k}\right)
            \wedge
              \bigvee_{
                \substack{
                  \emptyset \neq S\subseteq\places\colon\mathcal{Q}(S)\ge{}2\\
                  S = \emptyset\colon\mathcal{Q}(S)\ge{}1}
                }
                \left(
                  \begin{aligned}
                    &\bigwedge_{p \in S} \bm{j} \in \Xv_{p}\\
                    \land &\bigwedge_{p \in \places\setminus S} \bm{j} \notin \Xv_{p}
                  \end{aligned}
                \right)
            \right)
            \Bigg]
          \Bigg\}.
  \end{aligned}
\end{equation*}

We immediately get:

\begin{theorem}
  \label{thm:crowds}
  Let $\PN=\tuple{\places, \transitions}$ be a barrier crowd and let
  $\tuple{n, Q}$ be a local indexed trap of $\PN(n)$. Then every model of the
  formula $\Param_\interp{T}(\bm{n}, \Xv)$ defined above is an indexed trap of
  $\PN$.
\end{theorem}

\begin{remark}
  This theorem applies to all order invariant and homogeneous systems. It is
  easy to see that order invariance and homogeneity of a given parameterized
  net can be expressed in WS1S and verified automatically.
\end{remark}

%% file: tex/generalized-ppn.tex
\section{Generalized Parameterized Petri nets}
\label{sec:gppn}
Recall that the set of places of a parameterized net has the form $\places
\times [n]$, i.e., $n$ copies of the set $\places$ of place names. Intuitively,
the net $N_n$ consists of $n$ communicating processes, each of them with (a copy
of) $\places$ as states.

Unfortunately, this setting is not powerful enough to model many classical
distributed algorithms. In particular, it cannot model any parameterized mutual
exclusion algorithm, like Dekker's, Dijkstra's, Knuth's and others
\cite{CooperatingSequentialProcesses}. The reason is that in all these
algorithms agents need to execute loops in which they inspect the current value
of a flag in all other agents. We explain this point in detail taking
Dijkstra's mutual exclusion algorithm as an example.

\begin{example}
  \label{ex:dijkstra-spec}
  In Dijkstra's algorithm, each agent maintains a boolean variable $\flag$ that
  indicates whether the agent wants to access the critical section or not. The
  flag is initially set to $\falseval$. At any moment the agent $i$ may set it
  to $\trueval$, after which agent $i$ iteratively inspects the $\flag$
  variable of the other agents. Crucially, the inspection takes $n$ atomic
  steps, one for each agent. If some flag has value $\trueval$, then the  agent
  sets its $\flag$ to $\falseval$ and starts over; if all flags have value
  $\falseval$ (at the respective times at which the agent inspects them), then
  the agent moves to the critical section. If we assume that agents have
  identities $0, 1, \ldots, n-1$, then agent $i$ can be modeled by the code
  shown in Figure \ref{fig:dijkstra}.

\begin{figure}
\begin{minipage}{\textwidth}
\begin{lstlisting}[language=C]
init: flag[i] = true;
for(j = 0; j < n; j++) {
  if(i != j and flag[j] == true) {flag[i] = false; goto init;}
}
/* critical section */
restart: flag[i] = false; goto init;
\end{lstlisting}
\end{minipage}
\caption{Pseudocode for Dijkstra's mutual exclusion algorithm for agent $i$.}
\label{fig:dijkstra}
\end{figure}

  When agent $i$ inspects the flag of agent $j$, we say that agent $i$
  \emph{points to} agent $j$. Assume that when an agent is not executing the
  loop the variable $j$ has a special value $\bot$ (points to null). At every
  moment in time the local state of each agent is determined by its current
  position in the code, the value of its flag, and the agent it is pointing to
  (or null). We distinguish six positions: $\stateinitial$ (corresponds to the
  label \lstinline{init} above); $\stateloop$ (before the \lstinline{for}
  loop); $\statelooping$ (before the \lstinline{if} statement);  $\statebreak$
  (before the body of the \lstinline{if} statement); $\statecrit$ (critical
  section); and $\statedone$ (label \lstinline{restart}). The flag has two
  values, and $j$ can have $n+1$ different values.

  In the Petri net model for an instance of Dijkstra's algorithm with $n$
  agents, each agent is assigned six places for the positions, two places for
  the values of $\flag$, and $n+1$ places for the values of $j$. The net has a
  total of $n \cdot (8 + n)$ places. The crucial difference with respect to the
  net for the philosophers is that the number of places \emph{per agent}
  depends on $n$.
\end{example}

Since some places now involve two agents (the places indicating that agent $i$
points to agent $j$), sets of places must be modeled as relations, which leads
to the definition of generalized parameterized nets:

\begin{definition}[Generalized Parameterized nets]
\label{def:gpn}
  A generalized parameterized net is a triple $\PN = \tuple{\places,
  \relations, \transitions}$, where $\places$ is a finite set of place names,
  $\relations$ is a finite set of relation names, and $\transitions(\bm{n},
  \Xv, \Uv, \Yv, \Sv)$ is a second-order formula over one first-order variable
  $\bm{n}$ which represents the considered size of the instance, a tuple $\Xv$
  and $\Yv$ of monadic second-order variables for each place name of $\places$,
  a tuple $\Uv$ and $\Sv$ of dyadic second-order variables for each relation
  name of $\relations$. Moreover, we require $\transitions$ to only quantify
  first-order and \emph{monadic} second-order variables.

  For every $n \geq 1$, the \emph{$n$-instance} of $\PN$ is the net $\PN(n) =
  \tuple{P_{n}, T_{n}}$ given by:
  \begin{itemize}
    \item $P_{n} = (\places \times [n])  \cup \left( \relations \times [n]
      \times ([n] \cup \set{\bot}) \right)$,
  \item $\tuple{P_{1} \cup R_{1}, P_{2} \cup R_{2}} \in T_{n}$ if and only if
    \enquote{$\transitions(n, P_{1}, R_{1}, P_{2}, R_{2})$} holds;
    i.e., if the interpretation $\mu$ given by
    \begin{equation*}
    \begin{aligned}
      &\mu(\Xv_{p}) = \set{i \in [n] \mid \tuple{p, i} \in P_{1}},\\
      &\mu(\Uv_{r}) = \set{\tuple{i, j} \in [n] \times \left( [n] \cup
      \set{\bot} \right) \mid \tuple{r, i, j}
        \in P_{1}},\\
          &\mu(\Yv_{p}) = \set{i \in [n] \mid \tuple{p, i} \in P_{2}},
          \text{ and }\\
      &\mu(\Sv_{r}) = \set{\tuple{i, j} \in [n] \times \left( [n] \cup
      \set{\bot} \right) \mid \tuple{r, i, j}
        \in P_{2}}.
    \end{aligned}
    \end{equation*}
      is a model of $\transitions$.
\end{itemize}
\end{definition}

We now define generalized parameterized Petri nets
\begin{definition}[Generalized parameterized Petri nets]
  A generalized Petri net is a pair $\tuple{\PN, \markings}$ where
  \begin{itemize}
    \item $\PN = \tuple{\places, \relations, \transitions}$ is a generalized
      net, and
    \item $\markings(\bm{n}, \Xv, \Uv)$ is a second-order formula over the
      first-order variable $\bm{n}$, monadic second-order variables $\Xv$,
      and dyadic second-order variables $\Uv$.
  \end{itemize}
  As for $\transitions$ before, we restrict $\markings$ to only quantify
  first-order variables and monadic second-order variables.
\end{definition}
As before, we consider $\tuple{\PN(n), M}$ to be the marked $n$-instance if
\enquote{$\markings(n, M)$ is true}; that is $M$ is a 1-safe marking of
$\PN(n)$ for which $\mu$ with $\mu(\bm{n}) = n$ and $\mu(\Xv_{p}) = \set{i \in
[n] \mid M(p(i)) = 1}$ and $\mu(\Uv_{r}) = \set{\tuple{i, j} \in [n] \cup
\left( [n] \cup \set{\bot} \right) \mid M(r(i, j)) = 1}$ satisfies $\mu \models
\markings$.

\begin{example}
  \label{ex:dijkstra}
  We model Dijkstra's mutual exclusion algorithm as a generalized parameterized
  Petri net $\tuple{\PN, \markings}$ where $\PN = \tuple{\places, \relations,
  \transitions}$. We define

  \begin{equation*}
    \places = \set{\stateinitial, \stateloop, \statelooping, \statebreak,
    \statecrit, \statedone} \cup \set{\idled, \trying}
    \qquad
    \relations = \set{\ptr}
  \end{equation*}

  The set $\places$ contains the positions in the code plus the elements
  $\idled$ and $\trying$, which are abbreviations for \enquote{the value of
  $\flag$ is $\trueval$}, and  \enquote{the value of $\flag$ is $\falseval$},
  respectively. The relation symbol $\ptr$ (for pointing) indicates which agent
  is pointing to which one. The set of places of the instance $\PN_n$ is
  $(\places \times [n]) \cup (\set{\ptr} \times [n] \times ([n] \cup
  \set{\bot})$.

  We only describe one part of the formula $\transitions$ modeling the
  transitions that correspond to some agent \emph{successfully} advancing in
  its loop, either because the agent inspects itself, or because another agent
  that has not set its $\flag$ variable to $\trueval$. For every ${\cal P}'
  \subseteq \places$, we first introduce an auxiliary formula indicating that a
  set of places only contains places of type  $\places' \subseteq \places$.
  \begin{equation*}
    \allfrom{{\cal P}'} = \bigwedge_{p \notin {\cal P}'}
    \Xv_{p} = \Yv_{p} = \emptyset
  \end{equation*}

  \begin{table}[t]
  \caption{Part of the formula $\transitions$ for Dijkstra's mutual exclusion
    algorithm.}
  \label{formula:dijkstra}
  \begin{equation*}
    \begin{aligned}
      \exists \bm{i}  ~.~ \exists \bm{j}  ~.~
      &                && \bm{i} < \bm{n}  \; \land \;\bm{j} \neq  \bot \land
                          1 \leq \bm{j} < \bm{n} \ \\[0.2cm]
      & {\ } \land &&  \Xv_{\statelooping} = \set{\bm{i}}  \land
                       \Uv_{\ptr} = \set{\tuple{\bm{i}, \bm{j}}}   \\[0.2cm]
      & \land && \left(\begin{aligned}
        &\bm{i} = \bm{j} < \bm{n} - 1 \land
        \left[\begin{aligned}
          &\Sv_{\ptr} = \set{\tuple{\bm{i}, \bm{j} + 1}}  \\
          {\ } \land {\ } & \Yv_{\statelooping} = \set{\bm{i}}
        \end{aligned}\right]
        && \land
        \allfrom{\set{\statelooping}}
        \\
        {\ } \lor {\ \ } &\bm{i} \neq \bm{j} < \bm{n} - 1 \land
        \left[\begin{aligned}
          &\Sv_{\ptr} = \set{\tuple{\bm{i}, \bm{j} + 1}}\\
          {\ } \land {\ } & \Xv_{\idled} = \set{\bm{j}} \\
         {\ } \land {\ } &\Yv_{\idled} = \set{\bm{j}} \\
          {\ } \land {\ } &\Yv_{\statelooping} = \set{\bm{i}}
        \end{aligned}\right]
        && \land
        \allfrom{\set{\statelooping, \idled}}
        \\
         {\ } \lor {\ \ } &\bm{i} = \bm{j} = \bm{n} - 1 \land
        \left[\begin{aligned}
          &\Yv_{\statecrit} = \set{\bm{i}}\\
          {\ } \land {\ } &\Yv_{\statelooping} = \emptyset\\
          {\ } \land {\ } &\Sv_{\ptr} = \set{\tuple{\bm{i}, \bot}}
        \end{aligned}\right]
        && \land
        \allfrom{\set{\statecrit, \statelooping} }
        \\
       {\ } \lor {\ \ } &\bm{i} \neq \bm{j} = \bm{n} - 1 \land
        \left[\begin{aligned}
          &\Yv_{\statecrit} = \set{\bm{i}}\\
          {\ } \land {\ } & \Yv_{\statelooping} = \emptyset\\
          {\ } \land {\ } & \Xv_{\idled} = \set{\bm{j}}\\
          {\ } \land {\ } & \Yv_{\idled} = \set{\bm{j}} \\
          {\ } \land {\ } &\Sv_{\ptr} = \set{\tuple{\bm{i}, \bot}}
        \end{aligned}\right]
        && \land
        \allfrom{\set{\begin{aligned}
          &\statecrit, \idled, \\
          &\statelooping
        \end{aligned}}}
      \end{aligned}\right)
    \end{aligned}
  \end{equation*}
  \end{table}
  The transitions are given by the formula shown in
  Equation~\ref{formula:dijkstra}. The formula states that there is a
  transition $\tuple{P_{1} \cup R_{1}, P_{2} \cup R_{2}}$ if
  \begin{itemize}
    \item $P_1 = \set{\statelooping(i)}$, $R_1 = \set{ \ptr(i, j)}$, $P_2 =
      \set{\statelooping(i)}$, and  $R_2 = \set{ \ptr(i, j+1)}$ for some
      $i = j < n-1$; or
    \item $P_1 = \set{\statelooping(i), \idled(j)}$, $R_1 = \set{\ptr(i, j)}$,
      $P_2 = \set{\statelooping(i), \idled(j)}$, and, for some $i \neq j <
      n-1$, $R_2 = \set{\ptr(i, j + 1)}$; or
    \item $P_1 = \set{\statelooping(i)}$, $R_1 = \set{\ptr(i, j)}$, $P_2 =
      \set{\statecrit(i)}$, and  $R_2 = \set{\ptr(i, \bot)}$ for $i = j = n-1$;
      or,
    \item for some $i \neq j = n-1$, $P_1 = \set{\statelooping(i), \idled(j)}$,
      $R_1 = \set{ \ptr(i, j)}$, $P_2 = \set{\statecrit(i),\idled(j)}$, and
      $R_2 = \set{ \ptr(i, \bot)}$.
  \end{itemize}

  Finally, the initial markings are given by:
  \begin{equation*}
    \markings(\bm{n}, \Xv, \Uv) = \left(
    \begin{aligned}
      &\Xv_{\stateinitial} = [\bm{n}] \land \Xv_{\idled} =
      [\bm{n}]\\
      \land &\bigwedge_{p \in \places \setminus \set{\stateinitial,
      \idled}} \Xv_{p} = \emptyset\\
      \land &\forall \bm{i} ~.~ \bm{i} < \bm{n} \rightarrow
      \left[ \forall \bm{j} ~.~ \tuple{\bm{i}, \bm{j}} \in \Uv_{\ptr}
      \leftrightarrow \bm{j} = \bot \right]
    \end{aligned}
    \right)
  \end{equation*}
\end{example}

\subsection{A CEGAR approach}
In Section \ref{sec:param-cegar} we have described a CEGAR loop for the
analysis of parameterized Petri nets.  We extend the approach to generalized
parameterized Petri nets.

Recall that the CEGAR loop takes a parameterized Petri net $\tuple{\PN,
\markings}$, and a safety property described by a formula $\Safe(\bm{n}, \Mv)$,
as inputs. The loop maintains a set $\trapset$ of traps of the instances of
$\PN$, initially empty. In every iteration in the loop, the procedure first
constructs the formula $\Param_\interp{T}(\bm{n}, \Xv)$, then the formula
$\Safety_{\trapset}$ of (\ref{formula:safetycheck}), and
then sends $\Safety_{\trapset}$ to a \acs{WS1S}-checker. If the checker returns
that $\Safety_{\trapset}$ holds, then every instance of $\PN$ satisfies the
safety property, and the loop terminates. Otherwise, the checker returns a
model $\interp{M} = \tuple{n, M}$ of the formula $\PotReach_\trapset(\bm{n},
\Mv) \wedge \neg \Safe(\bm{n}, \Mv)$, and searches for a witness trap $\Xv$
that is marked at every initial marking, but empty at $\interp{M}$, with the
help of the SAT-formula $\WitTrap_\interp{M}(n, \Xv)$ defined in
(\ref{formula:wtrap}).

Let us see how to extend the CEGAR loop to generalized parameterized Petri
nets. We assume again that the generalized parameterized Petri net belongs to a
class with a special topology that allows one to compute the formula
$\Param_\interp{T}$ for a given trap $T$. The obstacle is that the formula
$\Safety_{\trapset}$ no longer belongs to \acs{WS1S}. Indeed, since the places
of a generalized Petri net are of the form $p(i)$ or $r(i, j)$, in
(\ref{formula:safetycheck}) we have to add to the placeset parameters $\Xv$ and
$\Mv$ relationset parameters $\Uv$ and $\Lv$, i.e., sequences of dyadic
predicate symbols, one for each relation in $\relations$. For example,
$\Param_\interp{T}(\bm{n}, \Xv)$ becomes $\Param_\interp{T}(\bm{n}, \Xv, \Uv)$.

While the extension of \acs{WS1S} with dyadic predicates is no longer
decidable, we show that the problem of checking if $\Safety_{\trapset}$ is true
can be reduced to the validity problem of first-order logic when
$\Safety_{\trapset}$ is a universal formula, i.e., a formula in prenex normal
form in which a block of universal second-order quantifiers is followed by a
first-order formula. In this case it is easy to construct a first-order formula
FO($\Safety_{\trapset}$) such that $\Safety_{\trapset}$ is true if{}f
FO($\Safety_{\trapset}$) is valid. This allows us to use an automatic
first-order theorem prover to check validity of FO($\Safety_{\trapset}$), and,
therefore, the truth of $\Safety_{\trapset}$. The price to pay is that, if
$\Safety_{\trapset}$ does not hold, then the checker no longer returns a model
of $\PotReach_\trapset(\bm{n}, \Mv) \wedge \neg \Safe(\bm{n}, \Mv)$; instead,
the checker just does not terminate (in practice, it reaches a timeout). For
this reason, we replace the CEGAR loop by the following one, consisting of two
communicating processes:
\begin{itemize}
  \item Process 1 iteratively constructs the Petri nets $\PN(1)$, $\PN(2)$,
    $\PN(3)$, \dots, and uses  the CEGAR loop for ordinary Petri nets (see
    Section \ref{sec:non-param-cegar}) to compute sets $\trapset_1$,
    $\trapset_2$, $\trapset_3$, \dots of traps proving that $\PN(1)$, $\PN(2)$,
    $\PN(3)$, \dots, satisfy the safety property. Whenever this process
    computes a new trap, it passes it to Process 2.
  \item Process 2 maintains a set $\trapset$ of traps, initially empty. First,
    for every $T \in \trapset$ it constructs the formula
    $\Param_\interp{T}(\bm{n}, \Xv, \Uv)$. It then constructs the formula
    FO($\Safety_{\trapset}$), and passes to a first-order theorem prover. If
    the prover returns that FO($\Safety_{\trapset}$) is valid, then
    $\Safety_{\trapset}$ is true, and the safety property holds for all
    instances. If the prover reaches a timeout, then Process 2 waits for
    Process 1 to send a new trap $T$, adds $T$ to $\trapset$, and iterates.
\end{itemize}
In order to complete the description of this procedure we must explain 
\begin{itemize}
  \item How to construct FO($\Safety_{\trapset}$). For \emph{universal}
    formulas $\Safety_{\trapset}$  this is a standard syntax-guided procedure,
    that we sketch below.
  \item How to construct $\Param_\interp{T}(\bm{n}, \Xv, \Uv)$. In the next
    section we do this for \emph{inspection programs}, a topology in which we
    can model Dijkstra's algorithm and other mutual exclusion algorithms.
\end{itemize}

\paragraph{The formula FO($\Safety_{\trapset}$).}
A formula of WS1S with dyadic predicates is \emph{universal} if it is in
prenex normal form and has the form $\varphi := \forall \bm{X_1} \ldots
\bm{X_n} \; \forall \bm{U_1} \ldots \bm{U_m} \; \widehat{\varphi}$, where the
$\bm{X_i}$ and $\bm{U_j}$ are monadic and dyadic predicates, respectively, and
$\widehat{\varphi}$ does not contain any second-order quantifiers, i.e.,
$\widehat{\varphi}$ is a formula over the syntax

\begin{equation*}
  \begin{array}{rcll}
    t & := & \bm{x} \mid 0 \mid \logicnext(t) & \quad \\[0.1cm]
    \varphi & := & t_1 \leq t_2 \mid \bm{x} \in \bm{X} \mid \tuple{\bm{x}, \bm{y}} \in \bm{U} \mid \varphi_1 \wedge \varphi_2
    \mid \neg\varphi_1 \mid \exists \bm{x} \colon \varphi
  \end{array}
\end{equation*}

We describe the folklore result that for any universal sentence $\varphi$ there
is a formula FO($\varphi$) of first-order logic that is valid if{}f $\varphi$
holds. Applying the result to $\Safety_{\trapset}$ we obtain the desired
formula FO($\Safety_{\trapset}$). The signature of FO($\varphi$) replicates the
syntax above: it contains two constant symbols $0$ and $N$, a unary function
symbol $\logicnext$, a binary predicate $\leq$, a monadic predicate
$\textit{InX}_i(\bm{x})$ for every monadic second-order variable $\bm{X_i}$,
and a dyadic predicate $\textit{InU}_j(\bm{x},\bm{y})$ for every dyadic
second-order variable $\bm{U_j}$. FO($\varphi$) is of the form $\psi_0
\rightarrow \psi[\widehat{\varphi}]$. The sentence $\psi_0$ ensures that $\leq$
is a discrete linear order with minimal element $0$ and maximal element $N$,
that $\logicnext$ is irreflexive, injective, and respects $\leq$, i.e.,
$\forall x \forall y \colon x \leq y \rightarrow \logicnext(x) \leq
\logicnext(y)$. The formula $\psi[\widehat{\varphi}]$ is defined inductively on
the structure of $\widehat{\varphi}$ as follows:
\begin{itemize}
  \item if $\widehat{\varphi} = \bm{x}, 0, \logicnext(t), t_1 \leq t_2$, then
    $\psi[\widehat{\varphi}] = \widehat{\varphi}$.
  \item if $\widehat{\varphi} = \bm{x} \in \bm{X}, \tuple{\bm{x}, \bm{y}} \in
    \bm{U}$ then $\psi[\widehat{\varphi}] = \textit{InX}(\bm{x}),
    \textit{InU}(\bm{x},\bm{y})$, respectively.
  \item if $\widehat{\varphi} = \neg \widehat{\varphi}_1, \widehat{\varphi}_1
    \wedge \widehat{\varphi}_2$, then $\psi[\widehat{\varphi}] = \neg
    \psi[\widehat{\varphi}_1], \psi[\widehat{\varphi}_1] \wedge
    \psi[\widehat{\varphi}_2]$, respectively.
\end{itemize}
Intuitively, $\psi[\widehat{\varphi}]$ is the result of dropping all universal
second-order quantifiers from $\varphi$, and interpreting the set membership
symbol $\in$ as a unary or binary predicate. For example, we have

\begin{equation*}
  \begin{array}{rcll}
    \varphi & =  &                            & \forall \bm{U} \, \forall
    \bm{X} \, \forall \bm{x} \, \exists\bm{y} \colon (\bm{x} \in \bm{X} \land
    \tuple{\bm{x}, \bm{y}} \in \bm{U}) \rightarrow \bm{y} \in \bm{X} \\
    \text{FO}(\varphi) & = & \psi_0 \rightarrow & \forall \bm{x} \, \exists
    \bm{y} \colon (\textit{InX}(\bm{x}) \land \textit{InU}(\bm{x}, \bm{y}))
    \rightarrow \textit{InX}(\bm{y}) \ \ .
  \end{array}
\end{equation*}
It is easy to see that if \text{FO}($\varphi$) is valid, then $\varphi$ holds.
For the other direction one observes that if $\varphi$ holds then
$\psi[\widehat{\varphi}]$ holds in every model satisfying $\psi_0$.


In our implementation of this procedure we construct the formula
FO($\Safety_{\trapset}$) directly using the specific topology of looping
programs without taking the detour through this translation. This also allows
us to represent FO($\Safety_{\trapset}$) in a form more suitable for a
first-order theorem prover. This direct construction is presented in Appendix
\ref{appendix:FO}.

\subsection{Inspection programs}
The loop of Dijkstra's algorithm in which an agent inspects the local state of
all other agents is a quite general construction at the core of many other
distributed algorithms \cite{Lynch96,HerlihyS08}. We introduce \emph{inspection
programs}, a topology tailored for describing these algorithms. In inspection
programs, agents maintain a local copy of a set of variables with finite
domain. For this, we assume that $\places$ is partitioned into a number of sets
$\places_{0}, \places_{1}, \ldots, \places_{k}$, each set representing the
range of one variable. This means means that $\places_{j} \times \set{i}$ is a
1BB-set in every $\PN(n)$ for all $1 \leq j \leq k$ and $i \in [n]$, since the
copy of every variable for each agent holds exactly \emph{one} value at every
moment in time. Further, we assume that there exists one distinct set of
\emph{states}, which we call $\states$ in the following. As in the previous
topologies, transitions are generated by transition patterns, in this case two
different ones: \emph{local} and \emph{loop} patterns.

\paragraph{Local patterns.} A local pattern
is of the form
\begin{equation*}
  \label{eq:local-transition}
  \tuple{\origin, \tuple{P_{1}, P_{2}}, \target}
\end{equation*}
where $\origin, \target \in \states$ and $P_{1}, P_{2} \subseteq \places$.
Roughly speaking, the pattern specifies that an agent can change its state from
$\origin$ to $\target$, simultaneously, changing its copies of the places of
$P_{1}$ to $P_{2}$. More formally, in the instance $\PN(n)$ the pattern
generates for every $i \in [n]$ a transition $\tuple{Q_{1}, Q_{2}}$ with
\begin{equation*}
  Q_{1} = \left(\set{\origin} \cup P_{1}\right) \times \set{i}
  \text{ and }
  Q_{2} = \left(\set{\target} \cup P_{2}\right) \times \set{i}.
\end{equation*}
The pattern is called \emph{local} because it involves only the places of
one particular index.

\paragraph{Loop patterns.}
Loop patterns contain four loop states, called $\origin$, $\loopingState$,
$\target_{\success}$, and $\target_{\fail}$. Intuitively, an agent initiates
the loop by moving from $\origin$ to $\loopingState$, and exits it by moving
from $\loopingState$ to $\target_{\success}$ or $\target_{\fail}$. The agent
moves to $\target_{\fail}$ whenever the agent being currently inspected fails
the inspection, and to $\target_{\success}$ if all agents pass the inspection.
Further, a relation $\loopingRelation \in \relations$ maintains the agent that
is being currently inspected. Finally, the condition being inspected is modeled
by two sets $P, \overline{P} \subseteq \places$: if the inspected agent is
currently occupying a state of $P$ resp. $\overline{P}$, then the agent passes
resp. fails the inspection. Additionally, we enforce that $\loopingState$ and
$\loopingRelation$ occur nowhere else in any pattern. Formally, a loop pattern
is of the form

\begin{equation*}
  \label{eq:iterating-transition}
  \tuple{\origin, \loopingState, \loopingRelation, P, \overline{P},
  \target_{\success}, \target_{\fail}}
\end{equation*}
where $\origin, \loopingState, \target_{\success}, \target_{\fail} \in \states$
while $P, \overline{P} \subseteq \places$ and $\loopingRelation \in
\relations$. For every agent $i \in [n]$, the pattern generates several
transitions within the instance $\PN(n)$:
\begin{itemize}
  \item A transition $\tuple{Q_{1}, Q_{2}}$, modeling the start of the loop,
    given by:
    \begin{equation*}
      Q_{1} = \set{\tuple{\origin, i}, \tuple{\loopingRelation, i, \bot}}
      \text{ and }
      Q_{2} = \set{\tuple{\loopingState, i}, \tuple{\loopingRelation, i, 0}}.
    \end{equation*}
  \item A transition $\tuple{Q_{1}, Q_{2}}$, modeling that the agent does not
    inspect itself, given by
    \begin{equation*}
      Q_{1} = \set{\tuple{\loopingState, i}, \tuple{\loopingRelation, i, i}}
      \text{ and }
      Q_{2} = \set{\tuple{\loopingState, i}, \tuple{\loopingRelation, i, i+1}}
    \end{equation*}
    for every $i \in [n-1]$, and
    \begin{equation*}
      Q_{1} = \set{\tuple{\loopingState, i}, \tuple{\loopingRelation, i, i,}}
      \text{ and }
      Q_{2} = \set{\tuple{\target_{\success}, i},
      \tuple{\loopingRelation, i, \bot}}
    \end{equation*}
    if $i = n-1$.
  \item A transition $\tuple{Q_{1j}, Q_{2j}}$ for every agent $j$, modeling a
    successful inspection of agent $j$, given by:
    \begin{equation*}
      Q_{1j} = \set{\tuple{\loopingState, i}, \tuple{\loopingRelation, i, j}}
      \cup P \times \set{j}
      \text{ and }
      Q_{2j} = \set{\tuple{\loopingState, i}, \tuple{\loopingRelation, i, j+1}}
      \cup P \times \set{j}
    \end{equation*}
    if $j \in [n-1]$, and
    \begin{equation*}
      Q_{1j} = \set{\tuple{\loopingState, i}, \tuple{\loopingRelation, i, j}}
      \cup P \times \set{j}
      \text{ and }
      Q_{2j} = \set{\tuple{\target_{\success}, i},
      \tuple{\loopingRelation, i, \bot}} \cup P \times \set{j}
    \end{equation*}
    if $j = n-1$.
  \item A transition $\tuple{Q_{1j}, Q_{2j}}$ for every agent $j$, modeling an
    unsuccessful inspection of agent $j$, given by:
    $\overline{P}$:
    \begin{equation*}
      Q_{1j} = \set{\tuple{\loopingState, i}, \tuple{\loopingRelation, i, j}}
      \cup \overline{P} \times \set{j}
      \text{ and }
      Q_{2j} = \set{\tuple{\target_{\fail}, i}, \tuple{\loopingRelation, i,
      \bot}} \cup \overline{P} \times \set{j}
    \end{equation*}
    for every $j \in [n]$.
\end{itemize}
\begin{example}
  \label{ex:top-dijkstra}
  The transitions of Dijkstra's mutual exclusion algorithm
  (Example~\ref{ex:dijkstra}) are generated by the four local transition
  patterns
  \begin{center}
    \begin{tabular}{llll}
      $\langle \stateinitial,$ & $\set{\idled},$ & $\set{\trying},$ & $\stateloop \rangle$\\
      $\langle \statebreak,$ & $\set{\trying},$ & $\set{\idled},$ & $\stateinitial \rangle$\\
      $\langle \statedone,$ &  $\set{\trying},$ & $\set{\idled},$ & $\stateinitial \rangle$\\
      $\langle \statecrit,$ & $\emptyset,$ & $\emptyset,$ & $\statedone \rangle$\\
    \end{tabular}
  \end{center}
  and the single loop transition pattern
  \begin{equation*}
    \tuple{\stateloop, \statelooping, \ptr, \set{\idled}, \set{\trying},
    \statecrit,  \statebreak}.
  \end{equation*}
  We depict parts of the instance $\PN(3)$ in
  Figure~\ref{fig:dijkstra-instance-3}.
\end{example}
It is straightforward to see that every inspection program can be modeled by a
generalized Petri net. Notice that we already gave a part of $\transitions$ in
Example~\ref{ex:dijkstra} such that, when instantiated in some instance
$\PN(n)$, these transitions coincide with the transitions induced by
the loop pattern above when the loop is advanced.

\begin{figure}
  \caption{We illustrate here parts of $\PN(3)$ for the generalized Petri net
  $\PN$ of Example~\ref{ex:top-dijkstra}. Specifically, we include all places
  that model the state of the first agent and the transitions that change the
  state of the first agent. Some of these transitions \enquote{observe} places
  of the state of the zeroth and second agent. These places are added with
  dashed lines.}
  \label{fig:dijkstra-instance-3}
  \begin{center}
    \scalebox{1}{
      \def\minWidth{1.5cm}
      \begin{tikzpicture}[every node/.style={scale=0.6, inner sep=0pt}]
        \node[place, minimum width=\minWidth] (idle1)    at (- 2,   3) {$\idled(1)$};
        \node[place, minimum width=\minWidth] (trying1)  at (  2,   3) {$\trying(1)$};

        \node[place, minimum width=\minWidth] (init1)    at (- 3,   0) {$\stateinitial(1)$};
        \node[place, minimum width=\minWidth] (loop1)    at (  0,   0) {$\stateloop(1)$};
        \node[place, minimum width=\minWidth] (looping1) at (  3,   0) {$\statelooping(1)$};

        \node[place, minimum width=\minWidth] (ptr1bot)  at (- 4.5, - 3) {$\ptr(1, \bot)$};
        \node[place, minimum width=\minWidth] (ptr10)    at (- 1.5, - 3) {$\ptr(1, 0)$};
        \node[place, minimum width=\minWidth] (ptr11)    at (  1.5, - 3) {$\ptr(1, 1)$};
        \node[place, minimum width=\minWidth] (ptr12)    at (  4.5, - 3) {$\ptr(1, 2)$};

        \node[place, minimum width=\minWidth] (break1)   at (- 3, - 6) {$\statebreak(1)$};
        \node[place, minimum width=\minWidth] (crit1)    at (  3, - 6) {$\statecrit(1)$};

        \node[place, minimum width=\minWidth] (done1)    at (  3, - 9) {$\statedone(1)$};

        \node[transition] (init2loop) at (- 1.5, 0) {}
          edge [pre] (init1)
          edge [pre] (idle1)
          edge [post] (loop1)
          edge [post] (trying1)
        ;

        \node[transition] (initLoop) at (1.5, 0) {}
          edge [pre]  (loop1)
          edge [post] (looping1)
          edge [pre]  (ptr1bot)
          edge [post] (ptr10)
        ;

        \node[place, minimum width=\minWidth, dashed] (idle0) at (0, -4) {$\idled(0)$};

        \node[transition] (loopstep0) at (0, -3) {}
          edge [pre]  (ptr10)
          edge [post] (ptr11)
          edge [<->, shorten <=1pt, shorten >=1pt] (idle0)
        ;

        \node[place, minimum width=\minWidth, dashed] (trying0) at (-2.5, -4) {$\trying(0)$};

        \node[transition] (loopfail0) at (-1.5, -4) {}
          edge [pre]  (ptr10)
          edge [post] (break1)
          edge [<->, shorten <=1pt, shorten >=1pt] (trying0)
          edge [pre]  (looping1)
          edge [post, bend right] (ptr1bot)
        ;

        \node[transition] (loopstep1) at (3, -3) {}
          edge [pre]  (ptr11)
          edge [post] (ptr12)
        ;

        \node[place, minimum width=\minWidth, dashed] (idle2) at (5.5, -4) {$\idled(2)$};

        \node[transition] (loopstep2) at (4.5, -4) {}
          edge [pre, bend left]  (looping1)
          edge [post] (crit1)
          edge [pre]  (ptr12)
          edge [post, bend left=50] (ptr1bot)
          edge [<->, shorten <=1pt, shorten >=1pt] (idle2)
        ;

        \node[place, minimum width=\minWidth, dashed] (trying2) at (2.5, -4) {$\trying(2)$};

        \node[transition] (loopfail2) at (3.5, -4) {}
          edge [pre, bend left=46]  (looping1)
          edge [post, bend left] (break1)
          edge [pre]  (ptr12)
          edge [post, bend left=46] (ptr1bot)
          edge [<->, shorten <=1pt, shorten >=1pt] (trying2)
        ;

        \node[transition] (crit2done) at (3, -7.5) {}
          edge [pre] (crit1)
          edge [post] (done1)
        ;

        \node[transition] (break2initial) at (-4, -6) {}
          edge [pre]  (break1)
          edge [post] (init1)
          edge [pre, bend left=12]  (trying1)
          edge [post, bend right=5] (idle1)
        ;

        \node[transition] (done2initial) at ( 2, -9) {}
          edge [pre]  (done1)
          edge [post, bend left=17] (init1)
          edge [pre]  (trying1)
          edge [post, bend right=8] (idle1)
        ;

      \end{tikzpicture}
      \let\minWidth\undefined
    }
  \end{center}
\end{figure}
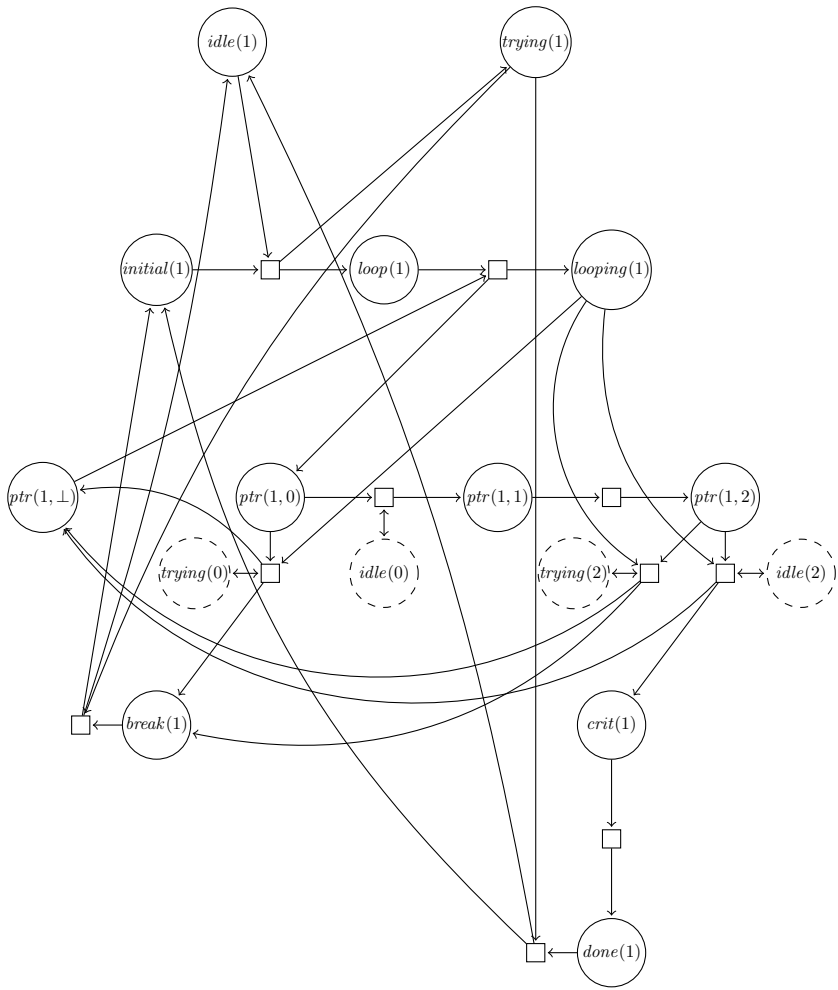

\subsection{Parametrizing traps of inspection programs}
We introduce parametrization results for inspection programs that, given a trap
for an instance, produce a set of traps for all instances. The essential
observation is that every transition involves a finite amount of indices: for
local transitions there is exactly one index involved, while for loop
transitions there are at most 3 involved indices; the index $i$ of the agent
that is executing the loop transition, the index $j$ of the agent being
currently inspected by agent $i$ and (potentially) the index $j+1$ to which
agent $i$ advances its pointer.

\paragraph{Indexed traps of inspection programs.}
We denote an indexed trap as a triple $\tuple{n, Q, R}$, where $Q \subseteq
\places \times [n]$ and $R \subseteq \relations \times [n] \times \left( [n]
\cup \set{\bot} \right)$ are sets of places such that $Q \cup R$ is a trap of
$\PN(n)$. Given an indexed trap $\tuple{n, Q, R}$, we introduce the following
notations, where $i, j \in [n]$:
\begin{align*}
  Q[i] &= \set{p \in \places \mid p(i) \in Q} \\
  R[i, j] & = \set{r \in \relations \mid r(i, j) \in R} \\
  R[-, j] & = \set{r(i) \in \relations \times [n] \mid r(i, j) \in R} \\
  L & = \set{i \in [n] \mid \text{ there exists }j \in \left([n] \cup
  \set{\bot} \right)\text{ with }R[i, j] \neq \emptyset}
\end{align*}
$Q[i]$ is used in the same way as before, while $R[i, j]$ and $R[-, j]$ are
generalizations of this concept to the relation symbols of the generalized net.
$L$ is the set of agents such that the trap $Q \cup R$ contains at least one
place for which there is some $r \in R$ such that $r(i, j)$ is in the trap for
some $j$. We call $L$ the set of \emph{looping indices} of the indexed trap.
This naming convention is inspired by the considered topology: any place $r(i,
j)$ corresponds to the fact that agent $i$ executes some loop and currently
points to $j$. Let us illustrate these various notions with an example.
\begin{example}
  \label{ex:loop-trap}
  Consider the parameterized net $\PN$ from Example~\ref{ex:dijkstra}.
  In the Petri net $\PN(7)$, the following set of places constitutes a
  trap:
  \begin{equation*}
    \set{
      \begin{aligned}
        &\statebreak(3), \stateloop(3), \falseval(3),\\
        &\statebreak(5), \stateloop(5), \falseval(5),\\
        &\ptr(3, 0), \ptr(3, 1), \ptr(3, 2),\\
        &\ptr(3, 3), \ptr(3, 4), \ptr(3, 5)\\
        &\ptr(5, 0), \ptr(5, 1), \ptr(5, 2), \ptr(5, 3),
      \end{aligned}
    }.
  \end{equation*}
  The corresponding indexed trap is
  \begin{equation*}
    \tuple{n, Q, R} := \tuple{
      7,
      \set{
        \begin{aligned}
          &\statebreak(3), \stateloop(3), \falseval(3),\\
          &\statebreak(5), \stateloop(5), \falseval(5),\\
        \end{aligned}
      }, \set{
        \begin{aligned}
          &\ptr(3, 0), \ptr(3, 1), \ptr(3, 2), \\
          &\ptr(3, 3), \ptr(3, 4), \ptr(3, 5),\\
          &\ptr(5, 0), \ptr(5, 1), \ptr(5, 2), \ptr(5, 3)
        \end{aligned}
      }
    }.
  \end{equation*}
  and we have
  \begin{center}
    \scalebox{0.8}{
    \begin{tabular}{r|c@{\hspace{2mm}}c@{\hspace{2mm}}c@{\hspace{2mm}}c@{\hspace{2mm}}c@{\hspace{2mm}}c@{\hspace{2mm}}c}
      $i$ & $0$ & $1$ & $2$ & $3$ & $4$ & $5$ & $6$  \\  \hline
      $Q[i]$ & $\emptyset$ &  $\emptyset$ & $\emptyset$ &  $\set{\statebreak, \stateloop, \falseval}$ & $\emptyset$ &  $\set{\statebreak, \stateloop, \falseval}$  &  $\emptyset$\\
      $R[-,i]$ &  $\set{r(3), r(5)}$  & $\set{r(3), r(5)}$ &  $\set{r(3),
      r(5)}$ & $\set{r(3), r(5)}$ &  $\set{r(3)}$ & $\set{r(3)}$ & $\emptyset$
    \end{tabular}
    }
  \end{center}
  \medskip
\end{example}

\paragraph{Parametrizing indexed traps.}
Let $C_{i}$ denote the column of the table of Example~\ref{ex:loop-trap} for
index $i$. The table itself is then determined by the sequence $C_0 \; C_1 \;
\cdots \; C_6$: the \emph{column representation} of the indexed trap.

\begin{definition}
  Let $\tuple{n, Q, R}$ be an indexed trap. We call the sequence $C_0 \; C_1 \;
  \cdots \; C_{n-1}$, where $C_i := (Q[i], R[-,i])$, the \emph{column
  representation} of $\tuple{n, Q, R}$.
\end{definition}

Since $\tuple{n,Q,R}$ and $C_0 \; C_1 \; \cdots \; C_{n-1}$ are different
representations of the same object, we abuse language and speak of the indexed
trap $C_0 \, C_1 \cdots C_{n-1}$.

Observe that the indexed trap of Example \ref{ex:loop-trap} satisfies $C_1 =
C_2 = C_3$. We are going to prove that, if an indexed trap $C_0 \; C_1 \;
\cdots \; C_{n-1}$ of $\PN(n)$ satisfies $C_{i-1} = C_i = C_{i+1}$ for some $0
< i < n-1$, then one can \enquote{insert another $C_{i}$} to get an indexed
trap for $n+1$. We introduce for this $\shift^{i}_{m}(C_j)$ which is the result
of shifting indices in $R[-, j]$ that are larger than $i$ by $m$ steps; i.e.,
applying the simultaneous substitution $[i+m \leftarrow i, i+m+1 \leftarrow
i+1, \ldots, n+m \leftarrow n]$ to $R[-, j]$. Then \emph{every} sequence of the
form
\begin{equation*}
  \shift^{i}_{m}(C_0) \; \cdots \; \shift^{i}_{m}(C_{i-1}) \;
  \shift^{i}_{m}(C_i) \;
  (\shift^{i}_{m}(C_i))^m \; \shift^{i}_{m}(C_{i+1}) \; \cdots \;
  \shift^{i}_{m}(C_{n-1})
\end{equation*}
for every $m \geq 0$ is an indexed trap of $\PN(n+m)$.

\begin{example}
  For Example~\ref{ex:loop-trap}, this observation allows us to find a trap in
  $\PN(8)$:
  \begin{center}
    \scalebox{0.75}{
    \begin{tabular}{r|c@{\hspace{2mm}}c@{\hspace{2mm}}c@{\hspace{2mm}}c@{\hspace{2mm}}c@{\hspace{2mm}}c@{\hspace{2mm}}c@{\hspace{2mm}}c}
      $i$ & $0$ & $1$ & $2$ & $3$ & $4$ & $5$ & $6$ & $7$  \\  \hline
      $Q[i]$ & $\emptyset$ &  $\emptyset$ & $\emptyset$ &  $\emptyset$ & $\set{\statebreak, \stateloop, \falseval}$ & $\emptyset$ &  $\set{\statebreak, \stateloop, \falseval}$  & $\emptyset$ \\
      $R[-,i]$ &  $\set{r(4), r(6)}$  & $\set{r(4), r(6)}$ & $\set{r(4), r(6)}$ & $\set{r(4), r(6)}$ & $\set{r(4), r(6)}$ &  $\set{r(4), r(6)}$ & $\set{r(4)}$ & $\emptyset$
    \end{tabular}
    }
  \end{center}
  Note that the looping indices $3$ and $5$ increased consistently in all $R[-,
  j]$ sets. This is a consequence of the $\shift^{1}_{1}$ operation, since $2 <
  3 < 5$.
\end{example}

\begin{lemma}
  \label{lem:looping}
  Let $T = C_{0} \; C_{1} \; \ldots \; C_{n-1}$ be an indexed trap (in column
  representation) for some inspection program $\PN$ with looping indices $L$.
  If there exists $0 < i < n-1$ such that $C_{i-1} = C_{i} = C_{i+1}$ and $i-1,
  i, i+1 \notin L$, then for every $m \geq 0$ the sequence
  \begin{equation*}
   T_{i,m} :=  \shift^{i}_{m}(C_{0}) \;
    \ldots \;
    \shift^{i}_{m}(C_{i-1}) \; \left(\shift^{i}_{m}(C_{i})\right)^{m + 1} \;
    \shift^{i}_{m}(C_{i+1}) \; \ldots \; C_{n-1}
  \end{equation*}
  is an indexed trap of $\PN(n+m)$.
\end{lemma}
\begin{proof}
 We only consider the special case $m = 1$, since the general case follows by
 applying the special case $m$ times. We have
 \begin{equation}
   \label{Ti1}
   T_{i,1} := \shift^{i}_{1}(C_{0}) \; \ldots \; \shift^{i}_{1}(C_{i-1}) \;
   \shift^{i}_{1}(C_{i}) \; \shift^{i}_{1}(C_{i}) \; \shift^{i}_{1}(C_{i+1}) \;
   \ldots \; C_{n-1}
  \end{equation}

  Let $\unshift_{1}^{i}$ be the inverse operation of $\shift_{1}^{i}$; that is,
  $\unshift_{1}^{i}$ applies simultaneously the substitution $[i \leftarrow
  i+1, i+1 \leftarrow i+2, \ldots, n \leftarrow n+1]$ to all $R[-, j]$ sets.
  One can easily verify that for every $j = i-1, i, i+1, i+2$ the following
  holds:
  \begin{equation*}
    \scalebox{0.95}{\mbox{\ensuremath{
    \begin{aligned}
      &\; \unshift_{1}^{i}(\shift^{i}_{1}(C_{0})) \ldots
      \unshift_{1}^{i}(\shift^{i}_{1}(C_{j-1}))
      \unshift_{1}^{i}(\shift^{i}_{1}(C_{j+1})) \ldots
      \unshift_{1}^{i}(\shift^{i}_{1}(C_{n-1}) \\
      = &\; C_{0} \; \ldots \;
      C_{j-1} \;
      C_{j+1} \; \ldots \;
      C_{n-1} \\
      = &\;T.
    \end{aligned}}}}
  \end{equation*}
  This is a consequence of $\shift^{i}_{1}(C_{i-1}) = \shift^{i}_{1}(C_{i}) =
  \shift^{i}_{1}(C_{i+1}) = \shift^{i}_{1}(C_{i+2})$, $C_{i-1} = C_{i} =
  C_{i+1}$, and $\unshift^{i}_{1}(\shift^{i}_{1}(C_{i})) = C_{i}$.

  For the sake of contradiction, assume that $T_{i,1}$ is not an indexed trap
  of $\PN(n+1)$. Hence, there exists a transition $\tuple{Q_{1}, Q_{2}}$ in
  $\PN(n+1)$ such that $T_{i,1} \cap Q_1 \neq \emptyset$ but $T_{i,2} \cap Q_2
  = \emptyset$.

  As mentioned before, at most $3$ indices are involved in any instance of any
  transition pattern of an inspection program. So there are indices $j_{1},
  j_{2}, j_{3}$ such that $Q_{1}, Q_{2} \subseteq \places \times \set{j_{1},
  j_{2}, j_{3}} \cup \relations \times \set{j_{1}, j_{2}, j_{3}} \times
  \set{j_{1}, j_{2}, j_{3}, \bot}$. Let $A_{0} \; \ldots \; A_{n}$ and $B_{0}
  \; \ldots \; B_{n}$ be the column representations of $Q_{1}$ and $Q_{2}$
  respectively. Pick now some $k \in \set{i-1, i, i+1, i+2} \setminus
  \set{j_{1}, j_{2}, j_{3}}$. We have $A_{k} = \tuple{\emptyset, \emptyset}$
  and $B_{k} = \tuple{\emptyset, \emptyset}$.
  Informally speaking, we are going to remove agent $k$ from the system
  and obtain a transition that contradicts $T$ being a trap.
        The rest of the proof is the implementation of this idea.

        Define $\tuple{Q_{1}', Q_{2}'}$
  as the transition of $\PN(n)$ given by (in column representation)
    \begin{align}
      Q_{1}' &:= \unshift_{1}^{k}(A_{0}) \; \ldots \;
      \unshift_{1}^{k}(A_{k-1}) \;
      \unshift_{1}^{k}(A_{k+1}) \; \ldots \;
      \unshift_{1}^{k}(A_{n+1})  \label{Q1p}\\
      Q_{2}' &:= \unshift_{1}^{k}(B_{0}) \; \ldots \;
      \unshift_{1}^{k}(B_{k-1}) \;
      \unshift_{1}^{k}(B_{k+1}) \; \ldots \;
      \unshift_{1}^{k}(B_{n+1}) \label{Q2p}
    \end{align}
  respectively. ($\tuple{Q_{1}', Q_{2}'}$ is indeed a transition of $\PN(n)$,
  obtained from  the same transition pattern as $\tuple{Q_{1}, Q_{2}}$ with
  appropriate indices.)

  We prove now that $Q_{1}' \cap T \neq \emptyset$ but $Q_{2}' \cap T =
  \emptyset$, contradicting the assumption that $T$ is an indexed trap of
  $\PN(n)$. Since $k \in \set{i, i+1, i+2, i+3}$, and none of these indices is
  a looping index of $T_{i,1}$, we have
      \begin{align}
    T \; & = & 
     \unshift_{1}^{i}(\shift^{i}_{1}(C_{0}))
    \ldots
    \unshift_{1}^{i}(\shift^{i}_{1}(C_{j-1}))
    \unshift_{1}^{i}(\shift^{i}_{1}(C_{j+1})) \ldots
    \unshift_{1}^{i}(\shift^{i}_{1}(C_{n-1})  \nonumber \\
    & = &  \unshift_{1}^{k}(\shift^{i}_{1}(C_{0}))
    \ldots
    \unshift_{1}^{k}(\shift^{i}_{1}(C_{j-1}))
    \unshift_{1}^{k}(\shift^{i}_{1}(C_{j+1})) \ldots
    \unshift_{1}^{k}(\shift^{i}_{1}(C_{n-1})  \label{T}
  \end{align}
  Let us prove $Q_{1}' \cap T \neq \emptyset$. Recall that $Q_1 = A_{0} \;
  \ldots \; A_{n}$. Since $Q_{1} \cap T_{i,1} \neq \emptyset$ holds by
  assumption, from (\ref{Ti1}) we obtain an index $a \in \set{j_{1}, j_{2},
  j_{3}}$ such that $\shift_{1}^{i}(C_{a}) \cap A_{a} \neq \emptyset$; that is,
  $\shift_{1}^{i}(C_{a}) = \tuple{Q^{C_{a}}, R^{C_{a}}}$ and $A_{a} =
  \tuple{Q^{A_{a}}, R^{A_{a}}}$ such that either $Q^{C_{a}} \cap Q^{A_{a}} \neq
  \emptyset$ or $R^{C_{a}} \cap R^{A_{a}} \neq \emptyset$. Then,
  however, $\unshift_{1}^{k}(\shift_{1}(C_{a})) \cap \unshift_{1}^{k}(A_{a})
  \neq \emptyset$.

  Similarly, we observe that $\shift_{1}^{i}(C_{a}) \cap B_{a} =
  \emptyset$ for all $a$. It follows, now, that
  $\unshift_{1}^{k}(\shift_{1}^{i}(C_{a})) \cap \unshift_{1}^{k}(B_{a}) =
  \emptyset$ for all $a$. Consequently, by (\ref{Q1p}), (\ref{Q2p}) and
  (\ref{T}) we obtain the contradiction that $T$ is not a trap. The result
  follows.
\end{proof}
Lemma~\ref{lem:looping} paves the way for a generalization result for
inspection programs which, given a trap $T$ of an instance, constructs a
parametrization $\Param_\interp{T}(\bm{n}, \Xv, \Uv)$.

Note that Lemma~\ref{lem:looping} essentially allows to insert $C_{i}^{*}$;
that is, an arbitrary repetition of the letter $C_{i}$ at the appropriate
location, to obtain a family of traps. One has to still account for the correct
$m$ in the $\shift_{m}$ operation but this observation might help to understand
the following generalization theorem. Let us introduce a bit of simplifying
notation first. We write
\begin{equation*}
  \Qequals{P}(t) \text{ for }\bigwedge_{p \in P}t \in \Xv_{p} \land
  \bigwedge_{p \in \places \setminus P}t \notin \Xv_{p}.
\end{equation*}
That means that \enquote{$Q[t]$} is $P$; i.e., this formula is satisfied by
some interpretation $\mu$ if $\set{\mu(t)} \times P \subseteq \bigcup_{p \in
\places}\mu(\Xv_{p}) \times \set{p}$ (the set of places $\mu$ encodes in the
placeset variables $\Xv$).
We introduce a similar notion to represent that $R[-, t]$ is of a certain shape
for some term $t$. Since $R[-, t]$ contains elements of the form $r(\ell)$ for
some $r \in \relations$ and some looping index $\ell \in L$, assume that the
looping indices are $\ell_{1} < \ell_{2} < \ldots < \ell_{k} < u_{1} < u_{2} <
\ldots < u_{m}$ and there exists a corresponding first-order variable
$\bm{\ell}_{i}$ for every $1 \leq i \leq k$ and $\bm{u}_{i}$ for $1 \leq i \leq
m$. Now we write
\begin{equation*}
  \Requals{R}(t) \text{ for }
  \left[
    \begin{aligned}
    &\bigwedge_{\substack{1 \leq v \leq k \\ r(\ell_{v}) \in R}}
      \tuple{\bm{\ell}_{v}, t} \in \Uv_{r} \land
    &\bigwedge_{\substack{1 \leq v \leq k \\ r(\ell_{v}) \notin R}}
      \tuple{\bm{\ell}_{v}, t} \notin \Uv_{r} \\
    \land &\bigwedge_{\substack{1 \leq v \leq m \\ r(u_{v}) \in R}}
      \tuple{\bm{u}_{v}, t} \in \Uv_{r} \land
    &\bigwedge_{\substack{1 \leq v \leq m \\ r(u_{v}) \notin R}}
      \tuple{\bm{u}_{v}, t} \notin \Uv_{r} \\
  \end{aligned}\right].
\end{equation*}
This allows us to formulate our generalization result in a compact way:
\begin{theorem}
  \label{thm:looping-gen}
  Let $T = C_{0} \; C_{1} \; \ldots \; C_{n-1}$ be an indexed trap (in column
  representation) in some looping program $\PN$ with looping indices $L = \set{
  \ell_{1}, \ldots, \ell_{k}, u_{1}, \ldots, u_{m}}$. If there exists $i \in
  [n]$ such that $C_{i-1} = C_{i} = C_{i+1}$ and $\ell_{1} < \ell_{2} < \ldots
  < \ell_{k} < i-1 < i < i+1 < u_{1} < u_{2} < \ldots < u_{m}$, then every
  model of
  \begin{equation*}
    \begin{aligned}
    \Param_\interp{T}& (\bm{n}, \Xv, \Uv) \coloneqq
      n \leq \bm{n} \; \land \; \exists \bm{y} ~.~ \bm{y} + (n - (i + 2)) = \bm{n}
      \land i+1 \leq \bm{y} \\
      \land &\bigwedge_{j < i-1} \Qequals{Q[j]}(j)
      \land \bigwedge_{i + 1 < j < n} \Qequals{Q[j]}(\bm{y} + (j - (i + 2))) \\
      \land &\forall \bm{j} ~.~ i-1 \leq \bm{j} \leq \bm{y} \rightarrow \Qequals{Q[i]}(\bm{j}) \\
      \land &\forall \bm{\ell}_{1}, \ldots, \bm{\ell}_{k}, \bm{u}_{1}, \ldots,
      \bm{u}_{m} ~.~ \bigwedge_{1 \leq v \leq k} \bm{\ell}_{v} = \ell_{v} \land
      \bigwedge_{1 \leq v \leq m} \bm{u}_{v} = \bm{y} + (u_{v} - (i + 2)) \\
      &\rightarrow \left(
      \begin{aligned}
        &\bigwedge_{j < i-1} \Requals{R[*, j]}(j) \land
        \bigwedge_{i < j} \Requals{R[*, j]}(\bm{y} + (j - (i+2)))\\
        \land & \forall i-1 \leq \bm{j} \leq \bm{y} \rightarrow \Requals{R[*,
        i]}(\bm{j})
      \end{aligned}
      \right)
    \end{aligned}
  \end{equation*}
  corresponds to an indexed trap.
\end{theorem}
\begin{proof}
  Let $\mu$ be a model of $\Param_\interp{T}(\bm{n}, \Xv, \Uv)$. Consider the
  triple $\tuple{n', Q', R'}$ such that $n' = \mu(\bm{n})$, $Q' = \bigcup_{p
  \in \places}\set{p} \times \mu(\Xv_{p})$, and $R' = \bigcup_{r \in
  \relations}\set{r} \times \mu(\Uv_{r})$. Examining $\Param_\interp{T}(\bm{n},
  \Xv, \Uv)$ closely reveals that its column representation is
  \begin{equation*}
    \shift^{i}_{m}(C_{0}) \; \shift^{i}_{m}(C_{1}) \; \ldots \;
    \shift^{i}_{m}(C_{i-1}) \; \shift^{i}_{m}(C_{i})^{m + 1} \;
    \shift^{i}_{m}(C_{i+1}) \; \ldots \; C_{n-1}
  \end{equation*}
  for some $m \geq 0$. The result follows now immediately from
  Lemma~\ref{lem:looping}.
\end{proof}

%% file: tex/experiments.tex
\section{Experiments}
\subsection{\ostrich}
We implemented the CEGAR loop and the parameterization techniques of Sections
\ref{sec:rings} and \ref{sec:crowds} in our tool \ostrich{}. \ostrich{} heavily
relies on MONA as a \ac{WS1S}-solver. The results of our experiments are
presented in Figure~\ref{fig:results}. In the first two columns the table
reports the topology and the name of the system to be verified. The array
topology is a linear topology where agents can refer existentially or
universally to agents with smaller or larger indices. Analogously to the other
topologies we derive a sound parameterization technique for traps, 1-BB sets,
and siphons. The rings are Dijsktra's token ring for mutual exclusion
\cite{FribourgO97} and a model of the dining philosophers in which philosophers
pick both forks simultaneously. For headed rings we consider
Example~\ref{ex:diningphil} and a model of a message passing leader election
algorithm. The array is Burns' mutual exclusion algorithm
\cite{DBLP:conf/tacas/JensenL98}. The crowds are Dijkstra's algorithm for
mutual exclusion \cite{CooperatingSequentialProcesses} and models of
cache-coherence protocols taken from \cite{Delzanno00}. Note that we check
inductiveness of the property; i.e., if it holds initially and there is no
marking satisfying the property and the current abstraction and reaching in a
single step a marking which violates the property. Additionally, we include in
the specification of the parameterized Petri net a partition of the places
$\places$ such that the places of every index in every instance form a 1BB-set.
Collectively, this ensures that all examples are 1-bounded and yields
invariants similar to (1), (2) for Example~\ref{ex:diningphil}. Since
\ostrich{} does not compute but only checks these invariants we do not count
them in Figure~\ref{fig:results} (leading to 3 semi-automatic invariants for
Example~\ref{ex:diningphil} since we omit (1), and (2)). Moreover, these
invariants already imply inductiveness of some safety properties; prominently
deadlock-freedom for all considered cache-coherence protocols.

The third column gives the time \ostrich{} needs to initialize the analysis;
this includes verifying that the given parameterized Petri net is covered by
1BB-sets, and that it indeed has the given topology. The fourth column gives
the property being checked. The specification of the cache coherence protocols
consists of a number of consistency properties, specific for each protocol. The
legend \enquote{consistency ($x$/$y$)} indicates that the specification
consists of $y$ properties, of which \ostrich{} was able to automatically prove
the inductiveness of $x$. Column~5 gives the time need to check the
inductiveness the property (or, in the case of the
cache-coherence protocols, either find a marking which satisfies all
constraints imposed by 1BB-sets, traps or siphons, or prove the inductiveness
of the properties together). Columns~6, 7, and 8 give the number of
WS1S-formulas, each corresponding to a parameterized 1BB-sets, trap, or siphon
that are computed by the CEGAR loop. Some of these WS1S-formulas have only one
model, i.e., they correspond to a single trap, siphon, or 1BB-set of one
instance. Such \enquote{artifacts} are needed when small instances
(e.g., arrays of size 2) require ad-hoc proofs that cannot be parameterized. In
these cases the \enquote{real} number of parametric invariants is the result of
subtracting the number of artifacts from the total number. The last column
reports the number of parameterized inductive invariants obtained by
the semi-automatic CEGAR loop. There the user is presented a series of counter
examples to the inductiveness of the property. The user can check for traps,
siphons or 1BB-sets to disprove the counter example. If the user then provides
an invariant which proves inductive it is used to refine the abstraction until
no further counter example can be found. The response time of \ostrich{} in
this setting is immediate which provides a nice user experience.
MOESI is an example which shows that the semi-automatic procedure can lead to
proofs with fewer invariants. For Dragon four of the seven invariants are
artifacts; thus, it also shows that a semi-automatic approach allows for proofs
with fewer invariants. The last step of the automatic procedure is to remove
invariants until no invariant can be removed without obtaining a counter
example again.

For Example~\ref{ex:diningphil} \ostrich{} automatically computes the following
family of 1BB-sets (additionally to the invariants (1) and (2)): (For
readability we omit some artifacts.)
\begin{equation*}
    \begin{aligned}
      &2 \leq \bm{n} \land \taken = \thinking = \emptyset \land \waiting =
      \eating = \set{0, 1} \land \free = \set{1}\\
      &3 \leq \bm{n} \land \taken = \waiting = \thinking = \emptyset
      \land \eating = \set{\bm{n}-1, 0} \land \free = \set{0}\\
      &4 \leq \bm{n} \land \taken = \thinking = \emptyset
      \land \free = \waiting = \set{\bm{n} - 1} \land \eating = \set{\bm{n}-2,
      \bm{n}-1}\\
      &2 \leq \bm{n} \land \exists \bm{i}: 1 < \bm{i} < \bm{n-2} \land
      \left(
      \begin{aligned}
        &\taken = \thinking = \emptyset \land \free = \waiting = \set{\bm{i}
        \oplusn 1}\\
        \land &\eating = \set{\bm{i}, \bm{i} \oplusn 1}
      \end{aligned}
      \right)
    \end{aligned}
\end{equation*}

\begin{figure}[h]
  \caption{Experimental results of \ostrich{}. The complete data is available
  at \cite{OstrichSubmissionURL}.}
  \label{fig:results}
  \begin{center}
  \scalebox{0.6}{
  \bgroup
  \setlength\tabcolsep{5pt}.
  \def\arraystretch{1.3}
    \begin{tabular}{c||l|r|lrrrrc}
       \hline
          \begin{tabular}{c}\multirow{2}{*}{Topology}\\\phantom{(ms)}\end{tabular}
        & \begin{tabular}{l}\multirow{2}{*}{Example}\\\phantom{(ms)}\end{tabular}
        & \begin{tabular}{r}Init.\\(ms)\end{tabular}
        & \begin{tabular}{c}\multirow{2}{*}{Property}\\\phantom{(ms)}\end{tabular}
        & \begin{tabular}{r}Check\\(ms)\end{tabular}
        & \begin{tabular}{r}\multirow{2}{*}{1BB-sets}\\\phantom{(ms)}\end{tabular}
        & \begin{tabular}{r}\multirow{2}{*}{Traps}\\\phantom{(ms)}\end{tabular}
        & \begin{tabular}{r}\multirow{2}{*}{Siphons}\\\phantom{(ms)}\end{tabular}
        & \begin{tabular}{c}Semi-automatic \\ invariants\end{tabular}\\\hline

        &
        &
        & deadlock
        & 40
        & 1 (1)
        & 0 (0)
        & 0 (0)
        &
        \\

        & \multirow{-2}{*}{Dijkstra ring}
        & \multirow{-2}{*}{558}
        & mutual exclusion
        & 125
        & 1 (1)
        & 1 (1)
        & 0 (0)
        & \multirow{-2}{*}{2}
        \\  \cline{2-9}
          \multirow{-3}{*}{ring}
        & atomic phil.
        & 409
        & deadlock
        & 79
        & 1 (1)
        & 0 (0)
        & 0 (0)
        & 4
        \\ \cline{1-9}

        & lefty phil.
        & 495
        & deadlock
        & 294
        & 7 (4)
        & 0 (0)
        & 0 (0)
        & 3
        \\\cline{2-9}
          \phantom{headed ring}
        & \phantom{leader election}
        & \phantom{695}
        & not $0$ and $n-1$ leader
        & 965
        & 1 (0)
        & 0 (0)
        & 2 (1)
        &
        \\
          \multirow{-3}{*}{headed ring}
        & \multirow{-2}{*}{leader election}
        & \multirow{-2}{*}{670}
        & not two leaders
        & --
        & --
        & --
        & --
        & \multirow{-2}{*}{1}
        \\\hline
          \phantom{array}
        & \phantom{Burns}
        & \phantom{508}
        & deadlock
        & 16
        & 0 (0)
        & 0 (0)
        & 0 (0)
        &
        \\
          \multirow{-2}{*}{array}
        & \multirow{-2}{*}{Burns}
        & \multirow{-2}{*}{501}
        & mutual exclusion
        & 379
        & 0 (0)
        & 8 (7)
        & 0 (0)
        & \multirow{-2}{*}{1}
        \\ \cline{1-9}

        &
        &
        & deadlock
        & 88
        & 2 (1)
        & 0 (0)
        & 0 (0)
        &
        \\

        & \multirow{-2}{*}{Dijkstra}
        & \multirow{-2}{*}{1830}
        & mutual exclusion
        & 1866
        & 0 (0)
        & 3 (1)
        & 0 (0)
        & \multirow{-2}{*}{3}
        \\ \cline{2-9}
          \phantom{crowd}
        & \phantom{Berkeley}
        & \phantom{414}
        & deadlock
        & 12
        & 0 (0)
        & 0 (0)
        & 0 (0)
        &
        \\

        & \multirow{-2}{*}{Berkeley}
        & \multirow{-2}{*}{414}
        & consistency (3/3)
        & 361
        & 0 (0)
        & 9 (1)
        & 0 (0)
        & \multirow{-2}{*}{1}
        \\ \cline{2-9}

        &
        &
        & deadlock
        & 19
        & 0 (0)
        & 0 (0)
        & 0 (0)
        &
        \\

        & \multirow{-2}{*}{Dragon}
        & \multirow{-2}{*}{538}
        & consistency (7/7)
        & 2334
        & 52 (7)
        & 0 (0)
        & 0 (0)
        & \multirow{-2}{*}{7}
        \\ \cline{2-9}

        &
        &
        & deadlock
        & 14
        & 0 (0)
        & 0 (0)
        & 0 (0)
        &
        \\

        & \multirow{-2}{*}{Firefly}
        & \multirow{-2}{*}{511}
        & consistency (0/4)
        & 232
        & 0 (0)
        & 2 (0)
        & 0 (0)
        & \multirow{-2}{*}{2}
        \\ \cline{2-9}

        &
        &
        & deadlock
        & 13
        & 0 (0)
        & 0 (0)
        & 0 (0)
        &
        \\

        & \multirow{-2}{*}{Illinois}
        & \multirow{-2}{*}{468}
        & consistency (0/2)
        & 180
        & 0 (0)
        & 3 (0)
        & 0 (0)
        & \multirow{-2}{*}{1}
        \\ \cline{2-9}

        &
        &
        & deadlock
        & 12
        & 0 (0)
        & 0 (0)
        & 0 (0)
        &
        \\
          \phantom{crowd}
        & \multirow{-2}{*}{MESI}
        & \multirow{-2}{*}{422}
        & consistency (2/2)
        & 500
        & 0 (0)
        & 13 (2)
        & 0 (0)
        & \multirow{-2}{*}{1}
        \\ \cline{2-9}

        &
        &
        & deadlock
        & 13
        & 0 (0)
        & 0 (0)
        & 0 (0)
        &
        \\

        & \multirow{-2}{*}{MOESI}
        & \multirow{-2}{*}{446}
        & consistency (7/7)
        & 1226
        & 0 (0)
        & 24 (4)
        & 0 (0)
        & \multirow{-2}{*}{1}
        \\ \cline{2-9}

        &
        &
        & deadlock
        & 12
        & 0 (0)
        & 0 (0)
        & 0 (0)
        &
        \\
          \multirow{-16}{*}{crowd}
        & \multirow{-2}{*}{Synapse}
        & \multirow{-2}{*}{420}
        & consistency (2/2)
        & 22
        & 0 (0)
        & 0 (0)
        & 0 (0)
        & \multirow{-2}{*}{0}
        \\ \cline{1-9}
    \end{tabular}
    \egroup
  }
  \end{center}
\end{figure}

\subsection{\heron}
We implemented the approach described in Section~\ref{sec:gppn} in our tool
\heron{} \cite{heron-gitlab,heron-artifact}. An illustration of the general
concept can be found in Figure~\ref{fig:heron}.
\begin{figure}
  \caption{An illustration of the designs of \ostrich{} (upper diagram) and
  \heron{} (lower diagram). For both the input is a (generalized) parameterized
  Petri net $\PN$ and a property $\Safe$ that we want to check. For
  \ostrich{} we can express $\WitTrap$ in WS1S, while \heron{}
  specifically relies on an appropriate embedding into SAT. Also,  \ostrich{} uses MONA to check $\Safety$ or get a
  counter-example, while \heron{} uses first-order provers. Further, \heron{} uses timeouts for its
  provers; e.g., \texttt{VAMPIRE} or \texttt{CVC4}. The
  experimental data suggests that these timeouts can be chosen small; in particular,
  $2-3$ seconds is appropriate for our mutual exclusion benchmarks.}
%
  \label{fig:heron}
  \begin{center}
    \input{tex/diagram-ostrich.tex}
  \end{center}
  \begin{center}
  \input{tex/diagram-heron.tex}
  \end{center}
\end{figure}
\heron{} is written in Python. It uses \texttt{clingo}
\cite{DBLP:journals/aicom/GebserKKOSS11} as SAT solver. To solve the
first-order queries \heron{} uses \texttt{VAMPIRE} \cite{DBLP:conf/cav/KovacsV13}
and \texttt{CVC4} \cite{DBLP:conf/cav/BarrettCDHJKRT11}. As benchmarks we
consider classical algorithms for mutual exclusion. These include a reduced
version of Dijkstra's algorithm for mutual exclusion
\cite{CooperatingSequentialProcesses}, which we presented as
Example~\ref{ex:dijkstra} above, a more precise formalization of Dijkstra's
algorithm, an algorithm by Knuth \cite{DBLP:journals/cacm/Knuth66}, one by de
Bruijn \cite{DBLP:journals/cacm/Bruijn67}, and one by Eisenberg \& McGuire
\cite{DBLP:journals/cacm/EisenbergM72}. Additionally, we have modeled
Szymanski's algorithm for mutual exclusion \cite{Szymanski90} as well. For all
these algorithms we consider the property that they indeed provide mutual
exclusion of processes in the critical section. For most of these algorithms we
need to expand the topology of inspection programs in various ways. However,
all these expansions maintain that every transition involves at most a finite
amount of indices and that every re-ordering of these indices also yield a
transition in an instance. Inspecting the proof of Lemma~\ref{lem:looping} one
observes that these are the crucial observations for the stated result.
Consequently, Theorem~\ref{thm:looping-gen} generalizes well to all these
expansions. We present positive results for all these examples but for the
algorithm of Szymanski. The table in Figure~\ref{fig:heron-results} reports
data on the positive results. The first column shows which algorithm we
prove. The second column states how many seconds \heron{}
needs to compute the positive result. The third column reports the maximal $n$
for which $\PN(n)$ is instantiated during this computation. In the fourth
column we give the amount of traps we computed during this computation, and in
the fifth column how many abducted trap families we used. The sixth column
gives the maximal amount of looping indices that occur in traps for this
example and, finally, we give the longest time it took for a successful query
to the prover.

For Szymanski's algorithm for mutual exclusion our algorithm fails. This is
because already instances with $n = 2$ do not allow to prove the mutual
exclusion property via traps, 1BB-sets, and siphons. In fact, Szymanski's
algorithm posed already a negative result for the approach of
\cite{EsparzaLMMN14}. This means that instances of Szymanski's algorithm are
out of reach -- even when one additionally uses the marking equation for Petri
nets to over-approximate reachable markings.

\begin{figure}
  \caption{Experimental results for \heron{}.}
  \label{fig:heron-results}
  \begin{center}
    \scalebox{0.8}{
      \begin{tabular}{l||cccccc}
        Algorithm & time (s) & max. N & \# traps & \# abducted traps
        & \begin{tabular}{c}max.\\ \# indices\end{tabular}
        & \begin{tabular}{c}max. proving \\ time (s)\end{tabular}
          \\\hline
        Example~\ref{ex:dijkstra} & 11 & 8 & 36 & 2 & 2 & 1 \\
        Dijkstra's & 22 & 8 & 75 & 5 & 2 & 1 \\
        Knuth's & 194 & 8 & 160 & 7 & 2 & 1 \\
        de Bruijn's & 76 & 8 & 164 & 6 & 2 & 1 \\
        Eisenberg \& McGuire's & 1055 & 9 & 126 & 6 & 2 & 1\\
      \end{tabular}}
  \end{center}
\end{figure}

The data suggests that \heron{}, as \ostrich{}, synthesizes only a small amount
of necessary invariants. Moreover, these invariants are \enquote{local}; that
is, they involve at most $2$ looping indices. Consequently, the proofs that
\heron{} constructs are readable and concise. The drawback, however, is the
significant amount of time we need to construct and verify these proofs.
Surprisingly, the queries to the theorem prover, once all the necessary
invariants are synthesized, are actually very fast;  most
time is spent on instantiating and proving finite instances of $\PN$.

%% file: tex/diagram-ostrich.tex
\scalebox{0.8}{
\begin{tikzpicture}[font=\small]
  \node (spec) {$\tuple{\PN, \Safe}$};

  \node [below=of spec] (safetyCheck) {
    $\Safety_{\trapset}$};

  \draw [->] (spec) to node [right] {$\trapset = \emptyset$} (safetyCheck);

  \node [right=1 of safetyCheck] (successfulReturn) {
    return \enquote{$\Safe$ holds}
  };

  \draw [->] (safetyCheck) to node [above] {satisfiable} (successfulReturn);

  \node [anchor=center,left=7 of safetyCheck.center, text width=3.5cm] (instance) {
    Is the unsafe marking $\Mv$ in $\PN(n)$ reachable?
  };

  \draw [->] (safetyCheck) to node [above] {
    counter-example: $\tuple{n, \Mv}$
  } (instance);

  \node [below=of instance] (instanceDismiss) {
    $\WitTrap_\interp{M}(\bm{n}, \Xv)$
  };

  \node [below=of instanceDismiss] (failure) {
    return \enquote{unsure}
  };

  \draw [->] (instance) to (instanceDismiss);

  \node (abduct) at (instanceDismiss-|safetyCheck) {
    Abduct: $\Param_\interp{T}(\bm{n}, \Xv)$
  };

  \draw [->] (instanceDismiss) to node [above] {model: $\interp{T}$} (abduct);
  \draw [->] (abduct) to node [right] {add to $\trapset$} (safetyCheck);
  \draw [->] (instanceDismiss) to node [left] {invalid} (failure);
\end{tikzpicture}
}

%% file: tex/diagram-heron.tex
\scalebox{0.8}{
\begin{tikzpicture}[font=\small]
  \node (spec) {$\tuple{\PN, \Safe}$};

  \node [below=of spec] (safetyCheck) {
    $\Safety_{\trapset}$};

  \draw [->] (spec) to node [right] {$\trapset = \emptyset$} (safetyCheck);

  \node [right=1 of safetyCheck] (successfulReturn) {
    return \enquote{$\Safe$ holds}
  };

  \draw [->] (safetyCheck) to node [above] {valid} (successfulReturn);

  \node [anchor=center,left=7 of safetyCheck.center, text width=3.5cm] (instance) {
    Is an unsafe marking in $\PN(n)$ reachable?
  };

  \draw [->] (safetyCheck) to node [above] {
    timeout: increase $n$
  } (instance);

  \node [below=of instance] (instanceDismiss) {
    $\WitTrap_\interp{M}(\bm{n}, \Xv, \Uv)$
  };

  \node [below=of instanceDismiss] (failure) {
    return \enquote{unsure}
  };

  \draw [->] (instance) to (instanceDismiss);

  \node (abduct) at (instanceDismiss-|safetyCheck) {
    Abduct: $\Param_\interp{T}(\bm{n}, \Xv, \Uv)$
  };

  \draw [->] (instanceDismiss) to node [above] {model: $\interp{T}$} (abduct);
  \draw [->] (abduct) to node [right] {add to $\trapset$} (safetyCheck);
  \draw [->] (instanceDismiss) to node [left] {invalid} (failure);
\end{tikzpicture}
}

%% file: tex/conclusion.tex
\section{Conclusion.}
We have refined the approach to parameterized verification of systems with
regular architectures presented in \cite{BozgaEISW20}. Instead of encoding the
complete verification question into large, monolithic WS1S-formula, our
approach introduces a CEGAR loop which also outputs an explanation of why the
property holds in the form of a typically small set of parameterized invariants
(see Example~\ref{ex:diningphil}). The explanation helps to uncover false
positives, where the verification succeeds only because the system or the
specification are incorrectly encoded in WS1S. It has also helped to find a
subtle bug in the implementation of \cite{BozgaEISW20} which hid unnoticed in
the complexity of the monolithic formula.  Additionally, our incremental
approach requires to check smaller WS1S-formulas, which often decreases the
verification time (cp.~the verification of Dijkstra's mutual exclusion
algorithm \cite{BozgaEISW20} in 10s to currently 2s).

On the other hand, seeing the abstraction helps one understand the analyzed
system. For example, we include in \cite{OstrichSubmissionURL} a leader
election algorithm for which the parameterization techniques of \ostrich{} are
too coarse to establish the general safety property of having always at most
one leader. However, \ostrich{} succeeds to prove the special case that not
agents $0$ and $n-1$ can become leader at the same time. For this proof
\ostrich{} finds a family of siphons which hint to a general inductive
invariant of the system. Using the semi-automatic mode of \ostrich{} we can
then verify this inductive invariant and, as a result of this, the general
safety property.

We wanted to expand our methodology to models that represent actual
implementations of distributed algorithms more accurately. Therefore, we
discussed an expansion of the original approach that allows to model non-atomic
global checks. Although one forfeits the decidability of the logical embedding
-- a corner stone of the CEGAR loop in the original model -- we can adapt our
approach to capture these expanded models as well. The resulting algorithm
solves a set of non-trivial examples. Moreover, it maintains the desirable
property of synthesizing concise and readable invariants for the considered
examples.

\paragraph{Future work.}
Parameterized Petri nets rely on \ac{WS1S} to specify their transitions and their initial
configurations, and it is well known that the languages expressible in \ac{WS1S} are exactly the 
regular languages. This suggests that our techniques might be extended to 
the regular systems analyzed in regular model-checking \cite{Bouajjani00}. In this approach a finite automaton describes the language of initial configurations, and a length-preserving transducer describes the possible
transitions. We think our techniques can be used to algorithmically compute a regular over-approximation of the
set of reachable configurations.

The \heron{} tool checks reachability in 1-safe nets by means of an
incomplete method that tests if a marking satisfies all constraints induced
by the traps, siphons, and 1-BB sets of the net. This is closely related to the approach of \cite{EsparzaLMMN14}, which relies on traps, siphons, and the marking equation.
Replacing the marking equation by 1-BB sets leads to a less precise test, but one that can be completely implemented
on top of a SAT-solver and can be generalized to the parameterized case. We plan to study if the benchmarks of 
\cite{EsparzaLMMN14} can already be successfully verified using traps, siphons, and 1-BB sets, or a suitable generalization thereof. 

Our method is currently restricted to looping programs. We think that it can be extended to programs with nested loops.
We also plan to study stronger invariants allowing us to verify Szymanski's algorithm for mutual
exclusion, for which our technique is not yet strong enough.

\paragraph{Data Availability Statement and Acknowledgements.}
\small{
  This work has received funding from the European Research Council(ERC) under
  the European Union's Horizon 2020 research and innovation programme under
  grant agreement No~787367 (PaVeS).

\noindent The tool \ostrich{} and associated files are available at
\cite{OstrichSubmissionURL}. The current
version is maintained at \cite{OstrichRepositoryURL}. The tool \heron{} and
associated files are available at \cite{heron-artifact}. The current version is
maintained at \cite{heron-gitlab}.

\noindent This is an expanded version of \cite{EsparzaRW21}; that is,
Section~\ref{sec:gppn} was added. It also relies on results of \cite{ERW21b}.
We thank the anonymous reviewers of the original versions and this submission
for their helpful comments.}

%% file: tex/embedding.tex
\section{Constructing FO($\Safety_{\trapset}$) for looping programs.}
\label{appendix:FO}

\paragraph{An embedding of a linear order into FO.}
First, we want to capture the linear topology of our agents in an \ac{FO}
theory. To this end, we gradually introduce an appropriate \ac{FO} theory in
the following. Initially, consider a relation symbol $\leq$. It is
straightforward to give a sentence $\varphi_{\leq}$ that ensures that $\leq$ is
a discrete linear order with a minimal element. Then, we introduce two constant
symbols, $0$ and $N$. We make sure that $0$ is the minimal element w.r.t.
$\leq$; that is, we add $\lnot \exists \bm{x} ~.~ \bm{x} \neq 0 \land \bm{x}
\leq 0$ to our theory. Furthermore, it is standard to obtain the immediate
successor of some element $x$ w.r.t. $\leq$. To ease presentation, we use a
function symbol $\logicnext$ instead (which is consistent with the successor
function in \ac{WS1S}). In the following, we allow for constant symbols $1, 2,
3, \ldots$ in \ac{FO} formulas. For this, we use the convention that the
constant symbol $i$ corresponds to the value that we obtain when applying the
function $\logicnext$ exactly $i$ times to $0$. Moreover, we model addition of
a constant value to some variable similarly by applying $\logicnext$
appropriately often to the variable symbol. This total order gives us now
access to the linear identities of the agents.

\paragraph{Representing configurations in FO.}
As before, we want to capture the current configuration as models of a formula.
Before we describe how we do this, we inspect the considered topology in more
detail. In this way we can identify invariants which allow us to simplify the
embedding into \ac{FO}:
\begin{itemize}
  \item Note that every agent maintains a local copy of variables, each of
    which has a finite domain: as we described before, the set $\places$
    partitions into $\places_{0}, \ldots, \places_{k}$ such that $\set{i}
    \times \places_{j}$ for any $i \in [n]$ and $0 \leq j \leq k$ is a 1-BB set
    in every instance $\PN(n)$.
  \item Additionally, one can also deduce that $\set{\tuple{r, i, j} \mid j \in
    \left( [n] \cup \set{\bot} \right)}$ is a 1BB-set for each $i \in [n]$ and
    $r \in \relations$ in every instance $\PN(n)$ of a looping program $\PN$.
  \item Moreover, we are assured that every relation symbol $\loopingRelation
    \in \relations$ is tied to exactly one loop transition pattern. Similarly,
    there is a unique $\loopingState$ for this transition pattern. By close
    inspection of the semantics of loop transition patterns, it is immediate
    that $\tuple{\loopingRelation, i, \bot}$ is marked if and only if
    $\tuple{\loopingState, i}$ is not marked for all $i \in [n]$ for all
    $\PN(n)$.
\end{itemize}
Hence, we can restrict our analysis to markings that satisfy these constraints.
Consequently, we call markings of $\PN(n)$ \emph{viable} if $\sum_{p \in
\places_{j} \times \set{i}}M(p) = 1$ for all $i \in [n]$ and $0 \leq j \leq k$,
$\sum_{p \in \set{\loopingRelation} \times \set{i} \times \left( [n] \cup
\set{\bot} \right)}M(p) = 1$ for all $i \in [n]$ and $\loopingRelation \in
\relations$, and $M(\tuple{\loopingState, i}) + M(\tuple{\loopingRelation, i,
\bot}) = 1$ for any $\loopingState$ and $\loopingRelation$ that occur in the
same loop transition pattern. From now on, we refer to
$\loopingState^{\loopingRelation}$ for the uniquely identified state value that
occurs with $\loopingRelation$ in a loop transition pattern. Similarly, we use
$\loopingRelation^{\loopingState}$.

Before we used monadic variables $\Xv_{p}$ for each $p \in \places$ which
capture which places of $\places \times [n]$ are marked in the considered
instance. But by restricting our analysis to viable markings we can express the
current value of this variable as a function symbol instead. For this, fix some
$0 \leq j \leq k$ and let $\val_{1}, \ldots, \val_{m}$ be an enumeration of
$\places_{j}$. We introduce now a function symbol $\var_{j}$ and constant
symbols $\val_{1}, \ldots, \val_{m}$. Then, we can express with
$\var_{j}(\bm{x}) = \val_{\ell}$ that the agent with index $\bm{x}$ currently
sets its $j$-th variable to the value $\val_{\ell}$. It is straightforward to
restrict the domain of $\var_{j}$ only to these constant symbols for all
agents: $\forall \bm{x} ~.~ \bigvee_{1 \leq \ell \leq m} \var_{j}(\bm{x}) =
\val_{\ell}$. In this way we translate the topological restriction of viable
markings implicitly to our representation: we use a function symbol $\var_{j}$
instead of $\ell$ many monadic variables. In the following, we refer to
$\var_{\val}$ to $\var_{j}$ such that $\val \in \places_{j}$.

Similarly, we use that every agent executes at most one loop transition pattern
at a time by representing \emph{all} relations simultaneously by \emph{one
single} function symbol $f$: we need $f$ to map from agents to agents or some
value that represents $\bot$. Therefore, we add $\forall \bm{x} ~.~ f(\bm{x})
\leq N$ to the restricting theory; here $N$ will be used to represent $\bot$.

Consider now any viable marking $M$ in some instance $\PN(n)$. Then, $M$
induces a model $\mu$. Namely, we set the universe of $\mu$ to $\set{0, 1, 2,
3, \ldots}$, $\mu(0) = 0$, $\mu(N) = n$, and $\mu(\leq)$ to be the natural
order. For every $j$ we choose some arbitrary enumeration $\val_{1}, \ldots,
\val_{m}$ of $\places_{j}$ and set $\mu(\val_{k}) = k$ for $1 \leq k \leq m$.
Moreover, we set $\mu(\var_{j})$ to any function such that $\mu(\var_{j})(i) =
\val_{k}$ if and only if $M(\tuple{\val_{k}, i}) = 1$. The last definition uses
that $M$ is viable since there is exactly one such tuple for every $i \in [n]$.
Since $M$ is viable there is at most one $\loopingRelation$ for every $i \in
[n]$ such that $M(\tuple{\loopingRelation, i, j}) = 1$ for any $j \in [n]$. If
this is the case, let $\mu(f)(i) = j$. Otherwise, set $\mu(f)(i) = n$. In this
way, we obtain an interpretation $\mu$ for every viable marking $M$.

\paragraph{Traps in FO.}
The general idea of our approach is to obtain traps via
Theorem~\ref{thm:looping-gen}. Then, we use the induced invariants of these
traps to refine the \ac{FO} theory of interpretations which we consider. More
precisely, traps induce an abstraction of all reachable markings. We need to
restrict our theory in such a way that it still contains an interpretation that
represents any viable marking that satisfies the constraints of all found
traps. To this end, let us introduce an \ac{FO} formula which coincides with
the invariant the models of $\Param_\interp{T}(\bm{n}, \Xv, \Uv)$ from
Theorem~\ref{thm:looping-gen} induce:
\begin{equation*}
    \scalebox{0.75}{\mbox{\ensuremath{
  \begin{aligned}
    &\left(n \leq \bm{n} \; \land \; \exists \bm{y} ~.~ \bm{y} + (n - (i + 2))
    = \bm{n}\right) \land i+1 \leq \bm{y} \\
    \land &\bigvee_{\substack{j < i-1\\ p \in Q[j]}} \var_{p}(j) = p \\
    \lor &\exists \bm{j} ~.~ i-1 \leq \bm{j} \leq \bm{y} \land
    \bigvee_{p \in Q[i]} \var_{p}(\bm{j}) = p\\
    \lor &\bigvee_{\substack{i + 1 < j < n\\ p \in Q[j]}} \var_{p}(\bm{y} + (j
    - (i + 2))) = p\\
    \lor &\exists \bm{\ell}_{1}, \ldots, \bm{\ell}_{k}, \bm{u}_{1}, \ldots,
    \bm{u}_{m} ~.~ \bigwedge_{1 \leq v \leq k} \bm{\ell}_{v} = \ell_{v} \land
    \bigwedge_{1 \leq v \leq m} \bm{u}_{v} = \bm{y} + (u_{v} - (i + 2)) \\
    &\land \left(
    \begin{aligned}
      &\bigvee_{\loopingRelation \in \relations}\bigvee_{j < i-1}
      \left[
      \begin{aligned}
        &\bigvee_{\substack{1 \leq v \leq k \\ r(\ell_{v}) \in R[*, j]}}
        \var_{\loopingState^{\loopingRelation}}(\bm{\ell}_{v}) =
        \loopingState^{\loopingRelation} \land f(\bm{\ell}_{v}) = j\\
        \lor &\bigvee_{\substack{1 \leq v \leq m \\ r(u_{v}) \in R[*, i]}}
        \var_{\loopingState^{\loopingRelation}}(\bm{u}_{v}) =
        \loopingState^{\loopingRelation} \land f(\bm{u}_{v}) = j
      \end{aligned}
      \right] \\
      \lor & \exists \bm{j} ~.~ i-1 \leq \bm{j} \leq \bm{y} \land
      \bigvee_{\loopingRelation \in \relations}
      \left[
        \begin{aligned}
          &\bigvee_{\substack{1 \leq v \leq k \\ r(\ell_{v}) \in R[*, i]}}
          \var_{\loopingState^{\loopingRelation}}(\bm{\ell}_{v}) =
          \loopingState^{\loopingRelation}
          \land f(\bm{\ell}_{v}) = \bm{j} \\
          \lor & \bigvee_{\loopingRelation \in \relations}
          \bigvee_{\substack{1 \leq v \leq m \\ r(u_{v}) \in R[*, i]}}
          \var_{\loopingState^{\loopingRelation}}(\bm{u}_{v}) =
          \loopingState^{\loopingRelation} \land f(\bm{u}_{v}) = \bm{j}
        \end{aligned}
      \right]\\
      \lor &\bigvee_{\loopingRelation \in \relations}\bigvee_{i + 0 < j}
      \left[
      \begin{aligned}
        &\bigvee_{\substack{1 \leq v \leq k \\ r(\ell_{v}) \in R[*, j]}}
        \var_{\loopingState^{\loopingRelation}}(\bm{\ell}_{v}) =
        \loopingState^{\loopingRelation} \land f(\bm{\ell}_{v}) = \bm{y} + (j -
        (i + 2))\\
        \lor &\bigvee_{\substack{1 \leq v \leq m \\ r(u_{v}) \in R[*, j]}}
        \var_{\loopingState^{\loopingRelation}}(\bm{u}_{v}) =
        \loopingState^{\loopingRelation} \land
        f(\bm{u}_{v}) = \bm{y} + (j - (i + 2))\\
      \end{aligned}
      \right] \\
      \end{aligned}\right)
  \end{aligned}}}}
\end{equation*}
Note the conversion of relationset variables to the logical representation of a
single function symbol. This conversion is driven by the observation that
viable markings $M$ ensure that $M(\tuple{\loopingState, i}) +
M(\tuple{\loopingRelation, i, \bot}) = 1$ for $\loopingState$ and
$\loopingRelation$ occurring in the same loop transition pattern. To this end,
we use that $M(\tuple{\loopingState, i}) = 1$ for some $\loopingState$
necessarily implies $M(\tuple{\loopingState', i}) = 0$ for all other states
attached to some loop transition pattern since the state values form a
1BB-set for every agent. This, in turn, ensures that
$M(\tuple{\loopingRelation^{\loopingState'}, i, \bot}) = 1$ and, by the
appropriate 1BB-cover, $M(\tuple{\loopingRelation^{\loopingState'}, i, j}) = 0$
for all $j \in [n]$.